\documentclass[prx,twocolumn,showpacs,superscriptaddress,preprintnumbers,amssymb]{revtex4-2} 
\usepackage{graphicx}
\usepackage{amsthm}
\usepackage{thmtools}
\usepackage{url}
\usepackage{tikz}
\usepackage{relsize}
\usepackage{latexsym}
\usepackage{amssymb}
\usepackage{amsmath}
\usepackage{mathtools}
\usepackage{amsmath} 
\usepackage{amsfonts} 
\usepackage{upgreek}
\usepackage{bm}
\usepackage{multirow}
\usepackage{enumitem}
\usepackage{color}
\usepackage[colorlinks, citecolor=blue, linkcolor=blue]{hyperref}

\usepackage{tikz-cd}

\usepackage{calligra}
\DeclareMathAlphabet{\mathcalligra}{T1}{calligra}{m}{n}

\usepackage{physics}
\usepackage{tcolorbox}
\usepackage{manfnt}
\usepackage{lipsum}
\usepackage[ruled,linesnumbered]{algorithm2e}
\usepackage{tabularx}
\usepackage{braket}
\usepackage[normalem]{ulem} 
\usepackage{enumitem}

\newcommand{\cC}{\mathcal{C}}

\newcommand{\cA}{{\cal A}}
\newcommand{\cB}{ {\cal B} }

\newcommand{\cH}{ {\cal H} }

\newcommand{\cL}{ {\cal L} }

\newcommand{\cS}{ {\cal S} }

\newcommand{\cZ}{ {\cal Z} }

\newcommand{\secref}[1]{Sec.\,\ref{#1}}
\newcommand{\appref}[1]{App.\,\ref{#1}}
\newcommand{\eqnref}[1]{Eq.\,\eqref{#1}}

\newcommand{\figref}[1]{Fig.\,\ref{#1}}

\newcommand{\cN}{{\mathcal{N}}}

\newcommand{\Z}{\mathbb Z}
\newcommand{\RR}{\mathbb R}
\newcommand{\PP}{\mathbb P}
\newcommand{\EE}{\mathbb E}

\newcommand{\eps}{\varepsilon}

\newcommand{\dist}{{\rm dist}}

\theoremstyle{plain} 
\newtheorem{theorem}{Theorem}

\newtheoremstyle{axiomwithbreak} 
  {3pt}   
  {3pt}   
  {} 
  {}      
  {\bfseries} 
  {.}     
  {\newline }   
  {}      

\theoremstyle{axiomwithbreak}

\theoremstyle{definition}               

\newtheorem{definition}{Definition}

\newtheorem{proposition}{Proposition}

\newtheorem{assumption}{Assumption}

\makeatother

\newtheorem{lemma}{Lemma}

\makeatletter

\newif\ifappendixtoc
\appendixtocfalse

\let\oldappendix\appendix
\renewcommand{\appendix}{%
  \oldappendix
  \global\appendixtoctrue
}

\let\oldaddcontentsline\addcontentsline

\renewcommand{\addcontentsline}[3]{%
  \oldaddcontentsline{#1}{#2}{#3}%
  \ifappendixtoc
    \def\tempa{#1}%
    \def\tempb{toc}%
    \ifx\tempa\tempb
      \oldaddcontentsline{atoc}{#2}{#3}%
    \fi
  \fi
}

\newcommand{\appendixtableofcontents}{%
  \begingroup
  \setcounter{tocdepth}{2}
  \begin{center}
    \large \bfseries Appendix
  \end{center}
  \@starttoc{atoc}%
  \endgroup
}

\makeatother

\begin{document}

\title{Renormalization for Robust Memory without Synchronization}

\author{Shashwat Chakraborty}
\affiliation{Department of Physics and Institute for Condensed Matter Theory, University of Illinois Urbana-Champaign, Urbana, Illinois 61801, USA}

\author{Jong Yeon Lee}
\email[Contact Author:\,\,]{jongyeon@illinois.edu}
\affiliation{Department of Physics and Institute for Condensed Matter Theory, University of Illinois Urbana-Champaign, Urbana, Illinois 61801, USA}
\affiliation{Korea Institute for Advanced Study, Seoul 02455, South Korea}

\date{\today}
\begin{abstract} 
Local error correction can preserve classical information against noise, but protection under synchronized updates need not survive when sites evolve at independent random times. 
Establishing stability is particularly challenging when protection relies on the collective action of competing local rules.
Here we develop a spacetime renormalization framework that reveals how independently acting local updates can collectively stabilize memory, even when these competing local rules occur at unequal rates. 
We identify sufficient conditions under which the geometry of stable interfaces and suppression of defect propagation produce reliable erosion at large scales. Applying this framework, we prove stability of the asynchronous \textsf{R} squeezing code and introduce robust processes in two and three dimensions. In these systems, both encoded phases persist under sufficiently weak noise, even when the noise favors one phase. Our results provide a criterion for designing autonomous memories sustained by competing local stochastic dynamics.
\end{abstract} 
\maketitle

\section{Introduction}

Cellular automata (CA) such as Conway's Game of Life~\cite{Gardner1970Life} support patterns that persist indefinitely. However, such persistence alone does not establish robust memory: stored information must remain distinguishable when the idealized assumptions of the dynamics are relaxed. 
Here we consider two such departures: random local errors and the absence of a global clock that keeps neighboring updates in step.
The broader problem of robust memory has long been studied in the theory of probabilistic cellular automata (PCA)~\cite{GrinsteinJayaprakashHe1985,BennettGrinstein1985,LebowitzMaesSpeer1990,MairesseMarcovici2014}, and has attracted renewed attention through recent work on cellular-automaton based decoders for quantum error correction~\cite{Herold_2015,Herold_2017,breuckmann2016localdecoders2d4d,balasubramanian2025localautomaton2dtoric,Kubica_2019,Ray_2024,paletta2025highperformancelocaldecodersdefect,lake2025localactiveerrorcorrection,lake2025fastofflinedecodinglocal,paletta2026localdecodertoriccode}. 

In synchronous cellular automata, all sites update simultaneously at discrete time steps, as illustrated in~\figref{fig:temp}(a). 
For a broad class of monotone synchronous rules, Toom's stability theorem~\cite{Toom1974,Toom1980,Toom1995} guarantees the stability of a homogeneous configuration under sufficiently weak noise when the noiseless dynamics erase every finite droplet of defects---a property known as \emph{erosion}.
Subsequent work clarified and simplified the proof and developed contour methods and related structural criteria for synchronous PCA~\cite{GacsToomProof,GacsTorma2022,SwartSzaboToninelli2022, SwartSzaboToninelli2025}. From the viewpoint of statistical mechanics, Toom-type models have become canonical examples of dynamically stabilized ordered phases far from equilibrium~\cite{BennettGrinstein1985,FernandezToom2003}.

\begin{figure}[!t]
    \centering
    \includegraphics[width=0.9\linewidth]{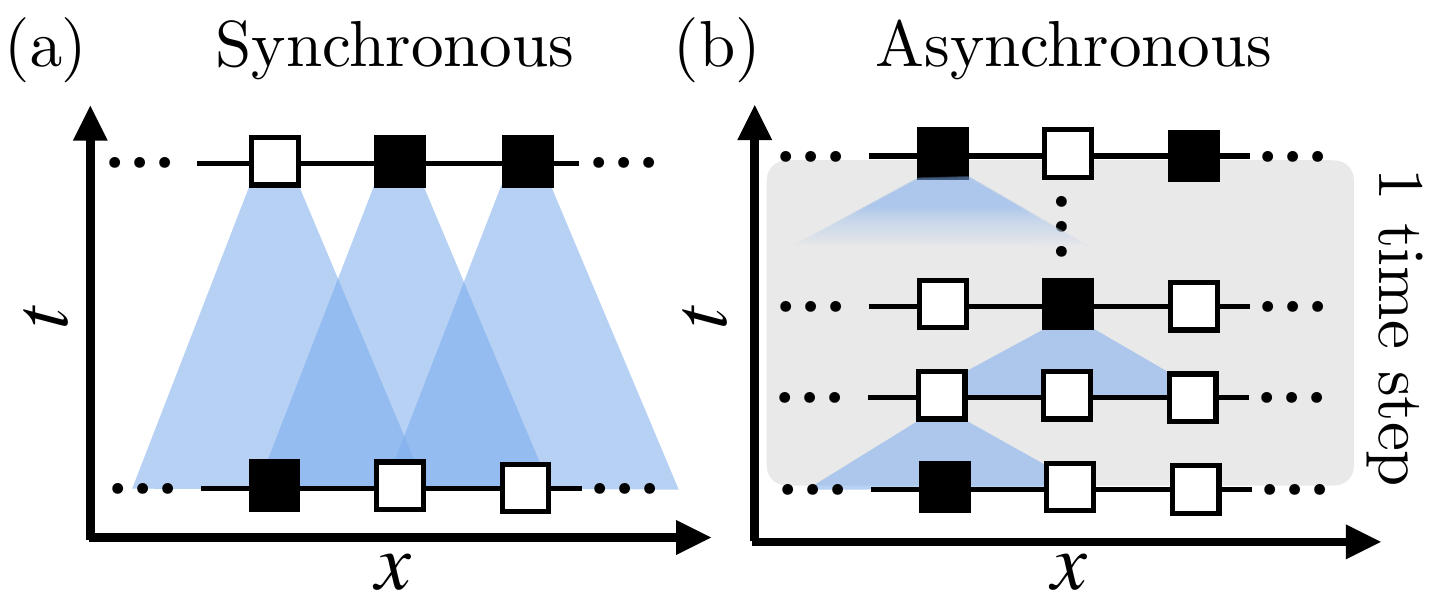}
    \caption{\textbf{Synchronous and asynchronous updates.} {\bf (a)} Under synchronous updates, all sites are updated simultaneously at each discrete time step. {\bf (b)} Under asynchronous updates on a system of $N$ sites, each time step is divided into $N$ substeps, and one randomly chosen site is updated at each substep. Thus, each site is updated once on average per time step, although the order and number of updates at individual sites are random. As $N\rightarrow\infty$, with substep duration $1/N$, the update clocks on any fixed finite collection of sites converge to independent rate-one Poisson processes.}
 \label{fig:temp}    
\end{figure}

\begin{figure*}[!t]
    \centering
    \includegraphics[width=0.95\linewidth]{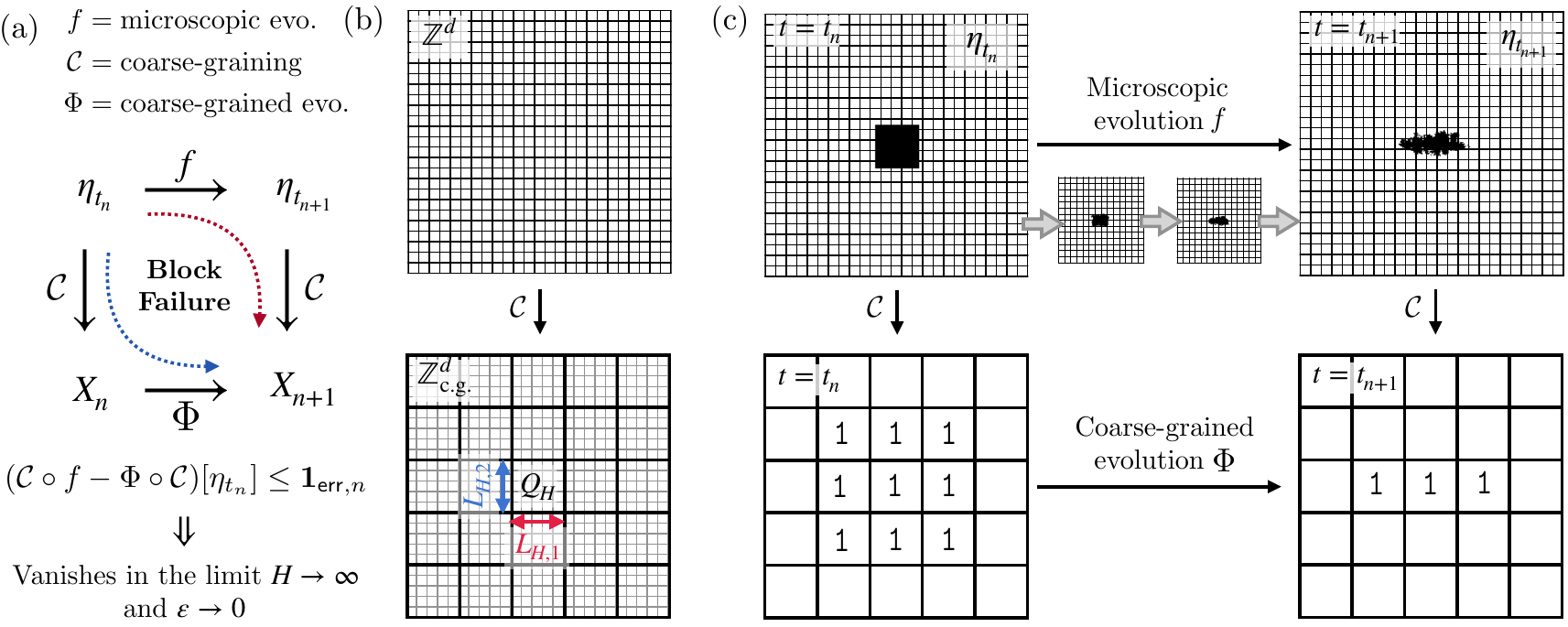}
    \caption{\textbf{Spacetime renormalization argument.} {\bf (a)} For a fixed microscopic update history and in the absence of ${\sf error}_{n,z}$, evolving a state and then coarse-graining is dominated by coarse-graining first and then applying $\Phi$, as stated in \eqnref{eq:domination}. {\bf (b)} The microscopic lattice is partitioned into congruent, possibly anisotropic boxes $Q_H$, which form the coarse-grained lattice. {\bf (c)} Illustration of the two paths in (a) for the $\mathsf{R}$ rule: black squares represent microscopic defects, while {\sf 1} denotes a defective coarse cell.
    The argument is completed by showing that the (one-sided) deviation between two paths (red and blue) is suppressed with the coarse-graining scale. Intuitively, violations of the domination $(\cC \circ f)[\eta_{t_n}] \leq (\Phi \circ \cC)[\eta_{t_n}]$ are rare.}
    \label{fig:coarse-graining}
\end{figure*}

Under asynchronous dynamics, each site updates at independent random times.
A fixed time interval therefore contains a random sequence of local updates, rather than a prescribed parallel update. 
Stability under synchronous dynamics need not survive this change in schedule. 
Building on the renormalization approach in Ref.~\cite{BramsonGray1991}, Gray established continuous-time stability for a class of processes with exclusive birth and death rates~\footnote{We say that a process has an exclusive birth-death structure if, for every site and every fixed configuration of its surroundings, a positive defect-birth rate forces the corresponding defect-death rate to vanish, and vice versa.} satisfying an additional geometric erosion condition~\cite{Gray1999}, including the continuous-time Toom North-East-Center (NEC) process. 
The key ingredient in Gray's proof is that evolving the exclusive birth–death dynamics asynchronously for time $T$ approximates a coarse-grained step of a corresponding synchronous eroder, with a deviation suppressed with $T$.
However, Gray's construction does not directly cover nonexclusive dynamics, in which competing stochastic updates can create and remove defects for the same surrounding configuration. 

The asynchronous versions of the recently introduced \emph{squeezing codes}~\cite{LakeRo2025} provide concrete examples of nonexclusive dynamics. 
The original two-dimensional rules eliminate minority defect droplets by squeezing them and are provably robust under synchronous updates.
Numerical studies of their asynchronous counterparts, however, revealed a dichotomy: some rules remain robust, whereas others fail to preserve memory. Even rules that remain robust in two dimensions can lose this property in their three-dimensional extensions.
These findings raise questions. What enables some squeezing rules to preserve memory without synchronization? More generally, what principles can guide the design of robust asynchronous memory dynamics?

In this work, we answer this question by developing a spacetime renormalization framework that establishes a sufficient stability criterion for monotone continuous-time dynamics without assuming microscopic exclusivity.
The central idea is to show that, on sufficiently large spatial and temporal scales, the stochastic dynamics remove defects at least as effectively as a deterministic eroder, apart from rare errors.
The geometry of stable interfaces provides the eroder, while probabilistic estimates control the errors in this comparison, allowing Toom-type stability arguments to be applied~\cite{Toom1980,GacsToomProof, SwartSzaboToninelli2022,GacsTorma2022}, see \figref{fig:coarse-graining}(a).

Using this framework, we provide the first rigorous proof of stability for the continuous-time {\sf R} squeezing process~\cite{LakeRo2025}, controlling its erosion via comparison with a subcritical one-dimensional contact process. Notably, this proof accommodates a range of unequal rates for the competing microscopic rules. Furthermore, we introduce a \emph{cross-gated} (CG) process, proving that the model in two and three dimensions retains two distinct stable phases in the infinite-volume limit under sufficiently weak noise, even with imbalanced update rates; this is contrasted with the earlier proposed {\sf R} rule which is unstable in 3D. Using Monte Carlo simulations, we study phase diagrams for the 2D CG and {\sf R} processes under unequal rates, confirming their robust behavior numerically.

The remainder of this paper is organized as follows. In \secref{sec:sync_vs_async} we introduce asynchronous dynamics and contrast them with synchronous dynamics, while \secref{sec:operator_spectrum} discusses the stability of asynchronous dynamics from the perspective of
local Markov generators. In \secref{sec:stab_criterion}, we develop the geometric framework and establish our general stability criterion through spacetime renormalization. In \secref{sec:applications}, we apply our stability criterion to several two- and three-dimensional continuous-time processes. Finally, we conclude with a discussion in \secref{sec:disc_and_outlook}. 

\section{Synchrony and Asynchrony} 
 \label{sec:sync_vs_async}
A cellular automaton (CA) on a lattice $\Lambda\subseteq\mathbb Z^d$ is specified by an onsite state space $\mathcal S$, a finite neighborhood $\mathcal N$, and a local update map $f:\mathcal S^{\mathcal N}\to\mathcal S$.
Let $\eta_{t}(x)$ denote the state at site $x$ and $t$.
In the standard \emph{synchronous} setting, time is discrete and the same local map is applied simultaneously at every site,
\begin{align}
\eta_{t+1}(x) = f\big( \{ \eta_t(x+\delta) \}_{\delta\in\mathcal N} \big).
\end{align}
More generally, one may periodically alternate among a finite collection of local maps, giving a Floquet CA. Parallel updates organize the resulting spacetime dynamics into discrete time layers, as illustrated in ~\figref{fig:temp}(a).

In \emph{asynchronous} dynamics, consider a finite collection of deterministic update types $a\in\mathcal A$, specified by neighborhoods $\mathcal N_a$ and maps $f_a:\mathcal S^{\mathcal N_a}\to\mathcal S$.
For each pair $(x,a)$, updates occur according to an independent Poisson clock of rate $\lambda_a$~\cite{Liggett1985}, as shown in~\figref{fig:temp}(b); when the clock rings, only the state at $x$ is replaced by $f_a((\eta(x+\delta))_{\delta\in\mathcal N_a})$.
The local maps are deterministic; randomness enters through the times and locations of their application. This defines a continuous-time Markov process even in the absence of added noise.
Changing the update schedule can alter how the local rules act together, imposing \emph{temporal disorders} on updates. In particular, erosion under synchronous updates does not by itself guarantee stability of the asynchronous process.

The squeezing codes of Ref.~\cite{LakeRo2025} provide a concrete setting in which to examine this dependence on the update schedule.
Each site carries a Boolean variable $\eta(x)\in\{{\tt 0},{\tt 1}\}$.
Their synchronous dynamics form a deterministic period-2 CA, in which OR and AND maps alternate globally,
\begin{align}
\eta_{t+1}(x)
&=
\begin{cases}
\displaystyle\bigvee_{\delta\in\mathcal N_{\vee}}\eta_t ({x+\delta}),
&t\in2\mathbb Z,\\
\displaystyle\bigwedge_{\delta\in\mathcal N_{\wedge}}\eta_t ({x+\delta}),
&t\in2\mathbb Z+1.
\end{cases}
\end{align}
For the $\mathsf R$ squeezing automaton,
\begin{align}
\mathcal N_{\vee}
&=
\{-e_1,0,e_1\}, \quad \,
\mathcal N_{\wedge} =
\{-e_2,0,e_2\},
\label{eq:R_nbd}
\end{align}
where $e_1$ ($e_2$) is the unit vector along the $x$ ($y$) direction. 
For a droplet of 1s in a background of 0s, the horizontal OR updates spread defects laterally, while the vertical AND updates remove defects along the upper and lower boundaries. Their alternating action produces the characteristic squeezing dynamics.

In continuous time, the global period-2 schedule is removed: OR and AND updates instead occur independently at each site according to Poisson clocks.
If a horizontal neighbor is in state {\tt 1} and a vertical neighbor is in state {\tt 0}, the OR map permits a ${\tt 0}\to {\tt 1}$ transition, while the AND map permits a ${\tt 1}\to {\tt 0}$ transition for the same surrounding configuration. The asynchronous {\sf R} process therefore has nonexclusive birth and death rates.

\section{Operator-spectrum perspective} 
\label{sec:operator_spectrum} 

Before developing our infinite-volume stability criterion, we offer a complementary perspective on memory retention through the relaxation of finite systems. This discussion is independent of the subsequent stability argument; readers may proceed directly to Sec. IV for the criterion and its proof.

On a finite lattice $\Lambda$, let $p_t$ denote the probability distribution over configurations, represented as a column vector.
For the Poisson updates introduced above, one can write the evolution as $\partial_t p_t = \cL p_t$, where $\cL$ is a Markov generator written as
\begin{align}
\cL  &= \sum_{x\in\Lambda}\sum_{a\in\mathcal A}\lambda_a(P_{x,a}-I)+\cL_{\mathrm{noise}}. \label{eq:ops:generator}
\end{align}
Here $P_{x,a}$ applies the update rule $a$ at site $x$; $\lambda_a$ is its rate, and $\cL_{\mathrm{noise}}$ contains the local error processes.
Thus, $H  \equiv -\cL $ is a local, generally non-Hermitian many-body operator, formally analogous to a Hamiltonian governing imaginary-time evolution~\cite{AlcarazEtAl1994}.
For irreducible noisy dynamics, the stationary distribution $\pi_{\mathrm{ss}}$ is the unique zero mode of $H $, and convergence to it eventually erases the stored information.
Long memory lifetimes, therefore, require slow global relaxation.

For reversible dynamics, detailed balance gives the generator a Hermitian representation.
Let $D$ be the diagonal matrix with entries $D_{\eta\eta}=\pi_{\mathrm{ss}}(\eta)$.
Similarity transformation by $D^{1/2}$ on the Markov generator $\cL$ gives
\begin{align} \label{eq:dtbalance}
H_{\mathrm{rev}} &= -D^{-1/2}\cL  D^{1/2},
\end{align}
which is real, symmetric, and positive semidefinite.
Its excitation energies are the relaxation rates of the Markov process~\cite{Henley2004}, with the smallest positive energy defining the relaxation gap.

Ising Glauber dynamics, familiar from single-spin Monte Carlo sampling, satisfies detailed balance and therefore admits the Hermitian representation above.
At low temperature, its two zero-field equilibrium phases are related by spin-flip symmetry.
In a finite system, rare transitions between these phases produce a small relaxation gap.
A magnetic field $h$ explicitly breaks this symmetry and selects a single equilibrium phase.
At fixed $h>0$, the relaxation gap on finite cubes with free boundary conditions remains bounded away from zero as their size increases, so the inverse-gap relaxation time is $O(1)$ in system size~\cite{Yoshida1998}.
This degeneracy lifting fits the familiar many-body picture of a phase selection by explicit symmetry breaking, as illustrated schematically in~\figref{fig:spectrum}(a).

\begin{figure}[!t]
    \centering
    \includegraphics[width=1\linewidth]{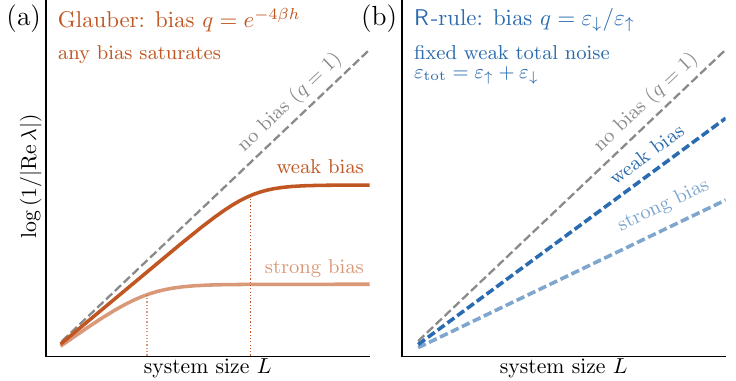}
    \caption{\textbf{Schematic spectral behavior.}
    Schematic dependence of the inverse relaxation gap on system size under a bias $q=\varepsilon_\downarrow/\varepsilon_\uparrow$, where $\varepsilon_\uparrow$ ($\varepsilon_\downarrow$) is the error rate attacking the disfavored (favored) phase.
    {\bf (a)} {\bf Low-temperature Ising}. Here the bias is applied via an external field $h$. At $h=0$, the spectral gap decreases exponentially with the linear system size.
    At any fixed positive field, the gap remains bounded away from zero uniformly in system size.
    {\bf (b) ${\sf R}$ process.} 
    At a fixed, sufficiently small noise strength $\eps_\uparrow +\eps_{\downarrow}$, both phases remain stable under arbitrary bias $q$, exhibiting a memory life-time that grows exponentially in system size~\cite{LakeRo2025}. The dashed curves schematically illustrate the conjectured exponential closing of the relaxation gap in $L$.
    }
    \label{fig:spectrum}
\end{figure}

In contrast, the asynchronous $\mathsf R$ squeezing process with nonzero noise in both directions violates detailed balance even at zero bias.
Its generator, therefore, remains non-Hermitian under the transformation above.
Like Toom's process, it supports two stable phases in infinite volume under sufficiently weak noise, even when a bias breaks the symmetry exchanging them.
This coexistence over a range of noise strengths and biases is known as \emph{generic bistability}~\cite{BennettGrinstein1985,MunozEtAl2005}.
Numerical studies find memory lifetimes consistent with exponential growth in the linear system size $L$, under both unbiased and maximally biased noise~\cite{LakeRo2025}.
If the slow dynamics are governed by transitions between the two encoded phases, exponentially long lifetimes in both phases would imply an exponentially small relaxation gap, as illustrated schematically in~\figref{fig:spectrum}(b).
Here, explicit symmetry breaking leaves the unfavored phase (meta)stable, even though it selects a unique phase in the equilibrium Ising model. 
Whether this robustness can be traced to structural features of the non-normal generator remains an intriguing open question.

Without detailed balance, the transformed generator can be non-normal, so its eigenmodes need not be orthogonal.
A spectral explanation of this robustness therefore requires control of both eigenvalues and eigenmodes as the system size grows.
For an irreducible generator $\cL$, define the relaxation gap as $g=\min_{\lambda\in\operatorname{spec}(\cL)\setminus\{0\}}[-\operatorname{Re}\lambda]$.
For a diagonalizable generator and any initial probability distribution $p_0$, the spectral decomposition gives
\begin{align}
\bigl\|e^{t\cL }p_0-\pi_{\mathrm{ss}}\bigr\|_{\mathrm{TV}} &\le C e^{-g  t}, \label{eq:ops:semigroup}
\end{align}
where $\|\cdot\|_{\mathrm{TV}}$ denotes total-variation distance and the time-independent prefactor $C$ depends on the spectral projectors.
The gap sets the slowest asymptotic decay rate, but relating its inverse to a memory lifetime requires control of this prefactor.
Nonorthogonality can make $C$ grow rapidly with system size, and explicit examples exhibit slow relaxation even when the gap stays open~\cite{MoriShirai2020,HagaEtAl2021}.
The same eigenmode structure also complicates perturbation theory: small overlaps between normalized left and right eigenvectors can amplify the shifts of eigenvalues~\cite{GreenbaumLiOverton2020}.
Controlling an expansion in the noise rate therefore requires bounds on this sensitivity that remain valid as the system size grows.

A separate obstacle to perturbation theory arises even in the reversible Ising example.
Although the field selects one phase, decay of the disfavored phase requires nucleation of a sufficiently large droplet of the favored phase.
For the two-dimensional system at a fixed temperature below the critical point, the relaxation time from the unfavored state satisfies $\tau(h) \sim e^{ C/h}$~\cite{SchonmannShlosman1998}, where $C$ is temperature dependent positive constant. The essential singularity in this expression makes the inverse metastable timescale smaller than every power of $h$, so no finite-order expansion in the field captures it.

Therefore, we take an approach different from spectral analysis of Markov generators.
Through spacetime renormalization, we directly ask whether locally created defects are eroded or can develop into a macroscopic region of the opposite phase.

\section{Stability criterion}\label{sec:stab_criterion}

In this section, we establish a stability criterion for monotone dynamics in continuous time.
We work on the configuration space $\Omega=\{{\tt 0},{\tt 1}\}^{\mathbb{Z}^d}$, where $d$ is the spatial dimension. We write $\eta_t\in\Omega$ for the configuration at time $t$. The homogeneous configurations are denoted by $\bar{\tt 0}$ and $\bar{\tt 1}$. We order configurations sitewise: for $\eta,\xi\in\Omega$, we write $\eta\leq\xi$ if $\eta(x)\leq\xi(x)$ for every $x\in\mathbb{Z}^d$. In the following discussion, we formulate the stability criterion for the $\bar{\tt 0}$ homogeneous configuration. Analogous results can be obtained separately for the opposite homogeneous configuration. 

Throughout, we consider dynamics that are finite range and translationally invariant, with a finite collection of update types $a\in\mathcal A$. 
At each site, updates of type $a$ occur according to an independent Poisson clock of rate $\lambda_a<\infty$, which is specified by a monotone local map $f_a$ satisfying $f_a(\bar{\tt 0})={\tt 0}$~\footnote{In this work, since we are primarily concerned with bistability, $f_a(\bar{\tt 1}) = {\tt 1}$ must hold as well.}. 
Thus, $\bar{\tt 0}$ is absorbing in the absence of noise. Monotonicity makes the process attractive, i.e., using the same realization of the Poisson clocks, $\eta_0\leq\xi_0$ implies $\eta_t\leq\xi_t$ for all $t\geq0$. 
This order-preserving coupling of processes underlies several arguments made below. 
Finite range and bounded update rates ensure that the infinite-volume process is well defined.
They also yield a probabilistic finite-speed bound: causal influence over a distance $r$ within time $t$ is exponentially unlikely when $r$ exceeds a fixed multiple of $t$; see Lemma~\ref{lem:speed} of Appendix~\ref{sec:app_A}.

We introduce noise through independent rate-$\eps$ Poisson clocks at every site, each of which forces the local state to ${\tt 1}$. These events are hostile to the $\bar{\tt 0}$ configuration. We denote by $\mathbb{P}_{\bar{\tt 0}}^{(\eps)}$ the law of the resulting process started from $\bar{\tt 0}$. 
To establish stability of ${\tt 0}$, it suffices to consider these one-sided errors; additional transitions ${\tt 1}\to{\tt 0}$ can only decrease the defect configuration and hence improve stability.
We say that the configuration $\bar{\tt 0}$ is stable if
\begin{equation}
\lim_{\eps\to0^+}\sup_{t\geq0}
\mathbb{P}_{\bar{\tt 0}}^{(\eps)}
\bigl(\eta_t(x)={\tt 1}\bigr)=0
\label{eq:stability_maintext}
\end{equation}
for every $x\in\mathbb{Z}^d$. By translation invariance, this probability is independent of $x$.

Finally, suppose that the noiseless dynamics are invariant under $\mathcal X\circ g$, where $\mathcal X$ is the global spin flip and $g$ is a lattice symmetry.
This transformation acts as $[(\mathcal X\circ g)\eta](x)=1-\eta(g^{-1}x)$ and exchanges the two homogeneous configurations and their corresponding hostile errors.
Stability of $\bar{\tt 0}$ against ${\tt 0}\to{\tt 1}$ errors therefore implies stability of $\bar{\tt 1}$ against ${\tt 1}\to{\tt 0}$ errors.
Applying the monotone comparison to each configuration then gives stability of both configurations under sufficiently weak noise in both directions, even when the two error rates differ.
This bistability allows robust encoding of a classical bit.

\subsection{Coarse-graining}

The basic idea on which the stability criterion relies is summarized in Fig.~\ref{fig:coarse-graining}(a).
We seek a spacetime coarse-graining under which the asynchronous microscopic dynamics reduce to a noisy discrete-time eroder. 
Real-space coarse-graining has a long history in cellular automata and interacting particle systems, including renormalization schemes for probabilistic cellular automata~\cite{Vespignani1997,deOliveiraSatulovsky1997} and block coarse-graining of deterministic and probabilistic cellular automata~\cite{IsraeliGoldenfeld2004,IsraeliGoldenfeld2006,EdlundJacobi2010,WeaverPrugelBennett2015}. 
Our construction differs in that we coarse-grain continuous-time dynamics jointly in space and time, with the aim of obtaining a coarse evolution dominated by a deterministic eroder, which is closer in spirit to the block-renormalization methods developed for interacting particle systems~\cite{DurrettGriffeath1983,BramsonGray1991}.

The construction has two ingredients.
Geometrically, we first identify the orientations of clean half-spaces that are preserved by every microscopic update rule. The geometry of these protected interfaces guides our choice of a coarse-graining map $\cC$ and the construction of a deterministic coarse eroder $\Phi$.
Dynamically, we assign a physical duration $T_H$ to each coarse step.
The central task is to choose the block dimensions and coarse time step $T_H$ so that, whenever $\Phi$ prescribes a clean target cell, the noiseless microscopic dynamics also leave that cell clean after time $T_H$, with probability approaching one as $H$ grows.

We record an effective error at a coarse spacetime block whenever the noiseless block comparison described above fails, or a hostile error occurs in the spacetime region involved in the block comparison (more details in Sec.~\ref{sec:main-block-failure}).
We collect the coarse spacetime points at which these errors occur in a random \emph{error set} $E$, following the definition in Ref.~\cite{GacsTorma2022}. 
For each coarse spacetime point $v$, let ${\sf error}_v$ denote the event that $v$ belongs to $E$, i.e., the event that an effective error is recorded at $v$.
We call the corresponding collection of binary indicators $\{\mathbf{1}_{{\sf err},v}\}_v$ the effective coarse \emph{error field}; as $E$ is random, the field is random as well. 
We then control the correlations in the error field to establish the sparsity needed to apply Toom's theorem: 
\begin{lemma}[Toom's theorem]
\label{lem:main-toom}
    Consider monotone synchronous dynamics $\Phi$ that erodes every finite defect set in a finite number of steps. Consider an error field $\mathbf{1}_{{\sf err}}$ that is sparse with parameter $q$, i.e., for any finite set of spacetime points $A$, the following holds:
    \begin{equation}
        \PP\left(\text{error is recorded for every }v\in A\right) \leq q^{|A|}.
    \end{equation}
    Then there exists $q_c > 0$ such that starting from $Y_0 = \bar{\tt 0}$, the noisy dynamics $Y_{n+1} = \Phi[Y_n] \lor \mathbf{1}_{{\sf err},n}$ satisfy
\begin{equation}
    \forall q \leq q_c, \quad \sup_{n \geq 0}\, \PP(Y_n(z) = {\tt 1}) \leq \rho_{\Phi}(q)
\end{equation}
for any site $z$, where $\rho_{\Phi}(q)$ is a function of the noise parameter that vanishes as $q \to 0$~\cite{Toom1974, Toom1980,GacsTorma2022}. 
\end{lemma}

To implement the general strategy outlined above, we first identify the geometric conditions under which the coarse-grained evolution $\Phi$ can be chosen to be an eroder.
For an update type $a\in\mathcal A$ with neighborhood $\mathcal N_a$, we call a subset $Z\subseteq\mathcal N_a$ a \emph{zero set} if fixing all spins in $Z$ to ${\tt 0}$ forces the output of the update to be ${\tt 0}$, irrespective of the remaining inputs:
\begin{align}
\eta(\delta)={\tt 0}
\text{ for every }\delta\in Z
\,\,\Longrightarrow\,\,
f_a(\eta)={\tt 0}.
\label{eq:zero_set}
\end{align}
Thus, a zero set specifies local input patterns sufficient to force a clean output.
For example, for Toom's NEC majority rule on $\mathcal N=\{0,e_1,e_2\}$, any two of the three sites form a zero set.

Physically, the zero sets determine which interfaces are stable under the noiseless dynamics.
Consider a configuration that is completely clean on one side of a hyperplane and defective on the other.
We call a nonzero linear form $\ell\in(\Z^d)^*$ a \emph{stable direction} if a clean half-space $\cH$ normal to $\ell$ (with $\ell$ pointing into $\cH$) cannot be invaded by defects.
Equivalently, defining
\begin{align}
\cN_a^{+}(\ell) = \{\delta\in\cN_a:\ell(\delta)\geq0\},
\end{align}
we require $\cN_a^{+}(\ell)$ to be a zero set for every update type $a\in\cA$.
Indeed, for a site on the clean side $\cH$ of the interface, all inputs in $\cN_a^{+}(\ell)$ remain clean; the zero-set condition, therefore, forces every update at that site to output ${\tt 0}$.
Thus, a stable direction defines a half-space boundary across which defects cannot propagate, as illustrated in~\figref{fig:geometric_conditions}(a).

\begin{figure}[!t]
    \centering
    \includegraphics[width=0.95\linewidth]{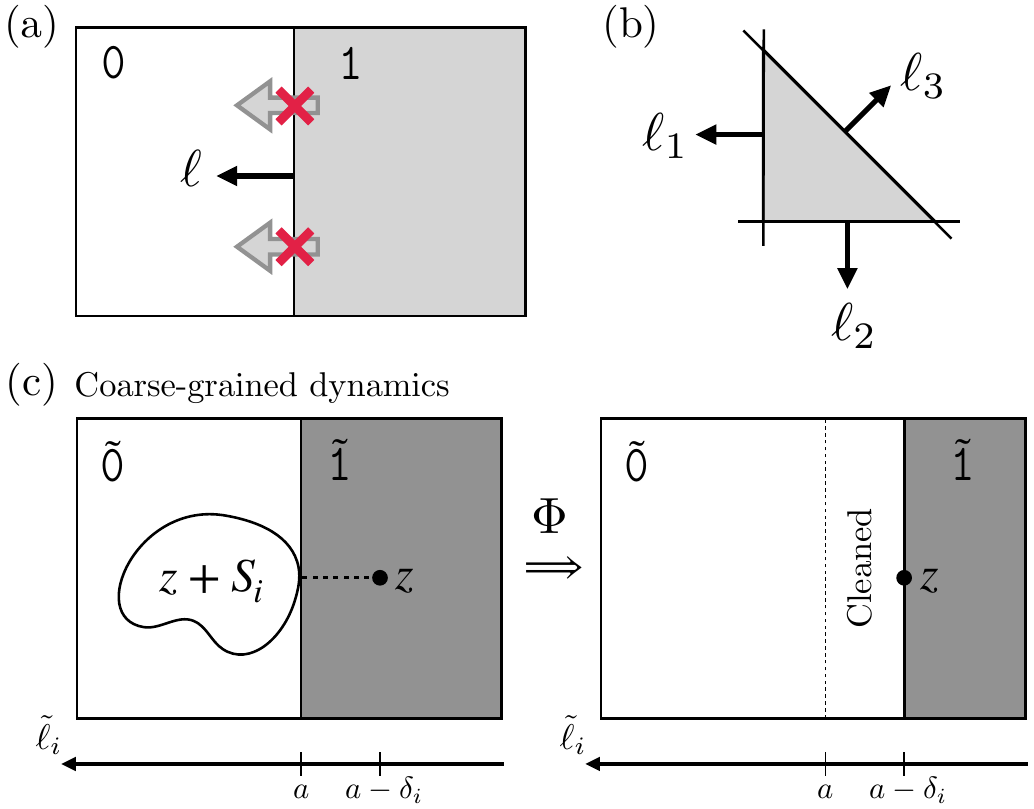}
    \caption{\textbf{(a,b) Geometric conditions for stability.} (a) illustrates that a direction $\ell$ is stable if a defect-free half-space cannot be invaded by defects under the noiseless dynamics. (b) shows a balanced family of stable directions $\ell_1 = -e^{\ast}_1$, $\ell_2 = -e^{\ast}_2$, and $\ell_3 =e^{\ast}_1+e^{\ast}_2$ with positive weights $\alpha_i = 1$ for $i \in \{1,2,3\}$, such that $\sum_i\alpha_i\ell_i=0$. 
    \textbf{(c) Coarse-grained eroder.} Initially, the half-space $\widetilde{\ell}_i(y)\geq a$ is clean (white).
    The condition $\widetilde{\ell}_i(s)\geq\delta_i$ for all $s\in S_i$ ensures that any target cell $z$ with $\widetilde{\ell}_i(z)\geq a-\delta_i$ has its entire predecessor set $z+S_i$ inside this clean half-space.
    The predecessor set is represented by the outlined region.
    One application of $\Phi$ therefore clears all such target cells, extending the guaranteed clean region to $\widetilde{\ell}_i(y)\geq a-\delta_i$.
    The dashed line marks the original interface; the remaining gray region may contain defects.
    }
    \label{fig:geometric_conditions}
\end{figure}

A single stable interface protects the clean half-space but need not make a droplet shrink.
At the coarse scale, a stable interface may remain stationary or move into the defective region as the clean region expands. We call the corresponding directions \emph{static} and \emph{advancing}, respectively.
To obtain erosion, we seek several stable interfaces that confine the droplet between them, so that advancing one interface without the others retreating squeezes the droplet.
We express this geometric requirement through a \emph{balanced} family of stable directions $\{\ell_1,\ldots,\ell_m\}$: there exist positive weights $\alpha_i>0$ such that $\sum_{i=1}^{m}\alpha_i\ell_i=0$, see \figref{fig:geometric_conditions}(b); intuitively, this allows the stable interfaces to confine defects between them in at least one spatial direction.
If at least one of these interfaces also advances, the resulting squeezing provides the geometric ingredient needed to construct the coarse-grained eroder $\Phi$.

With these geometric ingredients, we define the coarse-graining map $\cC$, illustrated in \figref{fig:coarse-graining}(b).
For a coarse length scale $H$, we partition the microscopic lattice into congruent, possibly anisotropic blocks $Q_H(z)$ indexed by coarse lattice sites $z\in\Z^d_{\rm cg}$.
Each block has side length $L_{H,j}$ (functions of $H$) along the $j$-th axis.
A coarse cell is declared defective if it contains at least one microscopic defect:
\begin{align}
\cC[\eta](z) &= \begin{cases} {\tt 1}, & \text{if $\eta(x)={\tt 1}$ for some $x\in Q_H(z)$},\\
{\tt 0}, & \text{otherwise}. \end{cases} \label{eq:coarse_graining}
\end{align}
We also introduce an $H$-dependent coarse time scale $T_H$ and sample the microscopic dynamics at times $t_n=nT_H$, writing $X_n=\cC[\eta_{t_n}]$.
The block dimensions $L_{H,j}$ and the coarse time $T_H$ are chosen to match the microscopic dynamics.
In particular, anisotropic blocks may be required when defect propagation occurs on parametrically different scales along different spatial directions.

We now use the balanced family of stable directions to construct the coarse-grained rule $\Phi$.
For each direction $\ell_i$, we choose a finite set of coarse displacements $S_i\subset\Z^d_{\rm cg}$, independently of the coarse-graining scale $H$.
For a target cell $z$, the translated set $z+S_i$ is its predecessor region associated with direction $i$.
The coarse rule declares $z$ clean exactly when all cells in at least one of these predecessor regions are clean:
\begin{align}
\Phi[X](z)=\bigwedge_{i=1}^{m}\left(\bigvee_{s\in S_i}X(z+s)\right). \label{eq:Phi_maintext}
\end{align}
The inner OR is zero precisely when every cell in $z+S_i$ is clean.
The outer AND therefore makes the output zero as soon as this condition holds for any one direction.

We choose the predecessor regions so that $\Phi$ preserves or advances the corresponding clean interfaces, see \figref{fig:geometric_conditions}(c).
Let $\widetilde{\ell}_i$ denote $\ell_i$ expressed in coarse coordinates, with larger values of $\widetilde{\ell}_i$ corresponding to the clean side.
This common change of coordinates preserves the balance relation $\sum_i\alpha_i\widetilde{\ell}_i=0$.
For each direction, we choose an offset $\delta_i\geq0$ and require
\begin{align}
\widetilde{\ell}_i(s)\geq\delta_i \qquad \text{for every }s\in S_i.
\end{align}
To understand this condition, suppose that all coarse cells on one side of the plane $\widetilde{\ell}_i(y)=a$, namely those satisfying $\widetilde{\ell}_i(y)\geq a$, are clean.
For any target cell $z$ satisfying $\widetilde{\ell}_i(z)\geq a-\delta_i$, every predecessor cell obeys $\widetilde{\ell}_i(z+s)=\widetilde{\ell}_i(z)+\widetilde{\ell}_i(s)\geq a$ and is therefore clean.
The rule $\Phi$ consequently makes the entire half-space $\widetilde{\ell}_i(z)\geq a-\delta_i$ clean after one coarse update.
Thus, $\delta_i=0$ guarantees preservation of the clean half-space, while $\delta_i>0$ guarantees its expansion into the defective region, toward smaller values of $\widetilde{\ell}_i$.

For a finite defect configuration, we can place a clean half-space beyond the furthest defect in each selected direction.
Under repeated application of $\Phi$, these half-space bounds are preserved or advance inward.
Together with balance, a positive net advance $\sum_i\alpha_i\delta_i>0$ forces the remaining defect region to disappear after finitely many coarse steps.
Thus, $\Phi$ is an eroder, as proved in \appref{sec:sm-coarse-eroder}. 
The block estimates below determine whether the microscopic dynamics reproduce this coarse-grained advance with high probability.

\subsection{Block failure}
\label{sec:main-block-failure}
The remaining dynamical question is whether, over a duration of $T_H$, the microscopic evolution is dominated by the coarse-grained action prescribed by $\Phi$ up to a rare failure event. 
Directly bounding the probability of such failures is challenging. We instead analyze a simpler but more adversarial process whose failure probability can be bounded tractably. Monotonicity then allows us to transfer this upper bound to the original dynamics.

To this end, we choose a finite region formed by coarse cells containing the core and its predecessor regions, which we call the \emph{buffer}, illustrated in \figref{fig:block-failure} as the region enclosed by the thick black line. 
We run an auxiliary noiseless microscopic process inside this buffer, with all exterior sites frozen to ${\tt 1}$.
By monotonicity, the auxiliary process dominates the original noiseless dynamics inside the buffer when both use the same update marks, provided the original initial configuration is no more defective there. 
Freezing the exterior to ${\tt 1}$ imposes the most defective boundary condition, so the comparison requires no assumption on the configuration outside the buffer in the original system.
We now specify the predecessor regions and the buffer. Fix a coarse cell $z$. For each stable direction $\ell_i$, we denote the set of microscopic sites in the predecessor region associated with that direction by 
\begin{equation}
Q_H(z+S_i) := \bigcup_{s\in S_i}Q_H(z+s).
\end{equation}
We write $\cB_H(z)$ for the set of microscopic sites around $z$ corresponding to the buffer.

\begin{figure}[!t]
    \centering
    \includegraphics[width=1\linewidth]{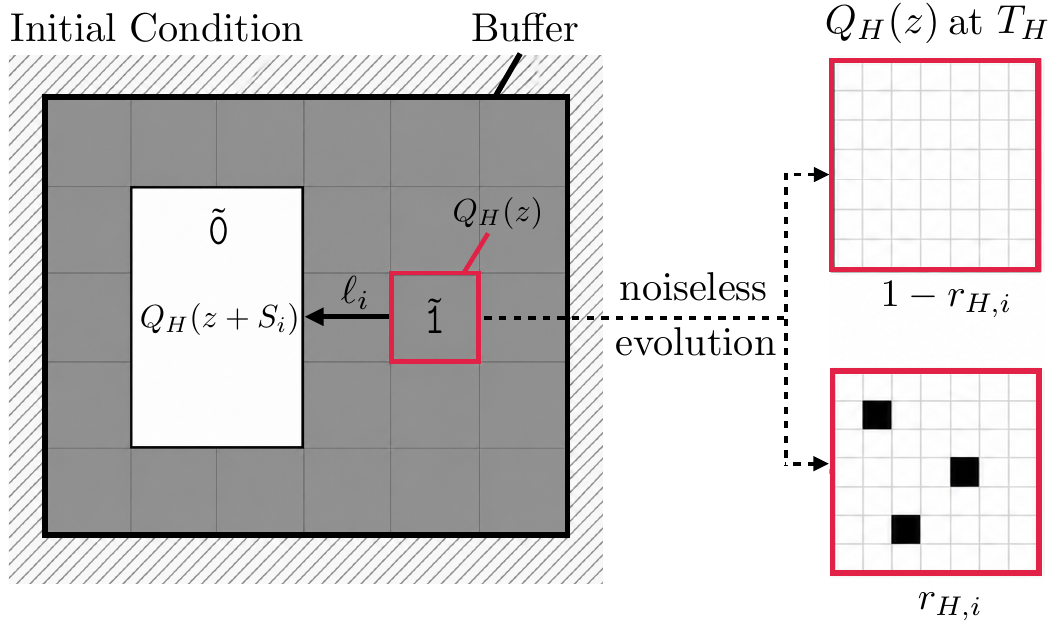}
    \caption{\textbf{Illustration of the noiseless block process.}
   (Left) For a stable direction $\ell_i$, we define the block process to start from the worst-case configuration: the predecessor region $Q_H(z+S_i)$ is clean ($\tilde{\tt 0}$), while the rest of the buffer is defective ($\tilde{\tt 1}$). All the sites outside the buffer are frozen to ${\tt 1}$. (Right) After noiseless continuous-time evolution for a block time $T_H$, the block process for the core cell $Q_H(z)$ succeeds if every site in the core is clean and fails if at least one defect remains (black squares). These outcomes have probabilities $1-r_{H,i}$ and $r_{H,i}$, respectively.}
    \label{fig:block-failure}
\end{figure}

For each $i$, we define the \emph{noiseless block process} by running the microscopic dynamics inside $\cB_H(z)$, initially clean on $Q_H(z+S_i)$ and defective on all remaining buffer sites, with the exterior frozen to ${\tt 1}$; see  \figref{fig:block-failure}. 
We compare the original noiseless dynamics and the auxiliary block process using the same update schedule inside the buffer. Since the auxiliary process starts from a more defective configuration and has a fully defective exterior, monotonicity ensures that it remains at least as defective as the original process throughout the evolution.

We define the \emph{block-failure probability} $r_{H,i}$ as the probability that the core cell $Q_H(z)$ still contains a defect after time $T_H$ (see \figref{fig:block-failure}):
\begin{align}
r_{H,i}
=
\PP\Big(
Q_H(z)\text{ contains a defect at time }T_H
\Big).
\label{eq:block_failure}
\end{align}
The essential dynamical requirement is that $r_{H,i}\to0$ as $H\to\infty$ for every direction in the balanced family.

The block-failure probabilities above control deviations from the prescribed coarse evolution caused by random update times in the noiseless dynamics. Hostile noise provides an additional source of deviations.
We therefore declare an effective error ${\sf error}_v$ at a coarse spacetime point $v=(n,z)$ whenever at least one of the corresponding noiseless block tests fails or a hostile noise mark occurs in $\cB_H(z)\times(nT_H,(n+1)T_H]$.
Writing $b_H:=|\cB_H(z)|$, the probability of an individual effective error satisfies
\begin{equation}
\PP({\sf error}_v)\leq p_H(\eps)
:= 
\sum_{i=1}^{m}r_{H,i}+\eps b_H T_H .
\end{equation}
which can be made to vanish as $H \rightarrow \infty$ and $\eps \rightarrow 0$.

However, these effective errors need not be independent: when overlapping buffers share a single hostile noise mark, errors at coarse cells become correlated. To make this correlated error sparse, we require each buffer to lie within a fixed finite neighborhood of $z$ in coarse-lattice coordinates, independent of $H$.
Each spacetime buffer then overlaps only a bounded number of others.
Consequently, from any collection of specified error events, we can select at least a fixed positive fraction, independent of $H$, whose spacetime buffers are pairwise disjoint.
The selected events depend on disjoint sets of Poisson marks and are therefore mutually independent.
This makes the simultaneous occurrence of all the specified errors exponentially suppressed with the number of events, i.e., \emph{sparse}.  
We now state the central result.

\begin{figure*}[!t]
    \centering
    \includegraphics[width=\linewidth]{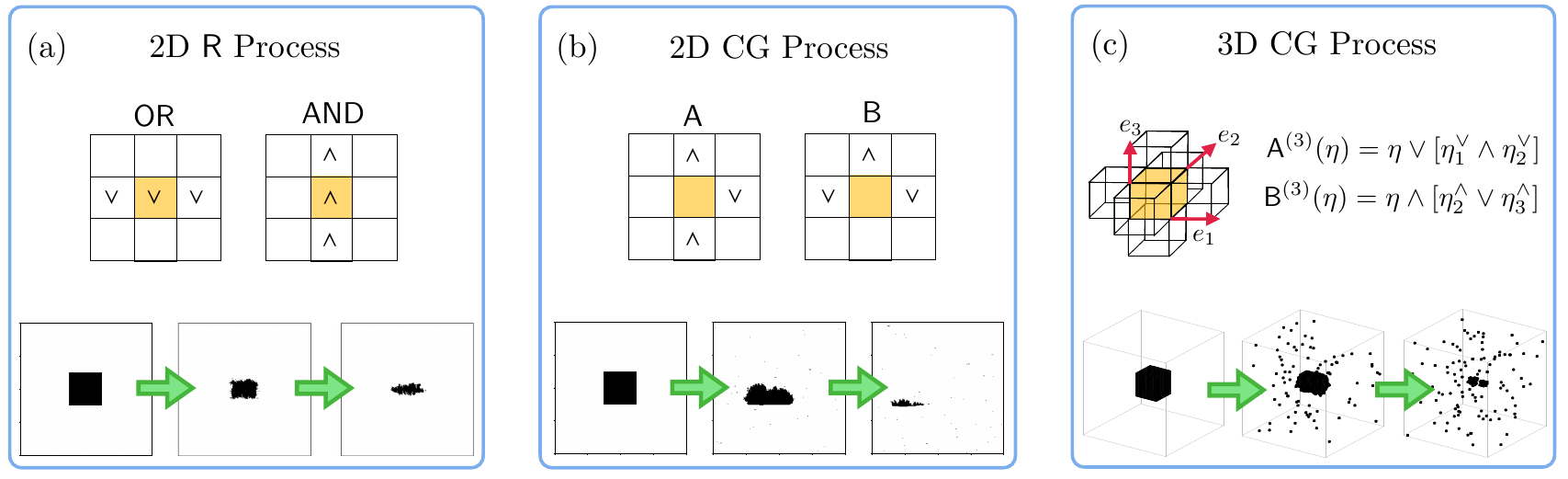}
    \caption{\textbf{Representative robust continuous-time PCAs.}
\textbf{(a)} The {\sf R} process, with horizontal OR and vertical AND updates.
\textbf{(b)} The 2D CG process, consisting of the two
local update rules ${\sf A}$ and ${\sf B}$ shown above.
\textbf{(c)} The 3D CG process, whose update rules
${\sf A}^{(3)}$ and ${\sf B}^{(3)}$ couple the $e_1e_2$ and $e_2e_3$
coordinate planes, respectively. Here $\eta=\eta(x)$, and
$\eta_i^\lor$ ($\eta_i^\land$) denotes the OR (AND) of the spins at $x\pm e_i$.
Yellow cells indicate the sites being updated.
The lower snapshots illustrate the erosion of an initially hypercubic-shaped defect under asynchronous dynamics with equal rates. The small black (white) points denote ${\tt 1}$ ($\tt 0$).
}
    \label{fig:overview}
\end{figure*}

\begin{theorem}[Continuous-time stability criterion]
\label{thm:maintext}
Let $\cA$ be a finite-range, translation-invariant, monotone continuous-time process for which $\bar{\tt 0}$ is an absorbing state in the absence of noise.
Suppose there exists a family of spacetime coarse-grainings, characterized by a scale $H>0$, with a scale-independent coarse map $\Phi$ constructed from a balanced family of stable directions $\{\ell_i,\alpha_i\}_{i=1}^m$, such that $\sum_{i=1}^{m}\alpha_i\delta_i>0.$ We require that the noiseless block-failure probabilities $r_{H,i}$ obey
\begin{align} \label{eq:assumptionKey}
\lim_{H\to\infty}\max_i r_{H,i}=0,
\end{align}
which are defined using a buffer whose size, in units of coarse-grained cells, remains uniformly bounded as $H \rightarrow \infty$.
Then the homogeneous state $\bar{\tt 0}$ is stable in the sense of \eqnref{eq:stability_maintext}.
\end{theorem}

\begin{proof}
    The proof follows directly from the coarse-graining construction discussed in the previous paragraphs.
Suppose that $\Phi[X_n](z)={\tt 0}$.
Then, by the definition of $\Phi$ [\eqnref{eq:Phi_maintext}], there is at least one stable direction $i$ for which every predecessor cell in $z+S_i$ is clean at time $t_n$. 
By monotonicity, the subsequent microscopic evolution inside the corresponding buffer is dominated by the worst-case noiseless block process used to define $r_{H,i}$, unless a hostile error occurs during the interval $(t_n,t_{n+1}]$.
Consequently, the coarse configuration after one microscopic time block obeys the pathwise bound
\begin{align}
X_{n+1}(z)
\leq
\Phi[X_n](z)\lor\mathbf{1}_{{\sf err},n}(z),
\label{eq:domination}
\end{align}
where ${\sf error}_{n,z}$ denotes the event that either some noiseless block process fails or a hostile error occurs in the buffer.
Thus, the microscopic dynamics are dominated by the deterministic eroder $\Phi$ supplemented by an effective spacetime error field, as illustrated in~\figref{fig:coarse-graining}(c).

To bound the effective error field independently of the particular configuration $X_n$, we define ${\sf error}_{n,z}$ conservatively to occur if any of the $m$ noiseless block processes associated with $z$ fail, or if at least one hostile error occurs in the buffer of $z$ during $(t_n,t_{n+1}]$.
Taking a union bound over the failure events gives
\begin{align}
\PP({\sf error}_{n,z}) &\leq \sum_{i=1}^{m} r_{H,i} + \left(1-e^{-\eps b_HT_H}\right)
\nonumber\\
&\leq \sum_{i=1}^{m} r_{H,i}
+
\eps b_HT_H
=:p_H(\eps),
\label{eq:block_error_bound}
\end{align}
where we have used the fact that translational invariance makes $b_H \equiv |\cB_H(z)|$ independent of $z$. 
The two terms in \eqnref{eq:block_error_bound} have distinct origins: the first measures the failure of the noiseless microscopic dynamics to realize the intended coarse evolution, while the second originates from physical noise during the coarse time interval.
For sufficiently large $H$, the first term is small by assumption \eqnref{eq:assumptionKey}. 
At that fixed $H$, the second term can be made arbitrarily small by taking $\eps$ sufficiently small~\cite{Gray1999}.

As mentioned before, the effective errors are not independent, since neighboring coarse cells may have overlapping buffers. The scale-independent locality condition, however, bounds these correlations independently of $H$. Indeed, each event ${\sf error}_{n,z}$ depends only on the Poisson marks in a fixed finite neighborhood of $z$ on the coarse lattice during the time interval $(t_n,t_{n+1}]$. The corresponding dependency graph on coarse spacetime therefore has a chromatic number $K$ bounded independently of $H$. For any finite set $A$ of coarse spacetime points, we may color this graph with $K$ colors, and let $A_k$ be the largest color class. Using the pigeonhole principle, it immediately follows that $|A_k|\geq |A|/K$. The events within $A_k$ depend on disjoint sets of Poisson marks and are therefore independent, which gives
\begin{align}
\PP\bigg(\bigcap_{v\in A}{\sf error}_v\bigg)
&\leq
\PP\bigg(\bigcap_{v\in A_k}{\sf error}_v\bigg) \leq
p_H(\eps)^{|A|/K}.
\label{eq:sparse_error_proof}
\end{align}
Thus, the effective error field is sparse~\cite{GacsTorma2022}, with a sparsity parameter of $q_H(\eps)=p_H(\eps)^{1/K}$, which can be made arbitrarily small by first increasing $H$ and then decreasing $\eps$ at that fixed $H$.

Finally, the condition $\sum_i\alpha_i\delta_i>0$ ensures that $\Phi$ is a geometric eroder. To apply Toom's theorem, introduce an auxiliary coarse process driven by the same effective error events:
\begin{align}
Y_{n+1}(z) &= \Phi[Y_n](z)\lor \mathbf{1}_{{\sf err},n}(z)
\end{align}
with $Y_0 = \bar{\tt 0}$. Since $X_0=Y_0=\bar{\tt 0}$, and the map $\Phi$ is monotone, the pathwise bound in \eqnref{eq:domination} gives
\begin{align}
X_{n+1}(z)
&\leq \Phi[X_n](z)\lor\mathbf{1}_{{\sf err},n}(z)\nonumber\\
&\leq \Phi[Y_n](z)\lor\mathbf{1}_{{\sf err},n}(z)
=Y_{n+1}(z)
\end{align}
whenever $X_n\leq Y_n$. Thus, by induction, $X_n\leq Y_n$ holds pathwise for every $n$. 
\end{proof}

By \eqnref{eq:sparse_error_proof}, the error field driving $Y_n$ is sparse, with parameter $q_H(\eps)$. Toom's theorem~\cite{Toom1980,GacsToomProof,GacsTorma2022} therefore bounds $\PP(Y_n(z)={\tt 1})$ uniformly in $n$ by a quantity that vanishes as $q_H(\eps)\to0$ (see Lemma~\ref{lem:main-toom}). For every microscopic site $x\in Q_H(z)$, we have $\eta_{t_n}(x)\leq X_n(z)\leq Y_n(z)$, implying that this bound transfers to the microscopic process at the block times. Choosing $H$ sufficiently large and then $\eps$ sufficiently small makes the bound vanishingly small. To extend the bound from discrete block times to all times, we use the temporal monotonicity of the noisy process, i.e., starting from $\bar{\tt 0}$ a given site is at least as likely to contain a defect at a later time (see \appref{sec:sm-monotone}). This concludes the proof of stability in the sense of \eqnref{eq:stability_maintext}.

\section{Applications}
\label{sec:applications}

As a consistency check, we establish that the stability condition in Gray's theorem~\cite{Gray1999} for exclusive birth-death processes is a special case of our criterion; see \appref{sec:gray}. 
Next, we focus on three continuous-time processes with competing microscopic update rules that violate birth-death exclusivity: ($i$) the {\sf R} process~\cite{LakeRo2025}, ($ii$) the 2D CG process, and ($iii$) the 3D CG process, see \figref{fig:overview}.

In these applications, the central task is to construct a coarse-graining that satisfies both the geometric and dynamical requirements of Theorem~\ref{thm:maintext}. 
The local update rules determine the available stable directions.
The difficult part is to choose the predecessor regions $\{ S_i \}$, block dimensions $H$, and coarse timestep $T_H$. 
These choices must ensure that $\Phi$ erodes every finite defect set and that, whenever it prescribes a clean target block, the noiseless microscopic dynamics leave that block clean at time $T_H$ with block-failure probability vanishing as $H\to\infty$. 
We also allow unequal update rates to investigate whether bistability persists under an imbalance between competing updates. For all processes considered, we find that bistability remains robust to such rate imbalances.

\subsection{Continuous-time {\sf R} Process}

Here we normalize the AND update rate to $1$ and denote the OR update rate $\mu$.  
For the $\bar{\tt 0}$ configuration of the {\sf R} process, the directions $\pm e_2^*$ form a balanced family of stable directions. 
Horizontal OR updates act within individual rows, while vertical AND updates preserve clean sites and can clear defects adjacent to a clean horizontal half-space, providing a mechanism for its advance. 
Thus, the geometric requirements of Theorem~\ref{thm:maintext} are readily verified. 
We use this geometry to construct a coarse eroder and choose the block parameters so that the corresponding block-failure probabilities $r_{H,i}$ vanish as $H\to\infty$.

\begin{figure}[!t]
    \centering
    \includegraphics[width=\linewidth]{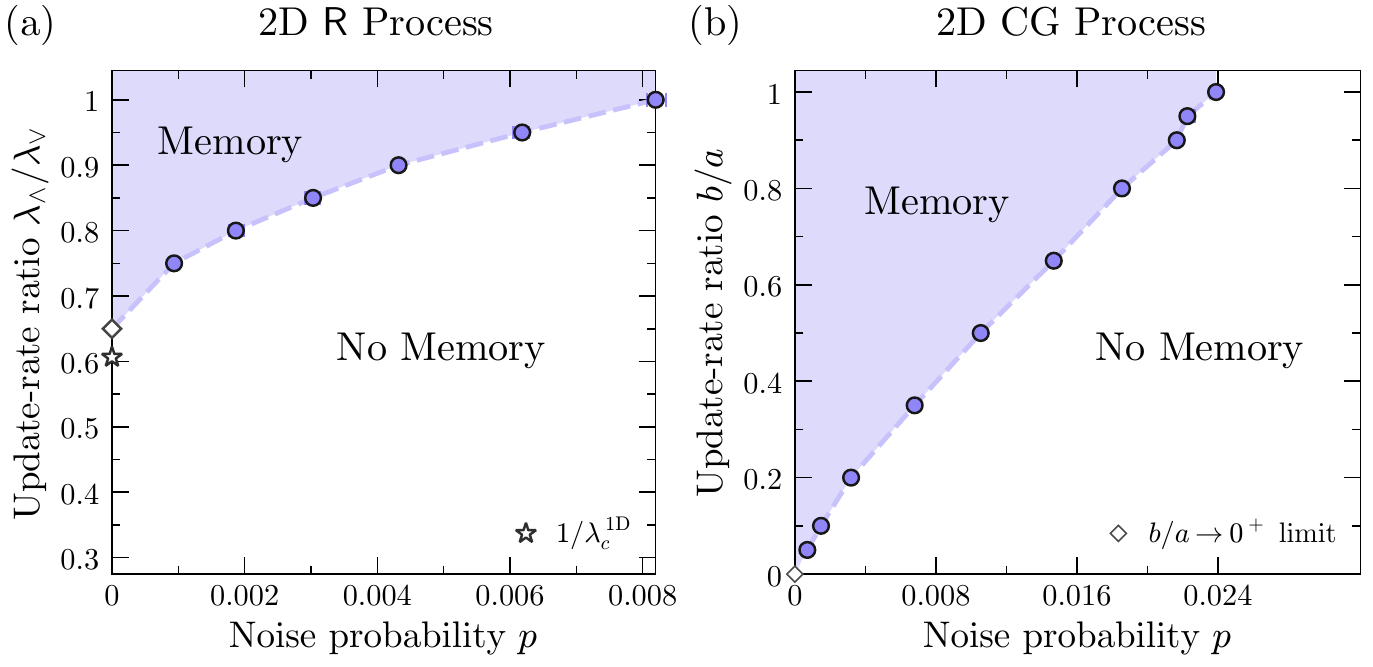}
    \caption{\textbf{Phase diagrams of the asynchronous 2D (a) {\sf R} and (b) CG processes.} Both processes start in the $\bar{\tt 0}$ state under maximally adversarial noise of probability $p$ ($\eta=1$). For CG, the bistable memory region extends toward $b/a=0$, while for {\sf R}, a nonzero minimum $\lambda_{\land}/\lambda_{\lor}$ is required for bistable memory. Both trends agree closely with our analytical results. The open star in (a) marks the lower endpoint $1/\lambda_c^{\rm 1D}$ of the analytically established bistability interval, where $\lambda_c^{\rm 1D}$ is the critical rate of the one-dimensional contact process. We identify the ``no memory'' phase using the lifetime-scaling approach of Ref.~\cite{Ray_2024} (see \appref{sec:app_numerics}).}
    \label{fig:phase-diagrams}
\end{figure}

To understand the asynchronous erosion mechanism of the {\sf R} process, consider the dynamics of a row immediately below an infinite clean row. The clean row remains clean under the noiseless dynamics, so the exposed row evolves autonomously: each defective site is reset to ${\tt 0}$ at rate $1$, independently of the configuration below, while each clean site becomes defective at rate $\mu$ whenever at least one horizontal neighbor is defective, as illustrated in \figref{fig:R-coarse}(a). Thus, at a site with $n\in\{0,1,2\}$ defective horizontal neighbors, the birth rate is $\mu\mathbf{1}_{\{n\geq1\}}$ and the death rate is $1$.

The exposed-row process described above is stochastically dominated by the standard one-dimensional contact process with a death rate of $1$ and a birth rate of $\mu$ per defective neighbor. Indeed, when a site has $n$ defective neighbors, the birth rate in the latter process is $\mu n\geq\mu\mathbf{1}_{\{n\geq1\}}$, while the death rates coincide. 
It is known that for $\mu<\lambda^{\rm 1D}_c\simeq 1.649$, the contact process is subcritical~\cite{ZieoldGrillenberger1988}. 
Standard subcritical decay estimates then imply that, even when the exposed row is initially fully defective, the probability that any given site is defective at time $t$ is bounded by $Ce^{-\gamma t}$ for some $C,\gamma>0$.

We now turn the exposed-row decay into the block estimates required by Theorem~\ref{thm:maintext}. 
To this end, we begin by coarse graining the lattice into anisotropic boxes $Q_H(z)$ of height $H$ and width $L_H$, whose $H$-dependence will be set later. As illustrated in Fig.~\ref{fig:R-coarse}(b), the predecessor regions associated with the advancing directions $\pm e_2^*$ are $S_1=\{u e_1+e_2:|u|\leq K_0\}$ and $S_2=\{u e_1-e_2:|u|\leq K_0\}$, where $K_0$ is a constant fixed by $\mu$ (via the decay rate of the corresponding contact process). The buffer $B_H(z)$ consists of all coarse cells within a distance $K_0+1$ of $z$, and its exterior is frozen at ${\tt 1}$. By reflection symmetry about the horizontal axis, it suffices to consider only $S_1$. In the corresponding noiseless block process, the source strip immediately above the core is clean, while the rest of the buffer is initially defective. We show that this clean strip clears all $H$ rows of the core by time $T_H$ with probability tending to one as $H\to\infty$.

\begin{figure}[!t]
    \centering
    \includegraphics[width=1\linewidth]{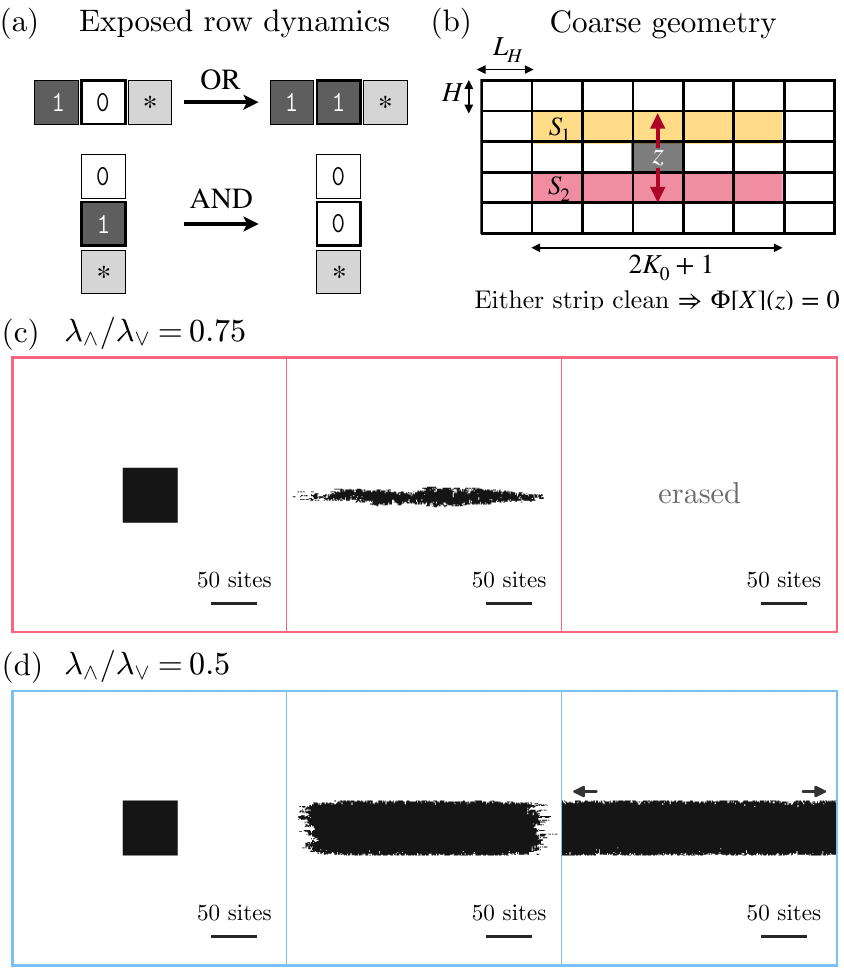}
\caption{
\textbf{Stability mechanism of the {\sf R} process.} 
\textbf{(a)} Contact-process comparison along an exposed row of defects. Scaled in units of $\lambda_{\land}^{-1}$, horizontal OR updates generate defects at a rate of $\mu=\lambda_{\lor}/\lambda_{\land}$ whenever a horizontal neighbor is defective. Conversely, vertical AND updates applied beneath a clean row eliminate defects at a unit rate. The symbol $\ast$ denotes an arbitrary state. 
\textbf{(b)} Coarse-grained geometry featuring cells of width $L_H$ and height $H$. The colored strips represent alternative predecessor regions for the target cell $z$; if either strip is clean, the map $\Phi$ dictates a clean target. Each block test is initialized with one designated strip clean, while all remaining interior sites and the fixed exterior buffer are set to the defective state ({\tt 1}). 
\textbf{(c,d)} Noiseless evolution for rate ratios $\lambda_{\land}/\lambda_{\lor}=0.75$ and $0.5$, respectively. Defects ({\tt 1}, black) initially form a $60\times60$ square within a clean background ({\tt 0}, white). The complete erasure observed in (c) contrasted with the persistent horizontal spreading in (d) demonstrates the breakdown of the erosion mechanism below a critical rate ratio. Arrows in the final panel of (d) indicate defect propagation continuing beyond the field of view.
}
\label{fig:R-coarse}
\end{figure}

\begin{proposition}
\label{prop:R_maintext}
    For every fixed $0<\mu<\lambda^{\rm 1D}_c$, there exists $K_0\in\mathbb N$, independent of $H$, such that the anisotropic coarse-graining described above with $L_H=H\lceil\log H\rceil$, $T_H=\mathcal{O}(H\log H)$, and buffer radius $K_0+1$ (on the coarse-grained lattice) satisfies
\begin{equation}
\lim_{H\to\infty}\max_{i\in\{1,2\}}r_{H,i}=0,
\end{equation}
where $r_{H,i}$ are the noiseless block-failure probabilities defined in \eqnref{eq:block_failure}.
\end{proposition}

\begin{proof}
If the clean region were horizontally infinite, the result would follow immediately by applying the exposed-row estimate successively to each row. The source strip, however, has a finite width and can therefore be invaded by defects from its lateral boundaries due to OR updates. We accommodate the effect of this invasion through a nested sequence of intervals $I_0\supset I_1\supset\cdots\supset I_H$. The interval $I_0$ lies in the clean row immediately above $Q_H(z)$ and is wider than the horizontal extent of the core, whereas $I_H$ coincides with the horizontal extent of $Q_H(z)$. At the $k$th stage of the block process, the interval $I_k$ acts as a clean ceiling for the row below. While that row is being cleared, defects may penetrate a controlled distance into $I_k$ from either end (see~\appref{sec:sm-R-proof}). Sacrificing these boundary regions leaves the smaller interval $I_{k+1}$ clean with exponentially suppressed failure (see \appref{sec:sm-R-block}). The row below is then cleared on this smaller interval $I_{k+1}$ by the subcritical exposed-row dynamics, and the resulting clean interval becomes a clean boundary for the next stage. 
The same penetration bound also ensures that previously cleared rows remain clean on $I_{k+1}$ with high probability.
Iterating through this step clears the core row by row.

To make the above argument more precise, let $\alpha$ be the time allotted to clear each row such that $T_H = H \alpha$, and let $l_{\rm lat}$ be the lateral margin sacrificed at each end in $\alpha$ duration. 
The exponential decay of the exposed-row process (from the domination by the subcritical contact process) allows us to take $\alpha=\mathcal{O}(\log(HL_H))$, which depends on $\mu$. 
Within this duration, it can be shown that the probability that the defects invading into the row from its lateral ends travel farther than a finite (but sufficiently large) multiple of $\alpha$ is exponentially suppressed in $\alpha$. This allows us to take $l_{\rm lat} = \mathcal{O}(\alpha)$. 
After $H$ stages, the total lateral loss is therefore $H l_{\rm lat}=\mathcal{O}(H\log (L_H H))$, which can be accommodated within a finite, scale-independent number of coarse cells by choosing $L_H=H\lceil\log H\rceil$ and taking $K_0$ sufficiently large but independent of $H$. Since $T_H=H\alpha$, it can be shown to be $\mathcal{O}(H \log H)$. 

Finally, each stage must both clear the newly exposed row and protect the rows cleared earlier from defects entering from the sides.
We bound the probability that either requirement fails and sum these bounds over all stages and relevant sites. 
Taking sufficiently large (but scale independent) constants in the choices of $\alpha$ and $l_{\rm lat}$ ensures that the strong suppression of these failures compensates for the growing number of events being counted. The resulting union bound gives $r_{H,1}=\mathcal{O}((HL_H)^{-2})\to0$
(see~\appref{sec:sm-R-proof} for a detailed proof and explicit choices of constants). Invoking the reflection symmetry immediately yields the same estimate for $r_{H,2}$, fulfilling the dynamical requirements of Theorem~\ref{thm:maintext}. 
\end{proof}

Accounting for hostile errors as in \eqnref{eq:block_error_bound}, together with Prop.~\ref{prop:R_maintext}, completes the stability proof of $\bar{\tt 0}$ under the continuous-time {\sf R} process. Finally, using the fact that at equal update rates, the {\sf R} automaton is symmetric under the combined action of a global spin flip and a $\pi/2$-rotation proves the stability of the $\bar{\tt 1}$ state.

Ref.~\cite{LakeRo2025} numerically observed the stability of the asynchronous {\sf R} process with equal rates for OR and AND updates; our contact-process comparison rigorously  predicts bistability beyond the equal rate case. In particular, our treatment establishes the stability of the $\bar{\tt 0}$ configuration even when OR updates are more frequent (and thus more adversarial) than AND updates, provided $\mu<\lambda_c^{\rm 1D}$. Applying the same argument to the opposite configuration gives the following sufficient condition for bistability:
\begin{align}
\label{eq:bistability-range}
    (\lambda^{\rm 1D}_c)^{-1} < \mu < \lambda^{\rm 1D}_c.
\end{align}
The existence of this bistability range suggests that there could be a critical threshold for $\lambda_{\land}/\lambda_{\lor}$, below which the noiseless {\sf R} process loses its ability to erode defects and maintain the uniform $\bar{\tt 0}$ state. Figs.~\ref{fig:R-coarse}(c) and (d) illustrate this breakdown, contrasting complete erasure with persistent defect propagation, corroborating the expectation.

We further conducted Monte Carlo simulations to numerically estimate the memory phase diagram for the {\sf R} process using the approach of Ref.~\cite{Ray_2024}. As shown in Fig.~\ref{fig:lifetime-scaling}(a), for each rate ratio $\lambda_\land/\lambda_\lor$, we determined the critical noise probability $p_c$ below which the mean lifetime of the $\bar{\tt 0}$ state grows exponentially with system size under maximal adversarial bias. These numerical thresholds agree closely with the analytical bistability range in Eq.~\eqref{eq:bistability-range}. This established robustness in 2D naturally raises the question of whether bistability persists in higher-dimensional generalizations. Interestingly, while the 3D {\sf R} process is noise-robust under synchronous updates, it fails to exhibit bistability in the asynchronous regime~\cite{LakeRo2025}. 

Without loss of generality, we take the 3D {\sf R} process to have two OR update types, $\lor_1$ and $\lor_2$, and one AND update type, $\land_3$, each occurring at a unit rate. Here, $\lor_i$ ($\land_i$) updates the state at $x$ to the OR (AND) of the spins at $x-e_i$, $x$, and $x+e_i$. To understand its instability, consider the dynamics on an exposed two-dimensional face of a defect set. We again compare the exposed-face dynamics with a contact process, now on $\Z^2$, but the direction of stochastic domination is reversed. The OR updates produce defects at a rate of at least $1/2$ per defective nearest neighbor, while the AND updates remove defects at a rate of at most one. Thus, the exposed-face dynamics stochastically dominate a contact process with a  birth rate of $1/2$ per neighbor and a death rate of one. The critical birth rate (per defective neighbor) of the 2D contact process has been numerically estimated to be $\lambda^{\rm 2D}_c\approx 0.41$~\cite{2D_CP_Vojta}. Provided that $\lambda^{\rm 2D}_c<1/2$, the 2D contact process with a birth rate per neighbor of $1/2$ is supercritical. By stochastic domination, defects on an exposed face of a defect set have a positive probability of persisting indefinitely under noiseless 3D {\sf R} dynamics, even when initially confined to a finite set.
Consequently, even though the 3D $\mathsf{R}$ process satisfies the geometric conditions for constructing a coarse eroder, it fails the dynamical requirements of Theorem~\ref{thm:maintext}: its block-failure probabilities do not vanish as $H\to\infty$ (see~\appref{sec:sm-3dR-proof} for a detailed proof).

\subsection{2D Cross-Gated Process}
\label{sec:main-2dCG}

\begin{figure}[t]
    \centering
    \includegraphics[width=\linewidth]{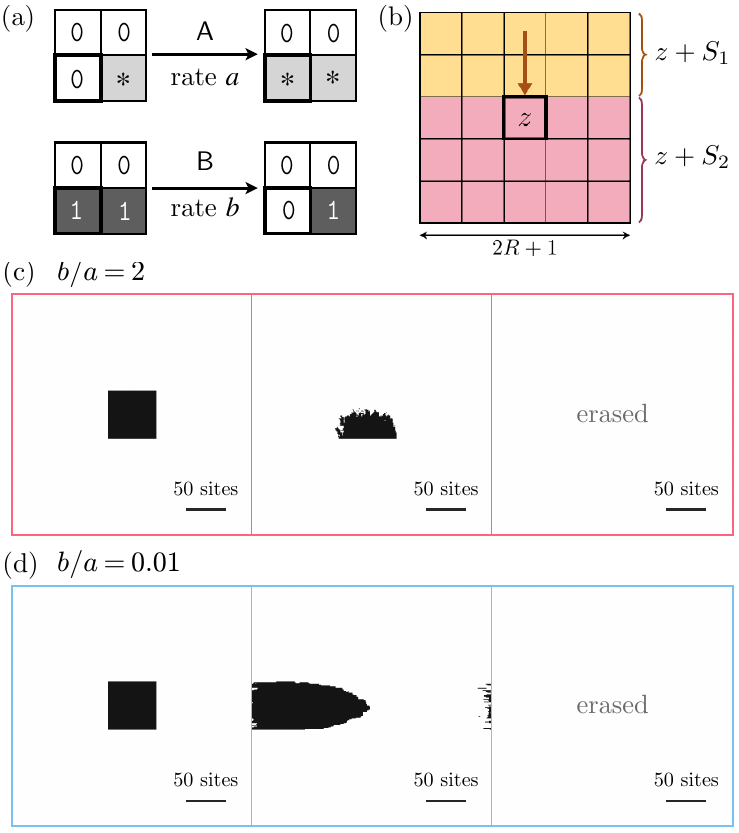}
    \caption{\textbf{Stability mechanism of the 2D CG process.} \textbf{(a)} Below a clean half-plane, the {\sf A} updates simply copy the eastern spin at rate $a$, while the {\sf B} updates clear defects at rate $b$.  Asterisks denote either state. \textbf{(b)} The buffer contains $(2R+1)^2$ coarse cells of side $H$. A clean upper predecessor region $z+S_1$ allows $\Phi$ to advance the clean front into $z$. \textbf{(c,d)} Noiseless evolution of a $60\times60$ defect square on a $300\times300$ torus at $b/a=2$ and $0.01$, respectively. The 2D CG process retains erosion under strongly unequal update rates, with substantial transient growth in (d), whereas the ${\sf R}$ process loses erosion beyond a critical rate imbalance. Black and white denote ${\tt 1}$ and ${\tt 0}$.}
    \label{fig:2dCG}
\end{figure}

Unlike the ${\sf R}$ process, whose update rules involve either only OR or only AND operations, we now introduce the two-dimensional cross-gated (CG) process in which each rule combines both operations. The process is defined on $\Z^2$ with binary variables placed on every site. It involves two update rules, ${\sf A}$ and ${\sf B}$, defined for every $x \in \Z^2$ and $\eta \in \{{\tt 0}, {\tt 1}\}^{\Z^2}$, as follows:
\begin{align}
    {\sf A}_x&: \eta(x) \,\mapsto\, \eta({x+e_1}) \lor \left[\eta({x+e_2}) \land \eta({x-e_2})\right],\\
     {\sf B}_x&: \eta(x) \,\mapsto\, \eta({x+e_2}) \land \left[\eta({x+e_1}) \lor \eta({x-e_1})\right],
\end{align}
where $a$ and $b$ are the rates for the $A_x$ and $B_x$ updates, respectively. Fig.~\ref{fig:overview}(b) illustrates the update rules {\sf A} and {\sf B} in the top panel, while the bottom panel highlights the erosion mechanism of the continuous-time CG process (for $a = b$). A striking feature of the 2D CG process is that it preserves memory for all positive rates $a, b > 0$, regardless of their ratio. 
This behavior is illustrated in \figref{fig:2dCG}(c,d), which demonstrates the stability of only the $\bar{\tt 0}$ configuration without loss of generality.

For the $\bar{\tt 0}$ ($\bar{\tt 1}$) configuration, one can readily verify that the two stable directions $\ell_1 = e^{\ast}_2$ ($e^{\ast}_1$) and $\ell_2 = -e^{\ast}_2$ ($-e^{\ast}_1$) form a balanced family, with the former advancing and the latter static (see~\appref{sec:sm-2dCG-proof} for the identification of stable directions).
In particular, for the $\bar{\tt 0}$ configuration, ${\mathsf B}_x$ erodes the defect set from above, while ${\mathsf A}_x$ induces an overall leftward drift of the defect set (the roles of the two rules are interchanged for the $\bar{\tt 1}$ configuration).
Indeed, at $b/a = 2$, the defect set rapidly shrinks from top to bottom, with the bottommost defects disappearing last.
At $b/a = 0.01$, the defect set instead spreads as it drifts leftward, yet its confinement along the vertical direction ultimately leads to complete erosion.

We coarse-grain $\Z^2$ using isotropic boxes $Q_H(z)$ ($z \in \Z^2_{\rm c.g.}$) of linear size $H$, and we take the block time as $T_H = \tau H$, where $\tau$ is a scale-independent positive constant to be fixed subsequently. For a finite positive integer $R$ (also to be fixed later), we define the buffer $B$ as the union of coarse displacements $s \in \Z^2_{\rm c.g.}$ such that $\|s\|_{\infty} \leq R$. The set of microscopic sites contained in the buffer of a cell $z$ is given by: $\cB_H(z) = \bigcup_{\|s\|_{\infty} \leq R} Q_H(z+s)$ whose cardinality $|\cB_H(z)|$ is $ (2R+1)^2 H^2$. We take the predecessor regions $S_{1,2}$ corresponding to the stable directions $\ell_{1,2}$ as follows:
\begin{align}
    S_1 &= \{s \in \Z^2_{\rm c.g.}: \|s\|_{\infty} \leq R, \, s_2 \geq 1\},\label{eq:main-2dCG-S1}\\
    S_2 &= \{s \in \Z^2_{\rm c.g.}: \|s\|_{\infty} \leq R, \, s_2 \leq 0\}. \label{eq:main-2dCG-S2}
\end{align}
Equivalently, the above choice corresponds to choosing the coarse displacements $\delta_i$ ($i \in \{1,2\}$) as $\delta_1 = 1$ and $\delta_2 = 0$. As shown in Fig.~\ref{fig:2dCG}(b), if the upper region $z+S_1$ is clean, $\Phi$ advances the clean front downward by one coarse unit ($\delta_1=1$).
The lower region $z+S_2$ contains the core cell $z$, so if this region is clean, $\Phi$ keeps $z$ clean ($\delta_2=0$).

Next, we run, for each stable direction $i = 1,2$, the adversarial finite-volume block process $\zeta^i$ over $\cB_H(z)$ (for some fixed $z \in \Z^2_{\rm c.g.}$), as defined in Sec.~\ref{sec:main-block-failure}. First, note that since all rates are bounded and the interactions are finite ranged, taking $R$ to be large enough, it is possible to suppress the probability that an exterior site $x\in \cB^c_H(z)$ causally affects some site in the core cell (see~\appref{sec:sm-finite-speed}). For an explicit choice of the constant $R$, refer to~\appref{sec:sm-2dCG-proof}. For the sake of this discussion, it suffices to assume that $R$ is large enough so that the probability of the exterior influence events described above is suppressed exponentially in the block scale. Thus, we now work out the probability that after block time $T_H$, there still exists a defect within the core cell $Q_H(z)$, assuming that no exterior influence event occurs. We can include the contribution of the exterior influence events to the block failure probability at the end.

\begin{proposition}
    For the coarse-graining scheme described above, the block failure probabilities for the noiseless CG process, $r_{H,i}$ for $i \in \{1,2\}$ (defined as in \eqnref{eq:block_failure}), vanish as $H \to \infty$. 
\end{proposition}

\begin{proof} We first bound the block failure probability corresponding to the static direction $i = 2$. For $i = 2$, since $\delta_2 = 0$, the core cell starts off clean initially. Since $\ell_2$ is a stable direction, the corresponding clean half-space remains invariant under the noiseless dynamics. Thus, a defect can appear in the core only through an exterior influence event, whose probability is exponentially suppressed in $H$. Consequently, $r_{H,2}\to 0$ as $H\to\infty$.

We now discuss the more challenging case of the advancing direction $\ell_1$. For $i=1$, proving that the block failure probability vanishes as $H \to \infty$ is slightly more nontrivial because, in this case, the core cell starts as fully defective. Consider an infinite-volume process $\hat{\xi}^1$ initialized as {\tt 0} in the upper half-space immediately above $Q_H(z)$ and {\tt 1} everywhere else, using the same update marks as $\zeta^1$. The two processes have identical initial conditions inside the buffer, so they can differ in the core only if an exterior influence event occurs. Thus, it suffices to estimate the probability that a defect remains in the core under $\hat{\xi}^1$ at time $T_H$, since we are working under the assumption that no exterior influence event occurs.

Let $k$ denote the depth into the defective half-plane starting from the boundary of the clean half-plane. For example, the core occupies rows corresponding to $k = 1, \cdots, H$. Each site of the core cell can now be represented as $(j,k)$, where $j$ is the coordinate transverse to $k$, i.e., simply the $x$ coordinate. We denote by $p_k(t)$ the density of defects at depth $k$, which can be expressed as the probability that $\hat{\xi}^1_t(0,k) = {\tt 1}$, due to translational invariance. Note that the update rules ${\sf A}_x$ and ${\sf B}_x$ restricted to sites $x$ of row $k =1$ simply become ${\sf A}_x(\hat{\xi}^1) = \hat{\xi}^1(x+e_1)$ and ${\sf B}_x(\hat{\xi}^1) = {\tt 0}$, as illustrated in Fig.~\ref{fig:2dCG}(a). By translational invariance, the east neighbor has the same defect density as the updated site, so its contribution cancels in the evolution equation of $p_k(t)$. For deeper layers $k>1$, we can obtain a simple recursion inequality by noting that for any configuration $\eta$, the following relations hold: ${\sf A}_x(\eta) \leq \eta(x+e_1) + \eta(x+e_2)$ and ${\sf B}_x(\eta) \leq \eta(x+e_2)$. 
Using these relations together with translational invariance along each layer yields the following:
\begin{equation}
\frac{d}{dt}p_k(t)
\leq -b p_k(t)+(a+b)p_{k-1}(t)
\label{eq:main-2dCG-recursion}
\end{equation}
with $p_0(t) = 0$, where the first term represents defect removal due to ${\sf B}$ updates, while the second term bounds defect propagation from the preceding layer (see \appref{sec:sm-2dCG-proof} for more details). 
To bound this recursion, we consider
$u_k(t)=\exp(\theta k-\gamma t/2)$, where
$\gamma=\min(a,b)$ and
$\theta=\log(2(a+b)/\gamma)$.
By construction, we have $p_k(0)\leq 1\leq u_k(0)$ and
$p_0(t)=0\leq u_0(t)$, along with 
\begin{align} \label{eq:CG_pk} 
    \frac{d}{dt} (p_k - u_k) \leq -b (p_k - u_k) \quad \Rightarrow \quad p_k(t) \leq u_k(t).
\end{align}
Ssee~\appref{sec:sm-2dCG-proof} for more details. 

Take the constant $\tau$ that appears in the definition of $T_H$ to be larger than $2(\theta+1)/\gamma$. Setting $t = T_H$ in \eqnref{eq:CG_pk} and taking a union bound over all $H^2$ sites in $Q_H(z)$, we finally obtain:
\begin{equation}
    \PP(\exists \,x\in Q_H(z): \hat{\xi}^1_{T_H} (x) = {\tt 1}) \leq H^2 e^{-H},
    \label{eq:2dCG-half-plane-bound}
\end{equation}
which vanishes as $H \to \infty$. Combining \eqnref{eq:2dCG-half-plane-bound} with the exponentially suppressed probability of exterior influence events, we find that $r_{H,1} \to 0$ as $H \to \infty$.
\end{proof} 

Since the buffer size is fixed in coarse-lattice units, the vanishing block-failure probabilities satisfy the dynamical requirements of Theorem~\ref{thm:maintext}, proving the stability of the $\bar{\tt 0}$ configuration under sufficiently weak noise. Furthermore, note that a spin flip combined with the swap $e_1\leftrightarrow e_2$ exchanges the two update rules and their rates. Since our stability proof above holds for arbitrary $a,b>0$, the $\bar{\tt 1}$ configuration is also stable. Thus, we have established that the 2D CG process is bistable under sufficiently weak noise of arbitrary bias. 

To test our analytical prediction that bistability persists for all $a, b > 0$, we perform Monte Carlo simulations following the approach of Ref.~\cite{Ray_2024}; the resulting memory phase diagram for the 2D CG process is shown in Fig.~\ref{fig:phase-diagrams}(b). For each rate ratio $b/a$, we obtain the critical noise probability $p_c$ below which the memory lifetime of $\bar{\tt 0}$-state scales exponentially with system size under maximally adversarial bias. Focusing on the unfavorable regime $b/a < 1$ (since $b/a > 1$ inherently favors $\bar{\tt 0}$), Fig.~\ref{fig:phase-diagrams} clearly demonstrates that the memory phase extends all the way to vanishingly small values of $b/a$, corroborating our analytical conclusion.

\subsection{3D Cross-Gated Process}
\label{sec:3d-CG}

We now introduce a three-dimensional generalization of the CG process, with a single binary variable $\eta(x)\in\{{\tt 0},{\tt 1}\}$ at each site $x\in\Z^3$. Similarly to the 2D CG process, it involves two update rules: ${\sf A}^{(3)}$ and ${\sf B}^{(3)}$. 
The two update rules couple intersecting coordinate planes rather than acting on independent two-dimensional layers. 
Defining
\begin{align}
\eta_a^{\lor}(x)
&=\eta(x+e_a)\lor\eta(x-e_a),
\\
\eta_a^{\land}(x)
&=\eta(x+e_a)\land\eta(x-e_a),
\end{align}
for $a = 1,2,3$, the update rules are given by:
\begin{align}
{\sf A}^{(3)}_x:
\eta(x)&\mapsto
\eta(x)\lor[\eta_1^{\lor}(x)\land\eta_2^{\lor}(x)],
\label{eq:3d-CG-A}
\\
{\sf B}^{(3)}_x:
\eta(x)&\mapsto
\eta(x)\land[\eta_2^{\land}(x)\lor\eta_3^{\land}(x)].
\label{eq:3d-CG-B}
\end{align}
At every site, two independent Poisson clocks apply these rules at rates $a$ and $b$, respectively. 
The first rule creates a ${\tt 1}$ when there is a neighboring ${\tt 1}$ along each of the $e_1$ and $e_2$ axes. 
The second rule removes a ${\tt 1}$ when there is a ${\tt 0}$-neighbor along each of the $e_2$ and $e_3$ axes. Both homogeneous configurations, $\bar{\tt 0}$ and $\bar{\tt 1}$, are absorbing states. 
Moreover, spin flip combined with the coordinate exchange $e_1 \leftrightarrow e_3$ exchanges the two rules.
Therefore, we first derive a sufficient condition for the stability of $\bar{\tt 0}$; the result for $\bar{\tt 1}$ then follows by interchanging $a$ and $b$. As a numerical demonstration of stability, we show a finite-size noisy simulation of the erosion of a defect set in Fig.~\ref{fig:overview}(c) [bottom panel]. 

For the $\bar{\tt 0}$ configuration, we first identify the stable directions as follows:
\begin{align}
\ell_1&=e_1^*, & \ell_2&=-e_1^*,
\nonumber\\
\ell_3&=e_2^*+e_3^*, & \ell_4&=e_2^*-e_3^*,
\nonumber\\
\ell_5&=-e_2^*+e_3^*, & \ell_6&=-e_2^*-e_3^*.
\label{eq:3d-CG-stable-directions}
\end{align}
Their sum vanishes, implying that they form a balanced family. To see that they are indeed stable, note that for
$\ell_3,\ldots,\ell_6$, both $e_1$ neighbors of a site in a clean half-space remain in that half-space, preventing creation by ${\sf A}^{(3)}$.
Similarly, for $\ell_1,\ell_2$, the same argument uses the $e_2$ neighbors. Since the ${\sf B}^{(3)}$ update cannot create defects, all six half-spaces (corresponding to $\ell_{1,\cdots,6}$) are invariant. We use $\ell_1,\ell_2$ with zero displacement ($\delta_{1,2} = 0$) and attach positive coarse displacements to the remaining stable directions in the balanced family, i.e., $\delta_i>0$ for $i = 3,\cdots,6$.

As in the 2D CG process, we coarse-grain the lattice using isotropic boxes
\begin{equation}
Q_H(z)=Hz+\{0,\ldots,H-1\}^3,
\end{equation}
and take the block time as $T_H=\tau H$. 
Let $B_R=\{s\in\Z^3_{\rm c.g.}:\|s\|_\infty\le R\}$ and $\cB_H(z)=\bigcup_{s\in B_R}Q_H(z+s)$. We define the predecessor regions as follows:
\begin{align}
S_i&=\{s\in B_R:\ell_i(s)\ge0\},
&&i=1,2,
\label{eq:3d-CG-static-zeros}
\\
S_i&=\{s\in B_R:\ell_i(s)\ge1\},
&&i=3,\ldots,6,
\label{eq:3d-CG-advancing-zeros}
\end{align}
which corresponds to the following choices: $\delta_1=\delta_2=0$ and $\delta_3=\cdots=\delta_6=1$. Both $\tau$ and $R$ are fixed independently of $H$.

\begin{proposition}
\label{prop:3d-CG-block}
Fix the rates $a,b>0$ such that $a/b<\lambda^{\rm 1D}_{c}$, where $\lambda^{\rm 1D}_{c}$ is the critical birth rate per neighbor of the 1D contact process with a unit death rate. For sufficiently large fixed $\tau$ and a sufficiently large fixed $R$, the noiseless block failure probabilities satisfy
\begin{equation}
r_{H,i}\le CH^3e^{-\alpha H},
\quad i=1,\ldots,6,
\end{equation}
for constants $C,\alpha>0$ independent of $H$.
\end{proposition}

\begin{proof}
Here we explain the mechanism underlying the advancing directions and summarize the results while deferring a rigorous proof to~\appref{sec:sm-3dCG-proof}.
By translational invariance, we can restrict ourselves to the core $Q_H(0)$ without any loss of generality. For each stable direction $i$, the block process starts clean on $\bigcup_{s\in S_i}Q_H(s)$ and defective elsewhere in the buffer, with the exterior permanently defective (frozen to {\tt 1}).
Let $\mathcal{E}_H$ represent the event that some site of $\cB^{c}_H(0)$ causally influences the core by time $T_H$. Once $\tau$ is fixed, finite-speed propagation of causal influence allows us to choose $R$ large enough so that $\PP(\mathcal{E}_H)\le C_0H^3e^{-c_0H}$, where $C_0$ and $c_0$ are finite positive constants independent of the block scale $H$. 
For $i=1,2$, the core is clean initially since $\delta_i = 0$. Thus, the only way a core site might end up being defective by the block time $T_H$ is through an exterior influence event $\mathcal{E}_H$, implying that $r_{H,i}\le C_0H^3e^{-c_0H}$.

For an advancing direction, we write $\ell_\sigma=\sigma_2e_2^*+\sigma_3e_3^*$, with $\sigma=(\sigma_2,\sigma_3)\in\{\pm1\}^2$. For an integer $h$, let $\hat\xi$ be an infinite-volume process initialized as ${\tt 0}$ on $\{x:\ell_\sigma(x)\ge h\}$ and ${\tt 1}$ elsewhere. 
The depth into the defective region is measured by $k=h-\ell_\sigma(x)$. The following bounds
\begin{align}
{\sf A}^{(3)}_x(\eta)
&\le\eta(x)\lor\eta_1^{\lor}(x),
\label{eq:3d-CG-A-bound}
\\
{\sf B}^{(3)}_x(\eta)
&\le\eta(x+\sigma_2e_2)\lor\eta(x+\sigma_3e_3)
\label{eq:3d-CG-B-bound}
\end{align}
enable us to define a monotone dominating process $\widetilde\xi$: it uses the right-hand sides as its update values, with the same clocks and initial condition as $\hat{\xi}$.
Hence, we have $\hat\xi_t\le\widetilde\xi_t$. 
For more details on the monotone coupling used here, refer to~\appref{sec:sm-monotone}.

Both parent sites in \eqnref{eq:3d-CG-B-bound} lie at depth $k-1$. 
In the first defective layer they remain clean, so every such update resets its site to zero. 
Let us fix a site $x_0$ in this layer and consider the line $\{x_0+je_1:j\in\Z\}$ parallel to $e_1$.

On this line, the dominating process $\widetilde \xi$ evolves as a one-dimensional process $\chi_t(j)$, initially fully defective, with two independent updates at every site $j$:
\begin{align}
\label{eq:maintext-3dCG-aux}
\chi(j)&\mapsto\chi(j-1)\lor\chi(j)\lor\chi(j+1) \text{ at rate }a,
\nonumber\\
\chi(j)&\mapsto{\tt 0} \text{ at rate }b.
\end{align}
The process $\chi_t$ is stochastically dominated by a 1D contact process with a birth rate of $a$ per infected nearest neighbor and a death rate of $b$. Let $\lambda_c^{\rm 1D}$ denote the critical birth rate per neighbor of the 1D contact process at unit death rate. The contact process is subcritical whenever $a/b<\lambda_c^{\rm 1D}$. In this regime, the expected number of descendants of a single defect under $\chi_t$ is bounded by $Ke^{-ct}$~\cite{Swart2016}, where $K\ge1$ and $c>0$ depend only on $a$ and $b$.

The evolution of a defect set under $\chi$ is additive~\footnote{Let $X_t^A$ denote the state of a set-valued Markov process started from $A$. The process is called \emph{additive} if copies started from different initial sets can be driven by the same randomness so that $X_t^{A\cup B}=X_t^A\cup X_t^B$ for all $A,B$ and $t\ge0$.}. Thus, for a fully defective initial line, we can bound the probability of a defect at the target by summing the probabilities that individual initial defects produce a descendant there. By translation invariance, this sum equals the expected number of descendants of a single defect, giving
\begin{equation}
\PP\bigl(\chi_t(0)={\tt 1}\bigr)\le Ke^{-ct}.
\label{eq:3d-CG-line-decay}
\end{equation}

For deeper layers ($k>1$), we represent each rate-$b$ update as a reset to ${\tt 0}$ followed by an insertion if at least one parent in layer $k-1$ is defective. The OR updates and resets reproduce the autonomous dynamics $\chi$, while the defect insertions account for the influence of the defects present in the preceding layer. This allows us to apply the single-defect estimate to each insertion, with later insertions counted separately.
The mean defect insertion rate per site is at most $2b$ times the defect density of the preceding layer. A defect inserted at time $s$ produces at most $Ke^{-c(t-s)}$ descendants on average at time $t$ under $\chi$. Additivity allows us to sum these contributions over insertion sites and times by taking a union bound. Together with the contribution from initial defects, this gives a recursion relation between the defect densities in successive layers (see \appref{sec:sm-3dCG-proof} for the precise recursion relation).
Iterating this recursion from the clean half-space and using $\hat\xi_t\le\widetilde\xi_t$, we obtain:
\begin{equation}
p_k(t):=\PP\bigl(\hat\xi_t(x)={\tt 1}\bigr)\le K\exp\!\left(\theta k-\frac{ct}{2}\right),
\label{eq:3d-CG-pk}
\end{equation}
where the site $x$ lies at depth $k=h-\ell_\sigma(x)$. The constants $K,c,\theta>0$ depend on the fixed rates $a,b$ but are independent of $k,\,h$ and $H$.

To apply \eqnref{eq:3d-CG-pk} to the core $Q_H(0)$, we choose $h=1+\max_{x\in Q_H(0)}\ell_\sigma(x)$.
Every core site then lies at a depth $1\le k\le2H-1$. This choice also ensures that the clean half-space of $\hat\xi$ intersects the buffer only within the corresponding predecessor region $Q_H(S_i)$. Thus, every buffer site that starts clean in $\hat\xi$ also starts clean in the block process. On the event $\mathcal E_H^c$, the sites outside the buffer have no influence on the core sites in the duration of the block process, implying that monotonicity ensures that every core site defective in the block process at time $T_H$ is also defective in $\hat\xi_{T_H}$. Applying \eqnref{eq:3d-CG-pk} at $T_H=\tau H$ and taking a union bound over the $H^3$ core sites gives:
\begin{align}
r_{H,i}
&\le KH^3e^{2\theta H-c\tau H/2}+\PP(\mathcal E_H) \nonumber\\
&\le KH^3e^{-H}+\PP(\mathcal E_H),
\label{eq:3d-CG-core-failure}
\end{align}
for $i=3,\ldots,6$, where the second inequality follows by choosing $\tau$ such that $c\tau/2\ge2\theta+1$.
\end{proof}

From the above discussion, we conclude that the coarse geometry and block estimates satisfy the requirements of Theorem~\ref{thm:maintext}, establishing the stability of $\bar{\tt 0}$ under sufficiently weak one-sided hostile noise, provided that $a/b<\lambda_c^{\rm 1D}$. A spin flip combined with $e_1\leftrightarrow e_3$ exchanges the two update rules and their rates, implying that the same argument establishes the stability of $\bar{\tt 1}$ when $b/a<\lambda_c^{\rm 1D}$. Consequently, the 3D CG process is bistable under sufficiently weak, arbitrarily biased noise whenever the rates $a$ and $b$ satisfy:
\begin{equation}
\left(\lambda_c^{\rm 1D}\right)^{-1}<\frac{a}{b}<\lambda_c^{\rm 1D}.
\end{equation}

\section{Discussion and outlook}
\label{sec:disc_and_outlook}

In this work, we have established a sufficient stability criterion for finite-range, translation-invariant, monotone continuous-time processes in arbitrary dimensions. This criterion reveals robustness as an emergent, coarse-grained property: suitable interface geometry and probabilistic control of the microscopic dynamics together generate reliable erosion at large scales. Consequently, memory can survive temporal disorder without global synchronization, provided the collective dynamics suppress large-scale defect propagation. Because this mechanism does not rely on an exclusive birth-death structure, our framework extends beyond the setting treated by Gray, as demonstrated by both the continuous-time {\sf R} process~\cite{LakeRo2025} and the 2D and 3D cross-gated processes introduced here.

Our results point to a broader design principle for autonomous error correction. Robustness need not be evident at the level of individual microscopic rules: defects may spread transiently, interfaces may roughen, and droplets may expand along certain directions. What matters instead is that these local excursions do not escalate into sustained large-scale propagation, ensuring that finite defect regions are ultimately driven to extinction. Furthermore, the asynchronous framework naturally accommodates unequal update rates, yielding reliable memory without requiring fine-tuned balance between competing dynamics or any underlying symmetry. From this perspective, designing nonequilibrium memories under temporal disorder amounts to identifying simple local rules whose competing actions collectively stabilize information.

More broadly, understanding how noisy local dynamics sustain global order is a central challenge across nonequilibrium physics. This is particularly true in biophysical and active-matter systems, where macroscopic organization emerges from local stochastic interactions rather than equilibrium free-energy minimization~\cite{Ramaswamy2010,Marchetti2013}. Our results suggest that coarse-grained suppression of defect propagation offers a useful lens for studying such systems well beyond the error-correction setting considered here. At the same time, important classes of dynamics remain outside our current criterion, including eroders that lack a balanced set of stable directions and whose defect droplets drift as they shrink. Relaxing these geometric constraints and extending the framework beyond binary memory should reveal more general mechanisms for stabilizing information. Ultimately, spacetime coarse-graining promises to serve not only as a rigorous proof technique, but also as a constructive guide for discovering robust nonequilibrium phases and fluctuation-stabilized memories.

\acknowledgements

We thank Ethan Lake for sharing his earlier results for Ref.~\cite{LakeRo2025} that inspired us back in 2024. 
S.C. and J.Y.L. are supported by the faculty startup grant at the University of Illinois, Urbana-Champaign. 

\bibliography{refs}

\clearpage 
\newpage

\appendix 
\onecolumngrid 
 
\appendixtableofcontents

\section{Conventions}
\label{sec:app_A}
In this section, we clarify the notation and conventions used in the main text. We denote the space of configurations on $\mathbb{Z}^d$ as $\Omega=\{{\tt 0},{\tt 1}\}^{\Z^d}.$ For $\eta,\, \xi\in\Omega$, we write $\eta\leq \xi$ if $\eta(x)\le \xi(x)$ for every $x\in\Z^d$. Throughout the appendix, we will study the stability of the $\bar{\tt 0}$ state, unless specified otherwise. We will use the notation $\mathbf{1}_{E}$, for any event $E$, to denote the following:
\begin{equation}
    \mathbf{1}_{E} = \begin{cases}
        1 &\text{if }E \text{ occurs},\\
        0 &\text{otherwise}.
    \end{cases}
\end{equation}

\subsection{Local dynamics and structural requirements}
\label{sec:sm-requirements}
For each update type $a\in\cA$, let $f_a:\{{\tt 0},{\tt 1}\}^{\cN_a}\to\{{\tt 0},{\tt 1}\}$ be a Boolean map. Each pair $(x,a)$ has an independent Poisson clock of rate $\lambda_a>0$. When it rings, $\eta(x)$ is replaced by $f_a((\eta(x+\delta))_{\delta\in\cN_a})$, with all other sites unchanged. We call the event $(t,x,a)$ an \emph{update mark}.
We impose the following assumptions.

\begin{assumption}[Finite range and translation invariance]
\label{ass:finite-range}
The set $\cA$ and all neighborhoods $\cN_a$ are finite.
The maps, neighborhoods, and finite rates $\lambda_a$ are the same at every site.
\end{assumption}

\begin{assumption}[Two traps]
\label{ass:two-traps}
Both homogeneous configurations are absorbing: $f_a({\tt 0},\ldots,{\tt 0})={\tt 0}$ and $f_a({\tt 1},\ldots,{\tt 1})={\tt 1}$ for every $a\in\cA$.
\end{assumption}

\begin{assumption}[Boolean monotonicity]
\label{ass:monotone}
Each $f_a$ is nondecreasing in every argument:
$\xi\le\zeta$ on $\cN_a$ implies $f_a(\xi)\le f_a(\zeta)$.
\end{assumption}

We add independent rate-$\eps$ hostile marks at every site, each setting the local state to ${\tt 1}$, and take $0\le\eps\le1$. Define
\begin{equation}
r:=\max_{a\in\cA}\max_{\delta\in\cN_a}\|\delta\|_\infty,
\qquad
\lambda:=1+\sum_{a\in\cA}\lambda_a.
\end{equation}
Thus, each update reads sites within distance $r$, and $\lambda$ bounds the total mark rate at each site. We disregard the trivial cases where $r = 0$, and hereafter assume that $r\geq 1$. Additional friendly ${\tt 1}\to{\tt 0}$ noise can only reduce defects by monotonicity.

For stability of the selected $\bar{\tt 0}$ configuration, only the
zero-trap condition is needed. If the noiseless dynamics,
including the rates, are invariant under $\mathcal X\circ g$,
where $\mathcal X$ flips all spins and $g\in GL(d,\Z)$,
stability of one configuration implies stability of the other.
Otherwise, the two configurations can be analyzed separately.

\begin{proposition}[Well-posedness]\label{prop:harris}
Almost surely, for every finite $D\subset\Z^d$ and $T<\infty$,
the trajectory $(\eta_t|_D)_{0\le t\le T}$ is determined by
finitely many marks and initial spins. The resulting process
is Markovian, and its graphical flow is monotone in the
initial configuration.
\end{proposition}

\begin{proof}
Let us set $K=(2r+1)^d$. A backward dependency chain of $n$ marks ending in $D$ has at most $|D|K^n$ possible spatial sequences.
For each sequence, the expected number of chronologically
ordered realizations in $(0,T]$ is at most $(\lambda T)^n/n!$. Hence, the expected total number of chains
is bounded by the convergent series
$ |D|\sum_{n\ge1}\frac{(K\lambda T)^n}{n!}$.
Almost surely, the backward exploration therefore encounters
only finitely many marks and initial spins. Applying the
relevant updates in temporal order defines the trajectory;
a countable exhaustion of space and time gives the
infinite-volume flow. Each update and hostile reset preserves
order, so the flow is monotone. Independence of future
Poisson increments gives the Markov property.
\end{proof}

\subsection{Monotone coupling}
\label{sec:sm-monotone}
All comparisons below use shared update and hostile-noise marks.

\begin{lemma}[Monotone coupling]\label{lem:couplings}
\begin{enumerate}[label=\textup{(M\arabic*)},leftmargin=3em]
\item\label{M1} \textup{(Initial condition)}
If $\eta_0\le\eta'_0$, then $\eta_t\le\eta'_t$ for all $t\ge0$.

\item\label{M2} \textup{(Frozen maximal boundary)}
Let $\zeta$ evolve only in a finite set $B$, with every site
outside $B$ frozen at ${\tt 1}$.
Then $\zeta_t$ depends only on $\zeta_0$ and the marks in
$B\times(0,t]$. If $\xi$ uses the same update maps and marks
inside $B$, with arbitrary exterior evolution, then
$\xi_0\le\zeta_0$ on $B$ implies
$\xi_t\le\zeta_t$ on $B$ for all $t\ge0$.

\item\label{M3} \textup{(Time monotonicity)}
At fixed $\eps$, let $\mu_t$ be the law of the process
started from $\bar{\tt 0}$. Then $\mu_s\preceq\mu_t$ whenever
$s\le t$. In particular, the defect probability at any fixed
site is nondecreasing in time, so its supremum over $t\ge0$
equals its limit along any sequence $t_n\to\infty$.
\end{enumerate}
\end{lemma}

\begin{proof}
\ref{M1} follows directly from Proposition~\ref{prop:harris}. For \ref{M2}, the key observation is that $\zeta$
always sees the largest possible value ${\tt 1}$
outside $B$. 
Thus, if $\xi\le\zeta$ initially inside $B$, every update of $\zeta$ receives inputs that are at least as large as those of $\xi$. 
Since the update maps are monotone, this ordering
is preserved at all times. 
Moreover, because the exterior of $\zeta$ is frozen, its evolution depends only on the marks inside $B$.

For \ref{M3}, consider two processes evolving under the same update marks for a duration $t$. 
The first starts from $\bar{\tt 0}$, while the second starts from the configuration obtained by evolving $\bar{\tt 0}$ for an additional time $s$.
Since $\bar{\tt 0}$ is the smallest possible configuration, the second process initially dominates the first. By \ref{M1}, this ordering persists throughout the evolution. The resulting laws are $\mu_t$ and $\mu_{s+t}$, respectively, giving $\mu_t\preceq\mu_{s+t}$ for all $s,t\ge0$.
\end{proof}

\subsection{Finite speed of propagation}
\label{sec:sm-finite-speed}

For any local process with bounded rates, the causal influence of any spacetime point propagates at a bounded speed. We make this statement more precise below. 

\begin{lemma}[finite speed of causal influence]\label{lem:speed}
Define a \emph{causal path} from a space-time set $D\times(0,t]$ to $(x,t)$ as a finite sequence of Poisson marks $(x_1,t_1),\dots,(x_m,t_m)$ with $0<t_1<\dots<t_m\le t$, $x_1\in D+[-r,r]^d$, $\|x_{i+1}-x_i\|_\infty\le r$ for each $i$ and $\|x-x_m\|_\infty\le r$. There exist finite positive constants $v_\ast$ and $c$, depending only on $d,r,\lambda$, such that for all $\rho\ge v_\ast t$, the following holds:
\begin{equation}
\PP\bigl(\exists\ \text{causal path from }\{y:\|y-x\|_\infty\ge\rho\}
\times(0,t]\ \text{to}\ (x,t)\bigr)\ \le\ e^{-c\rho}.
\end{equation}
\end{lemma}

\begin{proof}
Let $E_{\rho,t}$ denote the event in the statement. If a causal path with $m$ marks connects the complement of the ball of radius $\rho$ to $(x,t)$, then there exists $y\in \Z^d$ satisfying $\|y-x\|_\infty\geq\rho$ such that $\|y - x_1\|_\infty\leq r$. Consequently,
\begin{equation}
\rho
\leq \|y-x_1\|_\infty
+\sum_{i=1}^{m-1}\|x_{i+1}-x_i\|_\infty
+\|x_m-x\|_\infty
\leq (m+1)r.
\end{equation}
Thus, the path contains at least $M:=\max\left\{1,\left\lceil\frac{\rho}{r}\right\rceil-1\right\}$ marks.

Let us define $K = (2r+1)^d\lambda$. Working backward from $x$, there are at most $(2r+1)^{dm}$ possible spatial sequences of length $m$, while the expected number of chronologically ordered $m$-tuples of marks along any fixed sequence is $(\lambda t)^m/m!$. A union bound, therefore, gives:
\begin{equation}
\mathbb{P}(E_{\rho,t})
\leq \sum_{m\geq M}\frac{(Kt)^m}{m!}.
\label{eq:causal_path_sum}
\end{equation}
Using $m!\geq(m/e)^m$, we obtain---whenever $M\geq4eKt$---the following:
\begin{align}
\sum_{m\geq M}\frac{(Kt)^m}{m!}
&\leq \sum_{m\geq M}\left(\frac{eKt}{m}\right)^m
\leq \sum_{m\geq M}4^{-m}
\leq 2^{-M}.
\label{eq:causal_path_tail}
\end{align}

Let us set $v_\ast=8eKr$, and $c=\frac{\log 2}{4r}$, and suppose that $\rho\geq4r$ and $\rho\geq v_\ast t$. Then $M\geq\frac{\rho}{2r}
\geq4eKt$, so that Eqs.~\eqref{eq:causal_path_sum} and \eqref{eq:causal_path_tail} imply:
\begin{equation}
\mathbb{P}(E_{\rho,t})
\leq2^{-M}
\leq\exp\left(-\frac{\log 2}{2r}\rho\right)
\leq e^{-c\rho}.
\end{equation}

It remains to consider $0<\rho<4r$. Any causal path reaching $(x,t)$ contains at least one mark within a distance $r$ of $x$. Hence,
\begin{equation}
\mathbb{P}(E_{\rho,t})
\leq 1-e^{-Kt}
\leq Kt
\leq\frac{K\rho}{v_\ast}
<\frac{1}{2},
\end{equation}
whereas $e^{-c\rho}>e^{-4cr}=1/2$. The desired inequality therefore holds in this regime as well.
\end{proof}

\section{Proof of the Stability Criterion}
\label{sec:app_B}
Here we present a rigorous construction and proof of our stability criterion. We begin by introducing some geometric ideas that will be crucial for the formulation of the stability criterion. 

\subsection{Stable Directions}
For a linear form $\ell \in (\Z^d)^{*}\backslash \{\boldsymbol{0}^{\ast}\}$, we define a half-space $H_{\ell}(c)$ and a set of neighbors $\mathcal{N}^{+}_a(\ell)$ as follows:
\begin{align}
    H_\ell(c)&=\{x:\ell(x)\ge c\}\\
    \mathcal{N}_a^+(\ell)&=\{u\in \mathcal{N}_a:\ \ell(u)\ge0\}.
\end{align}

\begin{definition}\label{def:stable}
The direction corresponding to $\ell$ is said to be \emph{stable} if $\cN_a^+(\ell)\in\cZ_a$ for every $a\in\cA$, where $\cZ_a$ denotes the set of all zero-sets corresponding to the update-type $a$.
The set of stable directions is denoted by $\cS$.
\end{definition}

\begin{lemma}[half-space invariance]\label{lem:halfspace}
Let $\ell\in\cS$ and suppose $\eta_0=0$ on $H_\ell(c)$, with arbitrary
exterior. Then, under the noiseless dynamics, $\eta_t={\tt 0}$ on $H_\ell(c)$ for
all $t\ge0$, deterministically in the marks.
\end{lemma}

\begin{proof}
Let $F_{\ell}(c)=\{\eta:\eta|_{H_\ell(c)}\equiv {\tt 0}\}$ be the set of configurations that are defect-free on $H_{\ell}(c)$. By
Proposition~\ref{prop:harris} the flow on any bounded region is an a.s.\
finite composition of local maps, and the infinite-volume flow is the
increasing limit of finite-volume flows. Hence, it suffices to show that each map $f_a$ sends $F_{\ell}(c)$ into $F_{\ell}(c)$. Fix $\eta\in F_{\ell}(c)$ and an
update of type $a$ at $x$. If $x\notin H_\ell(c)$, the half-space is
untouched. If $x\in H_\ell(c)$, then every $u\in \mathcal{N}_a^+(\ell)$ satisfies
$\ell(x+u)=\ell(x)+\ell(u)\ge c$, so $x+u\in H_\ell(c)$ and $\eta(x+u)=0$.
Since $\cN_a^+(\ell)\in\cZ_a$ (set of zero sets as defined in the main text), the output is ${\tt 0}$.
\end{proof}

Thus, if $\ell$ is stable, then $H_{\ell}(c)$ is impenetrable by defects under noiseless dynamics.

\subsection{Balanced Families}

\begin{definition}\label{def:balanced}
A \emph{balanced family} is a finite list $\ell_1,\dots,\ell_m\in\cS$
together with weights $\alpha_i>0$ such that $\sum_i\alpha_i\ell_i=0$.
\end{definition}

Intuitively, the directions $\ell_i$ are inward normals of clean half-spaces that defects cannot cross, and a droplet of defects is caught only if such half-spaces can bound it in at least one direction. This translates to the positive linear dependence $\sum_i\alpha_i\ell_i=0$.

\subsection{Spacetime Coarse-Graining}

To begin with, we fix a scale $H$ and partition $\Z^d$ into congruent coarse cells $Q_H(z)$, where $z \in \Z^d$ labels the coarse blocks. These blocks are free to be anisotropic, and thus, for generality, we take the side length along axis-$j$ as $L^{(j)}_H$. 

Independently of the block scale $H$, we fix the following objects once and for all:
\begin{enumerate}
    \item a positive integer $m$ and finite nonempty sets  (predecessor regions) of displacements on the coarse-grained lattice, $S_1, \, S_2,\, \cdots, S_m \subset \Z^d$ (zero sets of the coarse-grained model); 
    \item linear forms $\widetilde{\ell}_1, \, \widetilde{\ell}_2, \cdots, \widetilde{\ell}_m$: $\Z^d_{\rm c.g.} \to \Z$ acting on the coarse-grained lattice $\Z^d_{\rm c.g.}$, each corresponding to a stable direction scaled to the coarse-grained lattice;
    \item displacements $\delta_1, \delta_2,\cdots, \delta_m \geq 0$. 
\end{enumerate}
We further impose the following condition:
\begin{equation}
\label{eq:support}
    \widetilde{\ell}_i(s) \geq \delta_i
\end{equation}
for every $s \in S_i$. With these entities in hand, we now define a state variable $X(z) = {\tt 0},\, {\tt 1}$ for $z \in \Z^d_{\rm c.g.}$ that tells us the state of a coarse cell $z$. It is ${\tt 1}$ if there is even a single defect inside $Q_H(z)$, and {\tt 0} otherwise. We also define a map $\Phi$ that acts on configurations of $X$ on $\mathbb{Z}^d_{\rm c.g.}$ as follows:
\begin{equation}
    \Phi[X](z)\ =\ \bigwedge_{i=1}^{m}\ \Bigl(\bigvee_{s\in S_i}X(z+s)\Bigr),
    \label{eq:Phi}
\end{equation}
The above definition of $\Phi$ asserts that if all cells $z+s$ with $s \in S_i$ are clean (i.e., $X(z+s) = {\tt 0}$ for every $s \in S_i$), then an application of $\Phi$ clears the cell $z$, i.e., $\Phi[X](z) = {\tt 0}$. 

\subsection{The Spacetime Block Estimate}

\begin{definition}[block process]\label{def:clause}
For each scale $H$, choose a finite set of coarse displacements $B\subset\Z^d_{\rm c.g.}$ such that $\{0\}\cup S_i\subseteq B$ for every $i = 1,\cdots, m$. For a coarse cell $z$, we denote by
\begin{equation}
\cB_H(z):=\bigcup_{s\in B}Q_H(z+s)
\end{equation}
the corresponding set of microscopic sites in its buffer. We write the \emph{predecessor region} associated with direction $i$ as
\begin{equation}
Q_H(z+S_i)\equiv\bigcup_{s\in S_i}Q_H(z+s)
\subseteq \cB_H(z).
\end{equation}

Let $\zeta=\zeta^{H,i,z}$ be the \emph{noiseless} process run \textit{only} on $\cB_H(z)$, with all sites of $\cB_H(z)^c$ frozen
permanently at {\tt 1}, starting from
\begin{equation}
\zeta_0(x)=
\begin{cases}
{\tt 0},&x\in Q_H(z+S_i),\\
{\tt 1},&x\in \cB_H(z)\setminus Q_H(z+S_i).
\end{cases}
\end{equation}
Let $r_{H,i}$ be the probability that some site in $Q_H(z)$ is defective at $t=T_H$ under the $\zeta$ dynamics. For stability, we require the following to hold for each $i$:
\begin{equation}\label{eq:block_failure_supp}
r_{H,i}:= \PP\bigl(\exists\,x\in Q_H(z): \zeta_{T_H}(x)={\tt 1}\bigr) \longrightarrow 0
\qquad(\text{as }H\to\infty),
\end{equation}
independently of $z$ by translational invariance.
\end{definition}

A few caveats about the above definition follow. (i) $\zeta$ depends on the update marks in $\cB_H(z)\times(0,T_H]$ alone; (ii) $\zeta$ dominates every process with a smaller-or-equal initial configuration on $\cB_H(z)$, because the exterior region is frozen at {\tt 1} in its definition.

Another important requirement we set is that the coarse-graining scheme and the choice of buffer should be such that the buffer spans only a bounded number of coarse cells, independent of $H$.

\begin{definition}[scale-independent locality]
There exists a finite positive constant $R_\ast$, independent of $H$,
such that
\begin{equation}\label{eq:locality}
B\subseteq \{u\in\Z^d_{\rm c.g.}:\|u\|_\infty\le R_\ast\}.
\end{equation}
Equivalently, $\cB_H(z)\subseteq \bigcup_{\|u\|_\infty\le R_\ast}Q_H(z+u)$.
\end{definition}

The point of \eqref{eq:locality} is that it is a bound on the coarse-grained lattice, not in Euclidean distance on the microscopic lattice. This distinction is important because, with anisotropic cells, the two are different. We will elaborate on this remark further when we apply the stability criterion to the continuous-time {\sf R} process.

\subsection{The Coarse Map is an Eroder}
\label{sec:sm-coarse-eroder}
\begin{proposition}[geometric eroder]\label{prop:eroder}
Suppose there are $\alpha_i>0$ with
\begin{equation}\label{eq:E}
\sum_{i=1}^m\alpha_i\widetilde{\ell}_i=0
\quad\text{and}\quad
\Delta:=\sum_{i=1}^m\alpha_i\delta_i>0 .
\end{equation}
Then $\Phi$ of \eqref{eq:Phi} is an eroder: for every finite
$D\subseteq\Z^d$ there is $n_0(D)$ with $\Phi^{n}(D)=\varnothing$ for
$n\ge n_0(D)$.
\end{proposition}

\begin{proof}
For a nonempty finite set $D$, we set $h_i(D)=\max_{x\in D}\widetilde{\ell}_i(x)$. If $z\in\Phi(D)$, then for each $i$, there is some $s_i\in S_i$ with $z+s_i\in D$, such that
\begin{equation}
\widetilde{\ell}_i(z)=\widetilde{\ell}_i(z+s_i)-\widetilde{\ell}_i(s_i)\ \le\ h_i(D)-\delta_i,
\end{equation}
since $\widetilde{\ell}_i(s_i) \geq \delta_i$ (\eqnref{eq:support}). Hence, $h_i(\Phi^n(D))\le h_i(D)-n\delta_i$. If $z\in\Phi^n(D)$, then, using
$\sum_i\alpha_i\widetilde{\ell}_i=0$, we obtain
\begin{equation}
0=\sum_i\alpha_i\widetilde{\ell}_i(z)\ \le\ \sum_i\alpha_i h_i(D)-n\Delta,
\end{equation}
which is impossible for large $n$. Since all $\delta_i\ge 0$, condition \eqref{eq:E} requires that the balanced family contains at least one advancing direction.
\end{proof}

\subsection{Block Process and Sparse Errors}

We now discuss the consequences of $\Phi$ being an eroder. Fix a spatial scale $H$, and a temporal coarse-graining scale $T_H$, i.e., the coarse-grained time variable takes values $t_n = n T_H$ with $n \in \mathbb{Z}_{\geq 0}$. Define a variable $X$ on the coarse-grained spacetime as follows:
\begin{equation}
    X_n(z) = \begin{cases}
        {\tt 1} &\text{if } Q_H(z) \text{ contains a defect at time $t_n$},\\
        {\tt 0} &\text{otherwise}.
    \end{cases}
\end{equation}
For each block $(n,z)$ and face $i$, build the block process
$\zeta^{H,i,z}$ introduced in Definition~\ref{def:clause} from the true \emph{noiseless
update marks} in $\cB_H(z)\times(t_n,t_{n+1}]$. Let ${\sf fail}_{(n,z,i)}$ be its failure event from \eqref{eq:block_failure_supp}, and let ${\sf hostile}_{(n,z)}$ be the event that at least one \emph{hostile} error mark falls in $\bigcup_i\cB_H(z)\times(t_n,t_{n+1}]$. We define the error event as:
\begin{equation}
    {\sf error}_{n,z} = {\sf hostile}_{n,z}\ \cup \ \bigcup_{i = 1}^{m} {\sf fail}_{n,z,i}
\end{equation}

\begin{lemma}[block inequality]\label{lem:block}
Pathwise, the following holds:
\begin{equation}
X_{n+1}(z) \le \Phi(X_n)(z) \vee \mathbf{1}_{{\sf err},n}(z).
\end{equation}
\end{lemma}

\begin{proof}
Suppose $\Phi(X_n)(z)=0$ and ${\sf error}_{n,z}$ fails. We must show that $X_{n+1}(z)=0$.
By \eqref{eq:Phi} there is a face $i$ such that $X_n(z+s)=0$ for all $s\in S_i$. This means that the true configuration at time $t_n$ vanishes on $Q_H(z+S_i)$ and is $\le1$ elsewhere in the buffer. Hence, $\eta_{t_n}\le\zeta^{H,i,z}_0$ on $\cB_H(z)$. On ${\sf hostile}_{n,z}^{\,c}$, the true process inside the buffer during $(t_n,t_{n+1}]$ is driven by noiseless marks only. Thus, by the frozen-boundary domination \ref{M2}, we obtain $\eta_{t_{n+1}}\le\zeta^{H,i,z}_{T_H}$ on $\cB_H(z)$. Given that ${\sf fail}_{(n,z,i)}$ does not occur, the right side vanishes on $Q_H(z)$.
\end{proof}

Next, we bound the probability of ${\sf hostile}_{(n,z)}$. Here, we invoke the fact that the noise strength is weak, or in other words, the noise rate $\eps$ is arbitrarily small. Let $b_H=\bigl|\bigcup_i\cB_H(0)\bigr|$. Since each $r_{H,i}$ is controlled by \eqref{eq:block_failure_supp} and hostile marks occur independently at a rate of $\eps$ at each site,
\begin{equation}\label{eq:sparsity}
\PP({\sf error}_{(n,z)})\ \le\ \sum_{i=1}^m r_{H,i}
+\bigl(1-e^{-\varepsilon b_HT_H}\bigr) \ \le\ \underbrace{\sum_i r_{H,i}+\varepsilon\,b_HT_H}_{=:p_H(\varepsilon)}.
\end{equation}

By construction, ${\sf error}_{n,z}$ is measurable with respect to the marks in
$\cB_H(z)\times(t_n,t_{n+1}]$. By uniform locality, \eqref{eq:locality}, the space-time regions attached to distinct sites $(n,z)$ of the coarse-grained spacetime form an intersection graph of degree bounded by some constant depending only on $d$, $R_\ast$, and $m$. Let $K$ be a proper coloring number (chromatic number) of that graph. Any finite set $A$ of sites of the coarse-grained spacetime splits as $A=A_1\cup\dots\cup A_K$ with mutual independence within each class.

\begin{lemma}[product bound via colouring]\label{lem:hoelder}
For every finite set $A$ of coarse space-time sites,
\begin{equation}\label{eq:product}
\PP\left(\bigcap_{v\in A}{\sf error}_v\right)\ \le\ p_H(\varepsilon)^{|A|/K}.
\end{equation}
\end{lemma}

\begin{proof}
If $p_H(\varepsilon)>1$, the bound is trivial. Otherwise, split $A=A_1\cup\dots\cup A_K$ into color classes and let $j$ index the largest, so $|A_j|\ge|A|/K$ by pigeonhole. Dropping the constraints outside class $j$ can only raise the probability, the events inside a class are independent, and $p_H(\varepsilon)\le1$:
\begin{equation}
\PP\Bigl(\bigcap_{v\in A}{\sf error}_v\Bigr)\ \le\
\PP\Bigl(\bigcap_{v\in A_j}{\sf error}_v\Bigr)\ =\
\prod_{v\in A_j}\PP({\sf error}_v)\ \le\ p_H(\varepsilon)^{|A_j|}\ \le\
p_H(\varepsilon)^{|A|/K} .
\end{equation}
\end{proof}

Interestingly, for static directions, the block failure probability $r_{H,i}$ can be straightforwardly shown to vanish as the block scale approaches infinity. 

\begin{lemma}[static directions]\label{lem:static_direction}
Let $\ell_i\in\cS$ with $\delta_i=0$, and let $l^{(i)}_{\rm min}$ be the minimal
$\ell_i$-height of the core $Q=Q_H(z)$. Define $\cH_i=\{x:\ell_i(x)\ge l^{(i)}_{\rm min}\}$ so that $Q\subseteq\cH_i$. Suppose the mismatch set $F$, given by
\begin{equation}
F\ :=\ \bigl((\cH_i\cap \cB_{H}(z))\setminus Q_H(z+S_i)\bigr)\ \cup\ \cB_H^c(z)
\end{equation}
satisfies the condition $\dist_\infty(Q,F)\ge W\ge v_\ast T_H$. Then
$r_{H,i}\le|Q|\,e^{-cW}$.
\end{lemma}

\begin{proof}
Let $\hat\xi$ be the infinite-volume noiseless process, on the same marks, from {\tt 0} on $\cH_i$ and {\tt 1} elsewhere. By Lemma~\ref{lem:halfspace}, it vanishes on $\cH_i\supseteq Q$ at all times. Every site with $\zeta_0>\hat\xi_0$ lies in $F$. To see this more clearly, note that on $Q_{H}(z+S_i)$, $\zeta_0={\tt 0}$, whereas on $(\cB_H(z)\setminus Q_H(z+S_i))\setminus\cH_i$, both are {\tt 1}. Furthermore, $(\cH_i\cap \cB_H(z))\setminus Q_H(z+S_i)$ --- the only place inside $\cB_H(z)$ with $\zeta_0={\tt 1}>{\tt 0}=\hat\xi_0$ --- is in $F$ by definition. Sites of $Q_{H}(z+S_i)$ outside $\cH_i$, if any, have $\zeta_0=0\le\hat\xi_0$, and that is completely harmless. Thus, using Lemma~\ref{lem:speed}, we obtain the following: 
\begin{equation}
\PP(\zeta_{T_H}(x)={\tt 1})\le\PP(F\text{ causally affects }(x,T_H))\le e^{-cW},\qquad x\in Q.
\end{equation}
Finally, a union bound over the sites of $Q$ completes the proof.
\end{proof}

\subsection{Completion of the proof of Theorem 1}

Having developed all the necessary machinery, we now present the proof of Theorem~\ref{thm:maintext}. To begin with, note that by Proposition~\ref{prop:eroder}, $\Phi$ is an eroder. Let us set $Y_0=\bar{\tt 0}$ and $Y_{n+1}(z)=\Phi(Y_n)(z)\vee\mathbf{1}_{{\sf err},n}(z)$. Using Lemma~\ref{lem:block} and the monotonicity of $\Phi$, we obtain $X_n\le Y_n$ pathwise by induction. Using Eqs.~\eqref{eq:sparsity}--\eqref{eq:product}, the error field of the coarse-grained model satisfies the sparsity condition of Lemma~\ref{lem:main-toom} with $q=q_H(\varepsilon)=p_H(\varepsilon)^{1/K}$. We write $q_H(0)=\bigl(\sum_ir_{H,i}\bigr)^{1/K}$ for its noiseless part. At this point, it is important to fix the order of limits, as this is where the scale-independence of $\Phi$ and of the exponent $K$ becomes crucial. Let $\delta >0$ be some fixed finite constant. Since $q_H(0)\to0$ by \eqnref{eq:assumptionKey} and $\rho_\Phi(q)\to0^+$ as $q\to 0^+$,  we can choose $H_{\delta}$ such that for $H > H_\delta$, $\rho_\Phi\bigl(2q_H(0)\bigr)<\delta$. Next, we choose $\varepsilon_0\equiv\varepsilon_0(\delta,H)>0$ so small that $\varepsilon\,b_HT_H\le\sum_ir_{H,i}$ for all $\varepsilon\le\varepsilon_0$, whence
$q_H(\varepsilon)\le2^{1/K}q_H(0)\le2q_H(0)$. Lemma~\ref{lem:main-toom} then gives, for all $n$ and all $\varepsilon\le\varepsilon_0$, the following
\begin{equation}
\PP_{\bar{\tt 0}}^\varepsilon\bigl(\eta_{nT_H}(\boldsymbol{0})={\tt 1}\bigr)\ \le\
\PP\bigl(X_n(\boldsymbol{0})={\tt 1}\bigr)\ \le\
\rho_\Phi\bigl(q_H(\varepsilon)\bigr)\ <\ \delta,
\end{equation}
using $Q_H(0)\ni 0$ for the first inequality. By the time monotonicity \ref{M3}, $\sup_{t\ge0}\PP_{\bar{\tt 0}}^\varepsilon(\eta_t(\boldsymbol{0})={\tt 1})$ equals the limit along the block times $t_n \to \infty$. Hence, it is smaller than $\delta$ for all $\eps\le\eps_0$. As $\delta>0$ was arbitrary, $\limsup_{\varepsilon\to 0^+}\sup_{t\ge0} \PP_{\bar{\tt 0}}^\varepsilon(\eta_t(\boldsymbol{0})={\tt 1})={0}$, which is precisely the claim~\eqref{eq:stability_maintext}.

\subsection{Relation to Gray's stability theorem}
\label{sec:gray}

We show that Gray's shrinker condition~\cite{Gray1999} implies our stability criterion. We write $A\subseteq\mathbb Z^d$ for the set of clean sites, and let $\beta_x(A)$ and $\mu_x(A)$ be the rates at which $x$ is added to and removed from $A$, respectively. Assume that the rates are translation invariant, finite range, and weakly attractive: positive-birth events are increasing in $A$, while positive-death events are decreasing. Assume also that both homogeneous states are absorbing and that the rates are exclusive, i.e., 
\begin{equation}
    \beta_x(A\setminus\{x\})\,
    \mu_x(A\cup\{x\})=0.
\end{equation}
The associated Toom operator~\cite{Gray1999} is
\begin{equation}
    T(A)= \{x\in A:\mu_x(A)=0\} \cup \{x\notin A:\beta_x(A)>0\}.
    \label{eq:gray-operator}
\end{equation}
These assumptions make $T$ monotone, translation invariant, and finite range, with $T(\varnothing)=\varnothing$ and $T(\mathbb Z^d)=\mathbb Z^d$.

\begin{proposition}
\label{prop:gray-inclusion}
Suppose $T$ is a shrinker. There is an attractive process dominated by the original process when ordered with respect to clean sites that satisfies our stability criterion using a coarse-graining scheme involving congruent boxes $Q_H(z)$ ($z\in \Z^d_{\rm c.g.}$) of spatial block scale $H$ and block time $T_H = \tau H$, where $\tau$ is a finite positive constant independent of $H$. For attractive original rates, the same block estimates hold directly for the original process.
\end{proposition}

\begin{proof}
Finite range and translation invariance imply that $b:=\min\{\beta_0(A)>0:0\notin A\}>0$, $\sup_A\mu_0(A)=\text{finite}$. Let us choose $D>0$ with $D\ge\sup_A\mu_0(A)$, and define the following rates
\begin{equation}
\begin{split}
    \hat\beta_x(A) &=b\,\mathbf1_{\{x\in T(A)\}}, \quad x\notin A,
    \\
    \hat\mu_x(A) &=D\,\mathbf1_{\{x\notin T(A)\}}, \quad x\in A.
\end{split}
    \label{eq:gray-comparison-rates}
\end{equation}
They define an attractive process $\widehat A_t$, driven by independent birth and death clocks of rates $b$ and $D$.
For $A\subseteq B$, monotonicity of $T$ gives $\widehat\beta_x(A)\le\beta_x(B)$ for $x\notin B$, and $\widehat\mu_x(A)\ge\mu_x(B)$ for $x \in A$.

The order-preserving coupling, therefore, gives $\widehat A_t\subseteq A_t$. When the two processes share the same hostile marks, the inclusion $\widehat A_t\subseteq A_t$ is preserved, implying that it suffices to establish the block criterion for $\widehat A_t$.

\medskip
\noindent\emph{Stable and advancing directions.---}
For a nonzero rational linear form $\ell$, we write 
\begin{equation}
    H_\ell(h)=\{x\in\mathbb Z^d:\ell(x)\ge h\}.
\end{equation}
Gray's half-space formula~\cite{Gray1999} implies that there is a number $a(\ell)$ such that
\begin{equation}
    T\bigl(H_\ell(h)\bigr)=H_\ell\bigl(h-a(\ell)\bigr).
    \label{eq:gray-halfspace}
\end{equation}
Thus, $a(\ell)\ge 0$ means that the half-space does not retreat, while $a(\ell)>0$ indicates that it advances.

Since $T$ is a shrinker, there are finitely many rational directions $\ell_1,\ldots,\ell_m$ with $a_i\equiv a(\ell_i)\ge0$ such that the corresponding shifted half-spaces cover all of $\Z^d$, or equivalently
\begin{equation}
    \bigcap_{i=1}^m \{x\in\mathbb R^d:\ell_i(x)\le-a_i\} =\varnothing.
    \label{eq:gray-shrinker}
\end{equation}
Geometrically, this means that the normals $\ell_i$ surround the origin. Consequently, there are positive coefficients $c_i$ such that the following holds:
\begin{equation}
    \sum_{i=1}^m c_i\ell_i=0, \quad \sum_{i=1}^m c_i a_i>0.
    \label{eq:gray-balance-pre}
\end{equation}
Now, we choose such a relation with the smallest possible number of terms. Then we have $m\le d+1$, and every proper subset of $\{\ell_i\}_{i = 1}^{m}$ is linearly independent. Since the directions are rational, the $c_i$ may also be chosen to be rational. Rescaling $\ell_i\mapsto c_i\ell_i$ and clearing denominators, we may therefore take the linear forms $\ell_i$ to be integer-valued such that they satisfy the following:
\begin{equation}
    \sum_{i=1}^m\ell_i=0, \quad a(\ell_i)\ge0, \quad a(\ell_m)>0.
    \label{eq:gray-balanced}
\end{equation}
Here, we have relabeled the directions so that the strictly advancing one is $\ell_m$.

Finally, we note that each $H_{\ell_i}(h)$ is preserved by the noiseless continuous-time dynamics. Indeed, if $A\supseteq H_{\ell_i}(h)$, then monotonicity and \eqref{eq:gray-halfspace} imply
\begin{equation}
    T(A)\supseteq T\bigl(H_{\ell_i}(h)\bigr)
    \supseteq H_{\ell_i}(h),
\end{equation}
so no clean site in the half-space is allowed to die.

\medskip
\noindent\emph{Advance of a clean half-space.---}
Let us set $a_m=a(\ell_m)>0$. By \eqnref{eq:gray-halfspace}, the origin is made clean when the clean set is $H_{\ell_m}(a_m)$. Since the dynamics have a finite range, there is therefore a finite set $F$ in the update neighborhood such that
\begin{equation}
    \ell_m(u)\ge1
    \qquad (u\in F),
    \label{eq:gray-forward-parents}
\end{equation}
and such that, whenever all sites $x+u$, $u\in F$, are clean, a birth update at $x$ makes $x$ clean. By exclusiveness, the same condition also prevents a subsequent noiseless death at $x$. For any site $x$, we will refer to the sites $x+u$ ($u\in F$) as \emph{predecessor sites}. 

We now start the noiseless comparison process from the clean half-space $H_{\ell_m}(h)$. First, we define a smaller growing clean set $G_t$ as follows.
All sites of $H_{\ell_m}(h)$ belong to $G_0$. A site $x$ outside this half-space is added to $G_t$ once all sites $x+u$ ($u\in F$) have already been added, and the next birth clock at $x$ rings. Note that once added, it is never removed. The preceding observation gives
\begin{equation}
    G_t\subseteq \widehat A_t
    \quad\text{for all }t\ge0.
\end{equation}

Let $\tau_x$ be the time at which $x$ enters $G_t$. Then, we can write the following recursion relation for $\tau_x$:
\begin{equation}
    \tau_x=
    \begin{cases}
        0, & \ell_m(x)\ge h,\\[1mm]
        W_x+\displaystyle\max_{u\in F}\tau_{x+u},
            & \ell_m(x)<h,
    \end{cases}
    \quad
    W_x\overset{\mathrm{iid}}{\sim}\operatorname{Exp}(b),
    \label{eq:gray-arrival}
\end{equation}
where ${\rm Exp}(b)$ denotes an exponential distribution with a rate of $b$. Indeed, once all sites in $x+F$ are clean, $x$ only has to wait for its next birth clock, whose waiting time is an independent exponentially distributed variable of rate $b$.

Next, let us set $p=|F|$. If $q=h-\ell_m(x)>0,$ then every step from $x$ to one of its predecessor sites $x+u$ increases
$\ell_m$ by at least one. Hence, any chain of predecessor sites reaches the initial half-space in at most $q$ steps, and there are at most $p^q$ such
chains. Along any one chain, the total waiting time is a sum of at most $q$ independent rate-$b$ exponentials. Since $\mathbb E e^{bW_x/2}=2$, the exponential Markov inequality gives a probability of at most $2^q e^{-bt/2}$ for any fixed chain to take longer than $t$. A union bound over the (at most) $p^q$ chains, therefore, yields the following bound:
\begin{equation}
    \mathbb P(\tau_x>t)
       \le (2p)^q e^{-bt/2}.
    \label{eq:gray-passage-tail}
\end{equation}

\medskip
\noindent\emph{Block construction.---}
We now translate the half-space estimates derived above to coarse-grained dynamics. Since every proper subset of $\{\ell_i\}_{i=1}^m$
is linearly independent, we may extend
$\ell_1,\ldots,\ell_{m-1}$ to the rows of a nonsingular integer matrix $M$.
Let us choose a positive integer $D$ such that $V=DM^{-1}$ has integer entries.
For integer $H$, we define the coarse cells as:
\begin{equation}
    Q_H(z)=
    \bigl(HV(z+[0,1)^d)\bigr)\cap\Z^d,
    \qquad z\in\Z^d.
\end{equation}
These parallelepiped blocks partition $\Z^d$
into congruent lattice translates.
By construction,
\begin{equation}
    \ell_j(Vz)=Dz_j\quad(j<m),
    \qquad
    \ell_m(Vz)=-D\sum_{j<m}z_j.
\end{equation}
We may therefore define the coarse linear forms as
\begin{equation}
    \widetilde\ell_j(z)=z_j\quad(j<m),
    \qquad
    \widetilde\ell_m(z)=-\sum_{j<m}z_j.
    \label{eq:gray_coarse_forms}
\end{equation}
Thus, the sum of $\widetilde{\ell}_i$ vanishes: $\sum_i \widetilde{\ell}_i=0$. Let us further define $\cB_H(z)=\bigcup_{\|s\|_\infty\le R}Q_H(z+s)$ to be a fixed-width buffer around the target cell, where $R$ will be chosen large enough below.

For $i<m$, let $S_i$ be the collection of cells in the buffer lying on the ${\widetilde{\ell}}_i\ge 0$ side of the corresponding coarse half-space. For the advancing direction, let $S_m$ consist of those lying on the $\widetilde{\ell}_m \ge 1$ side. Writing $x=HV(s+u)$ with $u\in[0,1)^d$, we have
$\ell_j(x)=DH(s_j+u_j)$ for $j<m$ and
$\ell_m(x)=-DH\sum_{j<m}(s_j+u_j)$.
Consequently, the predecessor regions contain the
corresponding microscopic clean half-spaces within
the buffer:
\begin{equation}
\begin{aligned}
H_{\ell_i}(0)\cap\cB_H(0)
&\subseteq Q_H(S_i), && i<m,\\
H_{\ell_m}(DH)\cap\cB_H(0)
&\subseteq Q_H(S_m).
\end{aligned}
\end{equation}

Let us define
\begin{equation}
    \Phi(X)(z)
    =
    \bigwedge_{i=1}^{m}
    \bigvee_{s\in S_i}X(z+s).
    \label{eq:gray-coarse-map}
\end{equation}
The first $m-1$ directions are therefore static, while the last one advances
by one coarse cell. Since we have the following:
\begin{equation}
    \sum_{i=1}^m {\widetilde{\ell}}_i=0,
    \qquad
    \sum_{i=1}^m\delta_i=1,
    \qquad
    \delta_i=
    \begin{cases}
        0,&i<m,\\
        1,&i=m,
    \end{cases}
\end{equation}
using Prop.~\ref{prop:eroder} one can show that $\Phi$ is a monotone eroder.

We next verify the corresponding block estimates.
For a static direction $i<m$, the infinite-volume
process initialized with the clean half-space
$H_{\ell_i}(0)$ remains clean throughout that
half-space, and hence on the core $Q_H(0)$.

For the advancing direction, consider instead the
infinite-volume process initialized with the clean
half-space $H_{\ell_m}(DH)$.
Every site of the core lies at a depth of at most $DmH$
behind this half-space.
By \eqnref{eq:gray-passage-tail}, we may choose a sufficiently large constant $\tau>0$ such that, with $T_H=\tau H$, the probability that any given
core site has not become clean by time $T_H$ is at most $e^{-cH}$.

To obtain the finite-volume block estimate, let us run the same dynamics only inside $\cB_H(0)$, with the exterior held defective. The finite and infinite-half-space processes ($\zeta$ and $\hat \xi$, respectively) are related on the core $Q_H(0)$ as $\zeta_t \leq \hat{\xi}_t$ unless causal influence travels from the boundary of the buffer to $Q_H(0)$ before time $T_H$.
The shortest distance between the exterior region and the core is of the order of $RH$, while $T_H$ is $\tau H$. Hence, by Lemma~\ref{lem:speed}, choosing the fixed buffer width (in the coarse-grained lattice) $R$ sufficiently large gives
\begin{equation}
    \mathbb P\!\left(
        \text{outside influence reaches }Q_H(0)
        \text{ before }T_H
    \right)
    \le C H^d e^{-cH}.
\end{equation}
Combining the above estimates, the block failure probability $r_{H,i}$ for each stable direction $i$ in the balanced family is upper-bounded as follows:
\begin{equation}
    r_{H,i}\le C H^d e^{-cH}\quad \forall i\in \{1,\cdots,m\}.
    \label{eq:gray-block-estimate}
\end{equation}

Finally, we add hostile noise of rate $\varepsilon$. An ``error'' in the coarse-grained process can occur either because the noiseless block process fails or because a hostile mark occurs somewhere in its space-time buffer. Since
$|\cB_H|=\mathcal{O}(H^d)$ and $T_H=\mathcal{O}(H)$,
\begin{equation}
    p_H(\varepsilon)
    \le
    C_1H^d e^{-cH}
    +
    C_2\varepsilon H^{d+1}.
    \label{eq:gray-block-error}
\end{equation}
for finite positive constants $C_1$, $C_2$, and $c$. The buffer contains only a fixed number of neighboring coarse cells, so
the block errors have a dependence range that is uniformly bounded across block scales. Moreover, by construction,
\begin{equation}
    X_{n+1}(z)
    \le
    \Phi(X_n)(z)\vee \mathbf{1}_{{\sf err},n}(z).
\end{equation}
Thus, all the requirements of our stability criterion are satisfied. Finally, choosing
$H$ large and then $\varepsilon$ small gives
\begin{equation}
    \lim_{\varepsilon\to 0^{+}}
    \sup_{t\ge0}
    \mathbb P^{\bar{\tt 0}}_\varepsilon
    \bigl(\eta_t(0)={\tt 1}\bigr)=0,
\end{equation}
which proves stability.
\end{proof}

\section{Examples}
\label{sec:app_C}
Here, we apply our stability criterion to several processes as illustrations. Additionally, we show that our stability criterion does not apply to the {\sf T} automaton and the 3D {\sf R} automaton, which are observed to be unstable under asynchronous updates from numerics~\cite{LakeRo2025}.

\subsection{Continuous-time {\sf R} process}
\label{sec:sm-R-proof}
In Boolean variables, the {\sf R} automaton is characterized by:
\[
\cA = \{a_{\lor}, a_{\land}\}, \ \mathcal{N}_{\lor} = \{0, \pm e_1\},\ \mathcal{N}_{\land} = \{0, \pm e_2\}, \ f_{\lor}(\xi) = \bigvee_{\delta \in \mathcal{N}_{\lor}}\xi(\delta), \ f_{\land}(\xi) = \bigwedge_{\delta \in \mathcal{N}_{\land}}\xi(\delta).
\]
Let the rates of the $a_\lor$ and $a_\land$ updates be $\mu$ and $1$,
respectively. The situation discussed in Ref.~\cite{LakeRo2025}
corresponds to $\mu = 1$, i.e., equal rates of $a_{\lor}$ and $a_{\land}$
updates. Below, we show that the continuous-time {\sf R} process is
stable not only at the equal-rate point $\mu = 1$, but also on a finite
interval around it.

We begin by identifying the zero sets of the process. An $a_{\lor}$
update can output {\tt 0} only if all sites in $\cN_{\lor}$ are {\tt 0},
implying that $\cZ^{\rm min}_{\lor} = \{\cN_\lor\}$. On the other hand,
an $a_\land$ update outputs {\tt 0} if at least one site of
$\cN_{\land}$ is {\tt 0}, implying that
$\cZ^{\rm min}_{\land} = \{\{-e_2\}, \{0\}, \{e_2\}\}$.

The minimal zero sets allow us to identify the stable directions of the process. Let $\ell$ be a stable direction. Then, for $\cN^{+}_{\lor}(\ell) = \{u \in \cN_\lor: \ell(u) \geq 0\}$ to be a zero set, the conditions $\ell(e_1) \geq 0$ and $\ell(-e_1) \geq 0$ must be simultaneously satisfied. This gives a constraint: $\ell(e_1) = 0$. On the other hand, $\cN^{+}_{\land}(\ell)$ is trivially a zero set for any $\ell$ because $\{0\}$ is a minimal zero set and $\ell(0) = 0$ ($\geq 0$). There are only two directions that satisfy the above constraint: $\ell_{1,2} = \pm e^{\ast}_2$ (up to positive scalars). The fact that these stable directions $\ell_{1,2}$ form a balanced family is also immediately obvious since $\alpha_1 = \alpha_2 = 1$ gives $\alpha_1 \ell_1 +\alpha_2 \ell_2 = 0$. Under noiseless {\sf R}, both of these stable directions are advancing. This can be easily checked by observing the interface dynamics of a horizontal domain wall between a completely clean and a completely defective half-space.

Having verified all the geometric conditions, we now coarse grain the spacetime. 

\subsubsection{Coarse-graining and the eroder map}

For the row-by-row construction below, we use anisotropic boxes and a block time of order $H\log H$. Each stage lasts $\mathcal{O}(\log H)$ and sacrifices a lateral margin of the same order. Over $H$ stages, the accumulated margin is therefore $\mathcal{O}(H\log H)$, which can be accommodated within a fixed number of coarse cells by taking $L_H=H\lceil\log H\rceil$.

We take the spatial boxes $Q_H(z)$ to be:
\begin{equation}
\label{eq:R_box}
Q_H(z)=\bigl(z_1L_H,\;z_2H\bigr)
+\{0,\dots,L_H{-}1\}\times\{0,\dots,H{-}1\},
\quad L_H:=H\lceil\log H\rceil.
\end{equation}
We fix, once and for all, an integer $K_0 \equiv K_0(\mu)$ (pinned down below) that depends only on the rate $\mu$. The predecessor regions are chosen to be:
\begin{equation}
S_1=\{\,ue_1+e_2:\ |u|\le K_0\,\},\qquad
S_2=\{\,ue_1-e_2:\ |u|\le K_0\,\},
\end{equation}
and the coarse displacements $\delta_1$ and $\delta_2$ are both taken to be $1$. This choice is consistent with the condition~\eqref{eq:support} because $\ell_i(s) \geq \delta_i$ for every $s \in S_i$. The set $S_1$ ($S_2$) is simply a row of $2K_0+1$ coarse cells directly above (below) the core $z$.

With these assignments in hand, we can write the coarse map (as defined in \eqnref{eq:Phi}) for the {\sf R} process $\Phi_{\sf R}$, as follows:
\begin{equation}
\Phi_{\sf R}[X](z)=
\Bigl(\bigvee_{|u|\le K_0}X(z{+}ue_1{+}e_2)\Bigr)\wedge
\Bigl(\bigvee_{|u|\le K_0}X(z{+}ue_1{-}e_2)\Bigr).
\end{equation}

Since the directions $\pm e^{\ast}_2$ form a balanced family and $\Delta=\sum_i\alpha_i\delta_i=2>0$, Proposition~\ref{prop:eroder} implies that the coarse map $\Phi_{\sf R}$ is a geometric eroder.

\begin{figure}
    \centering
    \includegraphics[width=0.7\linewidth]{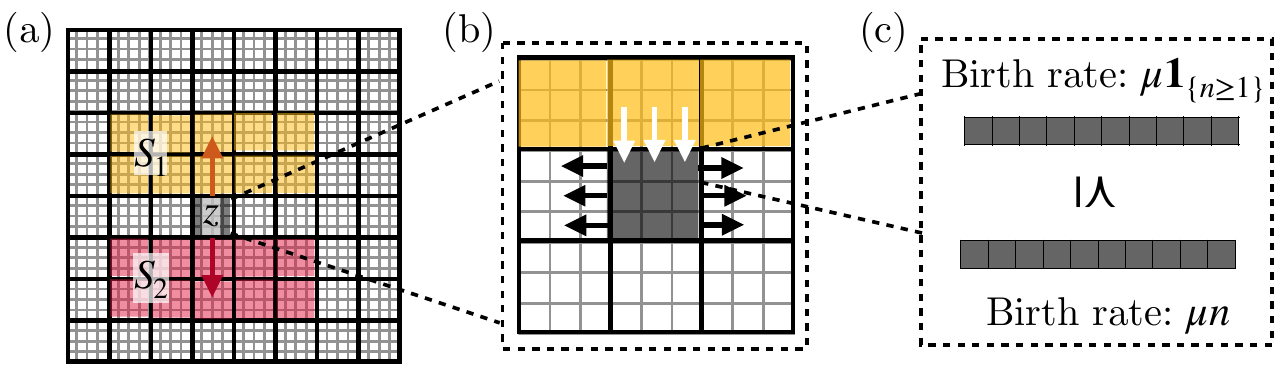}
    \caption{\textbf{Stages of coarse-graining the {\sf R} process.}
    (a) schematically illustrates the coarse neighborhoods $S_{1,2}$
    corresponding to the stable directions $\pm e^{\ast}_2$.
    Since both stable directions are advancing, we take
    $\delta_{1,2} > 0$, and the core cell starts off as defective.
    The dynamics must clear the defects in $Q_H(z)$ within block time.
    (b) and (c) pictorially depict that the dynamics of an exposed row
    of defects is stochastically dominated by a subcritical 1D contact
    process (provided $\mu < \lambda^{\rm 1D}_c$).}
    \label{fig:supp-R_summary}
\end{figure}

\subsubsection{Block estimates}
\label{sec:sm-R-block}
Here we show that the block failure probability for the {\sf R} process vanishes as $H \to \infty$ for the advancing direction $e^{\ast}_2$. 
The same estimate holds for $-e^{\ast}_2$ due to the reflection symmetry of the update rules about the horizontal axis. 
We assume that the normalized rate of OR updates is $0<\mu<\lambda^{\rm 1D}_c$, where $\lambda^{\rm 1D}_c$ is the critical per-neighbor birth rate of the standard contact process with a death rate of $1$.
A completely clean infinite row remains clean under noiseless {\sf R} dynamics: horizontal OR updates cannot create a defect in that row, and vertical AND updates cannot turn a zero into a one. Consequently, when an infinite row is initially clean (or is held fixed at zero), the row immediately below evolves autonomously as the threshold-$1$ contact process, with a death rate of $1$ and a birth rate of $\mu$ whenever at least one horizontal neighbor is defective.

Let $\Xi^{\mathrm{thr},\mathbf 1}_{\mu,t}$ denote the aforementioned threshold-$1$ contact process started from the all-one configuration. Its birth rates are bounded by those of the standard contact process ${\rm CP}_\mu$, whose birth rate is $\mu$ times the number of defective neighbors. By stochastic domination, self-duality of the standard contact process, and its subcritical exponential decay, there are $C,\gamma\in(0,\infty)$, depending only on $\mu$, such that
\begin{equation}
\label{eq:row_extinction}
 \mathbb P\bigl(\Xi^{\mathrm{thr},\mathbf 1}_{\mu,t}(x)={\tt 1}\bigr)
 \leq \mathbb P\bigl(x\in {\rm CP}^{\mathbb Z}_{\mu,t}\bigr)
 =\mathbb P\bigl({\rm CP}^{\{x\}}_{\mu,t}\neq\varnothing\bigr)
 \leq C e^{-\gamma t}.
\end{equation}
Here, the survival probability corresponds to a process started from a \emph{single} defective site, not from the all-one configuration; see, for example, Ref.~\cite{Liggett1999}. For every finite interval $D\subset\mathbb Z$, a union bound gives
\begin{equation}
\label{eq:R_finite_window_decay}
 \mathbb P\bigl(\exists x\in D:
       \Xi^{\mathrm{thr},\mathbf 1}_{\mu,t}(x)=1\bigr)
 \leq C|D|e^{-\gamma t}.
\end{equation}
In the following, we will apply this estimate to an auxiliary contact process while controlling the possibility that the neighboring clean row becomes defective. 

For an integer interval $D=[a,b]\cap\mathbb Z$ and an integer $n\geq1$, write $D\ominus n=[a+n,b-n]\cap\mathbb Z$. Consider the OR-only process $U$, which is (i) initially {\tt 0} on $D$ and {\tt 1} on $D^c$, (ii) uses rate-$\mu$ horizontal OR marks, and (iii) ignores all AND marks. This process is nondecreasing in time. On a given side, the penetration of defects up to a depth $n$ requires ordered OR marks at the first $n$ successive sites of $D$.
The waiting times for the corresponding pure-birth front are independent $\mathrm{Exp}(\mu)$ variables. Thus, for either fixed side,
\begin{equation}
 \mathbb P(\text{penetration reaches depth }n\text{ within time }u)
 \leq \mathbb P\bigl(\mathrm{Poisson}(\mu u)\geq n\bigr)
 \leq \frac{(\mu u)^n}{n!}
 \leq \left(\frac{e\mu u}{n}\right)^n.
\end{equation}
Taking a union bound over the two sides (left and right) gives
\begin{equation}
\label{eq:R_transverse_bound}
 \mathbb P\bigl(\exists s\in[0,u],\ x\in D\ominus n: U_s(x)=1\bigr) \leq 2\left(\frac{e\mu u}{n}\right)^n.
\end{equation}
Whenever an actual row is clean on the interval $D$ at the initial time, it is dominated by $U$ under the same OR update marks, regardless of the other rows, because an AND update never creates a defect. For any deterministic restart time $s$, the independence of future Poisson marks ensures that the penetration bound~\eqref{eq:R_transverse_bound} remains valid conditionally on the history up to $s$.

By translational invariance, we may take $z=0$. The set of microscopic sites in the core block is  $Q_H(0)=\{0,\ldots,L_H-1\}\times\{0,\ldots,H-1\}$.
Let $\zeta$ be the maximal-boundary noiseless comparison process in $B_H(0)=\bigcup_{\|u\|_\infty\leq K_0+1}Q_H(u)$ with its exterior frozen at {\tt 1}. Initially, it is {\tt 0} on $\bigcup_{s\in S_1}Q_H(s)$ and {\tt 1} elsewhere in the buffer. Write $\zeta_t^{(r)}(x)=\zeta_t(x,H-r)$ for $r=0,\ldots,H$ the row index; row $0$ is immediately above the core, while rows $1,\ldots,H$ are the core rows.
As discussed in the main text, we define cascaded intervals $\{I_k\}_{k = 0}^{H}$ where $I_0$ is wider than the horizontal extent of the core cell, while $I_H$ coincides with it. We divide the block time into $H$ shorter durations of length $\alpha$, which we take to be $\left\lceil\frac{3} {\gamma}\log(HL_H)\right\rceil$ with $\gamma$ the decay rate that appears in \eqnref{eq:R_finite_window_decay}. Set $t_k=k\alpha$ for $k=0,\ldots,H$, so that stage $k$ runs from $t_k$ to $t_{k+1}$ for $k=0,\ldots,H-1$.
We define the intervals $I_k$ concretely as follows:
\begin{equation}
\label{eq:R_shrinking_intervals}
 I_k=[-(H-k)m,\ L_H-1+(H-k)m]\cap\mathbb Z,
 \qquad k=0,\ldots,H,
\end{equation}
which satisfies $I_{k+1}=I_k\ominus m,$ with $I_H=\{0,\ldots,L_H-1\},$ and $|I_0|=L_H+2Hm$. Here, $m$ is the lateral margin that is sacrificed at each stage (of duration $\alpha$) of the block process. Concretely, we take $m$ to be $2\left\lceil{C_2\alpha}/{2}\right\rceil$, where $C_2$ is a constant depending on $\mu$ and $\gamma$ that we fix subsequently. 

The block time, by construction, is $T_H = H\alpha$, which can be expressed as $\tilde{\tau}(H) H \log H$ where the function $\tilde{\tau}$ remains uniformly bounded as $H \to \infty$. More precisely, it is straightforward to see that $\lim_{H \rightarrow \infty} \tilde{\tau}(H) = 6/\gamma$, and thus $T_H = \mathcal{O}(H \log H)$.  

The horizontal projection of the initially clean coarse row is $[-K_0L_H,(K_0+1)L_H-1]\cap\mathbb Z$, so $I_0\times\{H\}$ lies in the initial clean set whenever $Hm\leq K_0L_H$. 
Since $Hm/L_H = m/\lceil\log H\rceil$ approaches $6C_2/\gamma$ as $H \to \infty$,  we may fix $K_0$ to be any integer larger than $6C_2/\gamma+1$, so that $K_0 L_H > Hm$ holds for every $H>H_0(\mu)$, where 
$H_0(\mu)$ is a sufficiently large lower cutoff for the block scale. Considering block scales larger than $H_0(\mu)$ ensures that $I_0\times \{H \}$ lies entirely within the initially clean predecessor region. Consequently, we have 
\begin{equation}
\label{eq:R_width_bound}
 |I_0|\leq(2K_0+1)L_H.
\end{equation}

Let $\mathcal F_t$ denote the history of the process up to time $t$, including the initial configuration and all update marks up to time $t$.
Let us define
\begin{equation}
\label{eq:R_induction_events}
 \begin{split}
 {\sf B}_k&=\{\zeta_{t_k}^{(r)}(x)={\tt 0}
       \text{ for all }0\leq r\leq k\text{ and }x\in I_k\},\\
 {\sf A}_k&=\bigcap_{j=0}^k {\sf B}_j.
 \end{split}
\end{equation}
Here ${\sf B}_k$ requires the first $k+1$ rows to be clean on the interval $I_k$ at time $t_k$, while ${\sf A}_k$ records the success of all stages up to $k$. The events ${\sf A}_k$ are nested and measurable with respect to $\mathcal F_{t_k}$. By construction, ${\sf A}_0$ holds surely, and ${\sf A}_H$ ensures that the entire core is clean at time $T_H$.

Let us fix $k\in\{0,\ldots,H-1\}$ and define $J_k=I_k\ominus q$ and $I_{k+1}=J_k\ominus q$, where $I\ominus q$ denotes the interval obtained from $I$ by removing $q$ sites from each end, and $q \equiv m/2$ (taking the even rounding of $m$ without any loss of generality). In particular, taking $C_2=2e^2\mu+2\gamma+2$ ensures that 
\begin{equation}
\label{eq:R_margin_bounds}
 q\geq(e^2\mu+\gamma+1)\alpha,
 \quad
 2\left(\frac{e\mu\alpha}{q}\right)^q\leq2e^{-q},
 \quad
 e^{-q}\leq e^{-\gamma\alpha}\leq(HL_H)^{-3}.
\end{equation} 
Together with \eqnref{eq:R_transverse_bound}, these inequalities [\eqnref{eq:R_margin_bounds}] show that defects penetrate a distance $q$ into a clean interval from either end during time $\alpha$ with a probability of at most $2(HL_H)^{-3}$. Thus, our choice of $C_2$ makes lateral penetration sufficiently unlikely that summing over all rows and stages still gives a bound that vanishes as $H\to\infty$.

For each previous row $r=0,\ldots,k$, start an OR-only process $U^{r,k}$ at time $t_k$ that is initially {\tt 0} on $I_k$ and {\tt 1} outside, which uses the OR marks on that row. Let ${\sf G}_k$ be the event that all these processes are zero on $J_k$ at $t_{k+1}$. 
Since they are nondecreasing, on the event ${\sf A}_k\cap{\sf G}_k$ every old row remains clean on $J_k$ throughout the interval $[t_k,t_{k+1}]$.
Eqs.~\eqref{eq:R_transverse_bound} and~\eqref{eq:R_margin_bounds} imply the following:
\begin{equation}
\label{eq:R_old_rows}
 \mathbb P({\sf G}_k^c\mid\mathcal F_{t_k})
 \leq 2(k+1)e^{-q}.
\end{equation}

At time $t_k$, introduce two auxiliary processes on row $k+1$. The first, $\xi^k$, is a threshold-1 contact process on the entire row, started from the all-one configuration. It uses the same horizontal OR marks as the actual process, but treats every AND mark as an unconditional reset to ${\tt 0}$, regardless of the neighboring rows. Its clocks are extended independently beyond the finite buffer. The second, ${\rm Pen}^k$, starts from zero on $J_k$ and one on $J_k^c$, uses the same OR marks, and ignores all AND marks. While $\xi^k$ controls defect survival within the exposed row, ${\rm Pen}^k$ bounds defect penetration from its lateral boundaries.

On the event ${\sf A}_k\cap{\sf G}_k$, row $k$ remains clean on $J_k$ throughout $[t_k,t_{k+1}]$. We claim that the following pathwise comparison holds for every $t$ in this interval:
\begin{equation}
\label{eq:R_aux_process}
\zeta_t^{(k+1)}\leq\xi_t^k\lor{\rm Pen}_t^k.
\end{equation}
The inequality holds initially because $\xi_{t_k}^k\equiv{\tt 1}$, and it holds outside $J_k$ at all times because ${\rm Pen}_t^k\equiv{\tt 1}$ there. At each horizontal OR update mark inside $J_k$, the inequality is preserved because all three processes apply the same OR update rule.
At an AND update, the clean neighboring row resets the actual state to ${\tt 0}$, while $\xi^k$ is also reset to ${\tt 0}$ and ${\rm Pen}^k$ remains unchanged.
Thus, each update preserves the comparison. Since only finitely many marks occur in $J_k$ during the time interval, induction over these marks establishes \eqnref{eq:R_aux_process}.

Define the events that either auxiliary process leaves a defect in the target interval $I_{k+1}$ at time $t_{k+1}$:
\begin{equation}
\begin{split}
{\sf C}_k &=\{\exists x\in I_{k+1}: \xi_{t_{k+1}}^k(x)={\tt 1}\},\\
{\sf P}_k &=\{\exists x\in I_{k+1}: {\rm Pen}_{t_{k+1}}^k(x)= {\tt 1}\}.
\end{split}
\end{equation}
Since both auxiliary processes start from deterministic initial states at $t_k$, and future Poisson marks are independent of $\mathcal F_{t_k}$, the contact-process decay and penetration bounds give
\begin{equation}
\label{eq:R_new_row_errors}
\mathbb P({\sf C}_k\mid\mathcal F_{t_k})
\leq C|I_{k+1}|e^{-\gamma\alpha},
\quad
\mathbb P({\sf P}_k\mid\mathcal F_{t_k})
\leq 2e^{-q}.
\end{equation}
Importantly, these bounds do not require conditioning on ${\sf G}_k$. 
The subsequent argument uses only a union bound, so no independence between ${\sf C}_k$, ${\sf P}_k$, and ${\sf G}_k$ is needed.

On the event ${\sf A}_k\cap{\sf G}_k\cap{\sf C}_k^c\cap{\sf P}_k^c$, all previously cleared rows remain clean on $I_{k+1}\subset J_k$, while the comparison~\eqref{eq:R_aux_process} guarantees that row $k+1$ is also clean there at time $t_{k+1}$.
Thus, the failure of stage $k+1$, given the success through stage $k$, requires at least one of the events ${\sf G}_k^c$, ${\sf C}_k$, or ${\sf P}_k$ to occur:
\begin{equation}
 {\sf A}_k\cap{\sf A}_{k+1}^c \subseteq {\sf A}_k\cap ({\sf G}_k^c\cup{\sf C}_k\cup{\sf P}_k).
\end{equation}
Applying the preceding probability bounds and a union bound yields:
\begin{equation}
\label{eq:R_one_step}
 \mathbb P({\sf A}_k\cap{\sf A}_{k+1}^c\mid\mathcal F_{t_k})
 \leq\mathbf1_{{\sf A}_k}
 \left[C|I_{k+1}|e^{-\gamma\alpha}+2(k+2)e^{-q}\right].
\end{equation}

Since the events ${\sf A}_k$ are nested and ${\sf A}_0$ holds surely, failure to clear the core must occur at some stage of the induction. 
Summing the one-step bounds gives
\begin{align}
 r_{H,1}
 &=\mathbb P\bigl(\exists x\in Q_H(0):
 \zeta_{T_H}(x)={\tt 1}\bigr)
 \leq\mathbb P({\sf A}_H^c)\nonumber\\
 &=\sum_{k=0}^{H-1}
       \mathbb P({\sf A}_k\cap{\sf A}_{k+1}^c)\nonumber\\
 &\leq CH|I_0|e^{-\gamma\alpha}
       +H(H+3)e^{-q}\nonumber\\
 &\leq C'_{\mu}HL_H
             \bigl(e^{-\gamma\alpha}+e^{-m/2}\bigr)
 =\mathcal O\bigl((HL_H)^{-2}\bigr).
 \label{eq:R_r_H1_final}
\end{align}
Here we used the fact that $L_H\geq H$, together with \eqnref{eq:R_width_bound} and \eqnref{eq:R_margin_bounds}. The term proportional to $e^{-m/2}$ accounts for lateral defect penetration and the possibility that previously cleared rows become defective again, which is controlled by \eqnref{eq:R_margin_bounds}. Reflection under $(x_1,x_2)\mapsto(x_1,H-1-x_2)$ preserves the core and the buffer while exchanging the two predecessor strips, giving the same bound for $r_{H,2}$.

Finally, although the auxiliary processes extend beyond the buffer, the failure event itself is defined entirely by the finite-buffer process $\zeta$ and depends only on Poisson marks in $\cB_H(0)\times[0,T_H]$. Since the buffer contains exactly $(2K_0+3)^2$ coarse cells, its extent in coarse-lattice units remains independent of $H$.

The buffer contains $b_H=(2K_0+3)^2HL_H$ microscopic sites, and the block time is $T_H=H\alpha$. Hostile noise at a rate of $\eps$ contributes at most $\varepsilon b_HT_H$ to the block-failure probability, where we have chosen $\eps \ll (b_H T_H)^{-1}$. Combining this contribution with the noiseless block estimates, we obtain the noisy block failure probability as follows:
\begin{equation}
p_H(\varepsilon)
\leq 2C_\mu(HL_H)^{-2} +\varepsilon(2K_0+3)^2HL_H\,H\alpha,
\end{equation}
where $C_\mu$ is independent of $H$.
For any fixed $0<\mu<\lambda_c^{\rm 1D}$, we first choose $H$ sufficiently large to suppress the noiseless block failures and then choose $\varepsilon$ sufficiently small at that fixed scale. Theorem~\ref{thm:maintext} therefore establishes the stability of the $\bar{\tt 0}$ configuration.

To establish the stability of the $\bar{\tt 1}$ configuration, we complement all the bits and exchange the coordinate axes. The transformed process has a rate $1$ for OR updates, and a rate $\mu$ for AND updates. Rescaling time by $t'=\mu t$ restores a unit rate for AND updates, while transforming the rate of OR updates and hostile noise to $1/\mu$ and $\eps/\mu$, respectively.  Since $\mu>0$ is fixed, the rescaled noise rate still vanishes as $\eps\to0^+$. The same block construction, with the coordinate axes exchanged, therefore establishes stability of the $\bar{\tt 1}$ configuration whenever $1/\mu<\lambda_c^{\rm 1D}$.

Combining the conditions for both configurations, we conclude that the 2D {\sf R} process is bistable for every fixed ratio of rates $\mu$ satisfying
\begin{equation}
(\lambda_c^{\rm 1D})^{-1}<\mu<\lambda_c^{\rm 1D}.
\end{equation}
Since $\lambda_c^{\rm 1D}>1$~\cite{Sudbury2001}, this bistability interval includes the equal-rate case $\mu=1$, studied in Ref.~\cite{LakeRo2025}.

\subsection{2D Cross-Gated Process}
\label{sec:sm-2dCG-proof}
Here, we present a rigorous proof of the bistability of the CG process. The process lives on $\Z^2$ with one binary variable on each site. The two update rules that make up the CG process are given by:
\begin{align}
    {\sf A}_x&: \eta(x) \mapsto \eta({x+e_1}) \lor \left(\eta({x+e_2}) \land \eta({x-e_2})\right)\quad (\text{rate }a)\\
     {\sf B}_x&: \eta(x)\mapsto \eta(x+e_2) \land \left(\eta({x+e_1}) \lor \eta({x-e_1})\right)\quad (\text{rate }b)
\end{align}
where $a$ and $b$ are finite positive rates. Since the update rules are made up of logical OR and AND operations, they are Boolean monotone by construction. It is also easy to check that both homogeneous configurations $\bar{\tt 0}$ and $\bar{\tt 1}$ are traps (absorbing states) of both update rules. The erosion mechanism of the continuous-time CG process can be seen in Fig.~\ref{fig:CG_snapshots}, which shows snapshots of the noisy dynamics starting from a configuration containing a small (compared to the system size) defect set, i.e., the set of {\tt 1}s, in the background of {\tt 0}s. Our goal here is to explain this noise-robustness using the stability mechanism formulated in this work. 

As a first step, we determine the stable directions of the dynamics. Let $\ell_{u,v} = u e^{\ast}_1 + v e^{\ast}_2$ be some nonzero linear form parametrized by $u,v \in \Z$. Define a half-plane corresponding to $\ell_{u,v}$ as $\cH_{u,v}(c) = \{x\in\Z^2: \ell_{u,v}(x) \geq c\}$. For $\ell_{u,v}$ to be a stable direction, we require that under noiseless CG dynamics, every $x \in \cH_{u,v}(c)$ be clean, provided that $\cH_{u,v}(c)$ was initially clean. The process ${\sf A}_x$ provides the following constraint: $\eta({x+e_1}) = {\tt 0}$ and $\eta(x+e_2) \land \eta(x-e_2) = {\tt 0}$, implying that $u \geq 0$.  Similarly, ${\sf B}_x$ provides the following constraint: either $\eta({x+e_2}) = {\tt 0}$ or $\eta(x+e_1) \lor \eta(x-e_1) = {\tt 0}$, implying that $u = 0$ or $v \geq 0$. Therefore, the complete set of stable directions is given by:
\begin{equation}
    \cS  = \{(u,v) \neq (0,0): u \geq 0, \,v \geq 0\} \cup \{(0,v): v < 0\}. 
\end{equation}
The balanced family of stable directions is given by $\pm e^{\ast}_2$. For the $\bar{\tt 1}$ configuration, the balanced family can be similarly worked out to be $\pm e^{\ast}_1$. 

Since all the rates are bounded and the interactions are finite-range, there exists a positive constant $v_{\ast}$ that depends only on the rates, the interaction range, and the dimension, such that it sets the speed of causal influence (see Lemma~\ref{lem:speed}). 

\begin{figure}
    \centering
    \includegraphics[width=\linewidth]{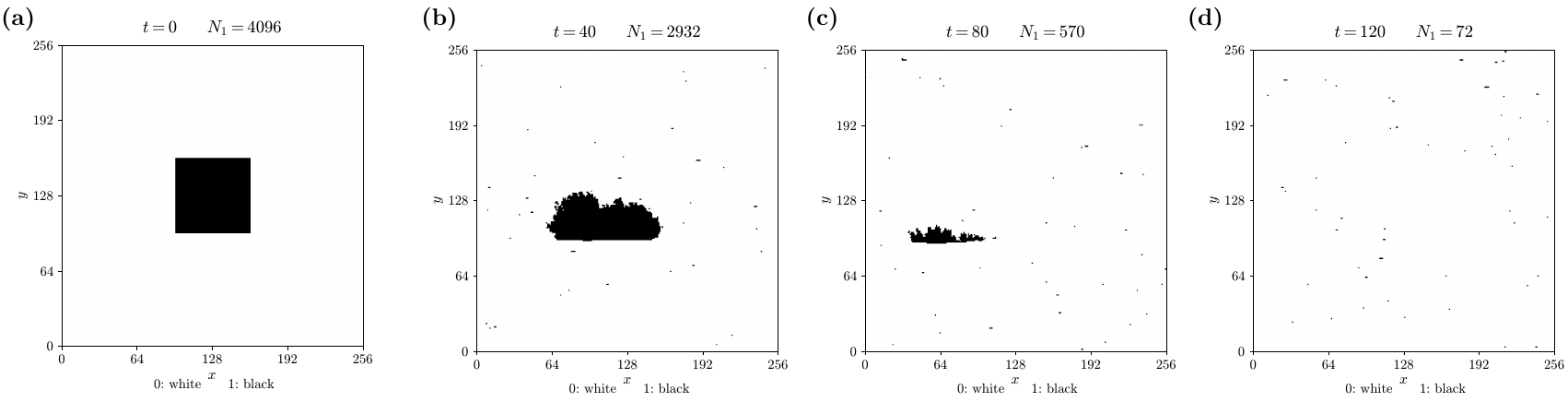}
    \caption{\textbf{Erosion of a finite defect droplet under noisy CG dynamics.} \textbf{(a--d)} Snapshots of the asynchronous CG dynamics (with equal update rates) of a finite system of linear size $L=256$ starting from a configuration containing a defect set of size $64 \times 64$ in a background of zeros. By $N_1$, we denote the number of defects ({\tt 1}s) remaining in the system. One-sided hostile noise acts independently on each site at rate $\eps=10^{-3}$.}
    \label{fig:CG_snapshots}
\end{figure}

\subsubsection{Block Estimates}
Next, to obtain the block estimates for the noiseless dynamics, we define (for each $i \in \{1,2\}$) a finite-volume adversarial process $\zeta^i_t$ that dominates the true microscopic dynamics $\eta_t$ when restricted to the buffer of any core cell $z$ and is evolved by the same update marks. We use the isotropic coarse-graining of Sec.~\ref{sec:main-2dCG}, with cells of side length $H$, block time $T_H=\tau H$, and predecessor sets $S_1,S_2$ defined in Eqs.~\ref{eq:main-2dCG-S1}--\ref{eq:main-2dCG-S2}.

By translational invariance, it suffices to fix $z = \boldsymbol{0}$ as the core cell and run block process $\zeta^i_t$ in $\cB_H(\boldsymbol{0})$. All binary variables outside the buffer are frozen to ${\tt 1}$ and the block process is initialized within the buffer as follows:
\begin{equation}
    \zeta^i_0(x) = \begin{cases}
        {\tt 0} &\text{if }x \in Q_H(z+S_i),\\
        {\tt 1} &\text{otherwise}.
    \end{cases}
\end{equation}

Since $\delta_2 = 0$, the initialization assumes that the core is fully clean at the beginning. The only way a defect might occur in the core at block time $T_H$ is through a long causal path to the core starting somewhere outside the buffer. The shortest possible length of such a path is at least $RH $. Choosing $R H > v_{\ast }T_H$, the probability of such long causal paths is super-polynomially suppressed in $H$, due to Lemmas~\ref{lem:speed} and~\ref{lem:static_direction}.

Since $i = 1$ is an advancing direction, bounding its block failure probability is slightly more cumbersome. As in the previous examples, we begin by defining the half-plane $\cH_1 = \{x \in \Z^2: x_2 > H-1\}$, and an auxiliary infinite-volume process $\hat{\xi}_t$ initialized as follows:
\begin{equation}
    \hat{\xi}_0(x) = \begin{cases}
        {\tt 0} &\text{if }x\in \cH_1,\\
        {\tt 1} &\text{otherwise}.
    \end{cases}
\end{equation}
By Lemma~\ref{lem:halfspace}, we know that under noiseless dynamics, $\cH_1$ will remain clean forever. We now show that under $\hat{\xi}_t$, all sites of the core cell $Q_H(\boldsymbol{0})$ are clean at block time $T_H$ with high probability. 

Let $p_1(t)$ denote the density of defects on the top row of the core cell, i.e., the row exposed to the clean sites of $\cH_1$. By translational invariance, we can write $p_1(t) = \EE(\hat{\xi}_t(x))$ for any $x$ in the first row. Using the classical master equation, we can write:
\begin{equation}
\begin{split}
    \frac{dp_1(t)}{dt} &= a \left(\EE({\sf A}_x(\hat{\xi}_t)) - \EE(\hat{\xi}_t(x)) \right) + b\left(\EE({\sf B}_x(\hat{\xi}_t)) - \EE(\hat{\xi}_t(x)) \right)\\
    &= a\left(\EE(\hat{\xi}_t(x+e_1)) - \EE(\hat{\xi}_t(x))\right) + b\left(0 - \EE(\hat{\xi}_t(x))\right)\\
    & = -b p_1(t),
\end{split}
\end{equation}
where we used translational invariance,  $\EE(\hat{\xi}_{t}(x+e_1)) = p_1(t)$, to obtain the last equality. Therefore, $p_1(t) = e^{-bt}$. Following the same approach, we can derive a more general result that would help us determine the block failure probability.

Let $k$ denote the depth into the defective region measured from the clean interface of $\cH_1$. For example, $k=0$ corresponds to the bottom-most row of $\cH_1$ and $k =1$ corresponds to the first defective row immediately below $\cH_1$. The core is contained within $k = 1,\cdots, H$. Let $j$ denote a coordinate transverse to $k$, which we simply take to be the $x$-coordinate, so that every point in the defective half-plane can be labeled by $(j,k)$. Set $p_0(t)=0$. For each $k\geq 1$, let $p_k(t)$ denote the defect density at that depth. Once again, using the master equation, we can write the following:
\begin{equation}
\label{eq:CG_defect_density}
    \frac{d p_k(t)}{dt}= a \left(\EE[{\sf A}_x(\hat{\xi}_t)] - p_k(t)\right) + b \left(\EE[{\sf B}_x(\hat{\xi}_t)] - p_k(t)\right).
\end{equation}
Note that ${\sf A}_x(\hat{\xi}_t) \leq \hat{\xi}_t(x+e_1) +\hat{\xi}_t(x+e_2)$ and ${\sf B}_x(\hat{\xi}_t) \leq \hat{\xi}_t(x+e_2)$, because $\sigma_1 \lor \sigma_2 \leq \sigma_1 +\sigma_2$ and $\sigma_1 \land \sigma_2 \leq \sigma_1$ (or $\sigma_2$) for any two binary variables $\sigma_{1,2}$. Thus, we can simplify \eqnref{eq:CG_defect_density} as follows:
\begin{equation}
\begin{split}
    \frac{dp_k(t)}{dt} &\leq a \left(p_k(t) + p_{k-1}(t) - p_k(t)\right) + b\left(p_{k-1}(t) - p_k(t)\right)\\
    &\leq -b p_k(t) + \lambda p_{k-1}(t),\quad \text{where }\lambda \equiv a + b. 
\end{split}
\end{equation}

Let us consider an ansatz $v_k (t) = \exp\left(\theta k - \gamma t/2\right)$, where $\gamma:= \min(a,b)$, and $\theta := \log (2\lambda/\gamma)$. We claim that $v_k(t)$ dominates $p_k(t)$ at all times. To see this, first note that this claim is trivially true for both $k = 0$ and $t  = 0$. Secondly, the time derivative of $v_k(t)$ is $-\gamma v_k/2$, which is greater than or equal to $-b v_k +\lambda v_{k-1}$. With these ingredients in hand, we now inductively prove the claim that $p_k(t) \leq v_k(t)$ for every $t \geq 0$. 

Let $p_{k}(t) \leq v_{k}(t)$ for some $k \ge 0$, and every $t \geq 0$. We want to show that under this induction hypothesis, $p_{k+1}(t) \leq v_{k+1}(t)$ for every $t \geq 0$. To this end, let us define $w_{k+1}(t) = p_{k+1}(t) - v_{k+1}(t)$. Its time derivative satisfies:
\begin{equation}
\begin{split}
    \frac{d w_{k+1}(t)}{dt} &\leq (-b p_{k+1}(t) + \lambda p_{k}(t)) - (-\gamma v_{k+1}(t)/2) \leq (-b p_{k+1}(t) + \lambda p_{k}(t)) - (-b v_{k+1}(t) + \lambda v_{k}(t)) \\
    &= - b w_{k+1}(t) + \lambda(p_{k}(t) - v_{k}(t))   \leq -b w_{k+1}(t),
\end{split}
\end{equation}
where, to obtain the final inequality, we used the induction hypothesis $p_k(t) - v_k(t) \leq 0$ for every time $t \geq 0$. Integrating both sides of the above inequality and using the fact that $w_{k+1}(0) = 1-e^{\theta (k+1)} < 0$, we obtain:
\begin{equation}
    e^{bt} w_{k+1}(t) \leq w_{k+1}(0) < 0 \implies w_{k+1}(t) \leq 0 \ \forall\, t\geq 0.
\end{equation}
This completes the proof of the claim that $p_k(t) \leq v_k(t)$ at all times. Choosing $\tau$ such that $\gamma \tau >2\theta$, we find the following for any core site:
\begin{equation}
    p_k(T_H) \leq e^{\theta H - \gamma T_H/2} = e^{-cH} \quad (c > 0),
\end{equation}
where we used the fact that for any core site $k\leq H$. The above result tells us that for each core site $x$, the probability that it is defective is upper bounded by $\exp(-cH)$. Taking a union bound over all $H^2$ core sites gives us the following:
\begin{equation}
\label{eq:CG_inf_process}
    \PP(\exists x\in Q_H(\boldsymbol{0}): \hat{\xi}_{T_H}(x) = {\tt 1}) \leq H^2 e^{-cH}.
\end{equation}

Finally, note that the only way $\zeta^1_t$ might be defective while $\hat{\xi}_t$ might be clean at some core site is via a long causal path starting from $\cB^c_H(\boldsymbol{0})$ of length at least $RH$. We know from Lemma~\ref{lem:speed} that for $R H> v_{\ast} T_H$ (as we had assumed above), the probability of such causal influence is superpolynomially suppressed in $H$. Combining \eqnref{eq:CG_inf_process} with the bound on the probability of long causal paths, we obtain the desired block estimate: $r_{H,1} \to 0$ as $H \to \infty$. 

Since $R$ is a scale-independent finite number, the sparseness of the error field of the coarse eroder immediately follows. Therefore, all conditions required by Theorem~\ref{thm:maintext} are satisfied by the CG process, implying that the $\bar{\tt 0}$ state is stable. Note that the 2D CG process is symmetric under the bit-flip combined with the exchange of $e_1$ and $e_2$ and the rates $a$ and $b$. Since our stability proof works for arbitrary $a,b>0$, we can immediately conclude the stability of the $\bar{\tt 1}$ configuration.

\subsection{3D Cross-Gated process}
\label{sec:sm-3dCG-proof}

We establish the block estimates for the process defined in Eqs.~\eqref{eq:3d-CG-A}--\eqref{eq:3d-CG-B} of the main text, with the updates ${\sf A}^{(3)}$ and ${\sf B}^{(3)}$ occurring independently at rates $a$ and $b$, respectively, at every site. We first derive a sufficient condition on $a/b$ for the stability of $\bar{\tt 0}$. A spin flip combined with $e_1\leftrightarrow e_3$ maps $\bar{\tt 1}$ to $\bar{\tt 0}$ and exchanges the rates $a$ and $b$. Thus, the result for $\bar{\tt 1}$ then follows by interchanging $a$ and $b$.

To coarse-grain the spacetime, we use isotropic cells $Q_H(z)=Hz+\{0,\ldots,H-1\}^3$ and block time $T_H=\tau H$,  where $z\in\Z^3_{\rm c.g.}$ and $\tau>0$ will be fixed subsequently. We take the buffer of a cell $z$ to be:
\begin{equation}
\mathcal B_H(z)=\bigcup_{\|s\|_\infty\le R}Q_H(z+s),
\end{equation}
where $R$ is a scale-independent integer to be chosen after fixing $\tau$. For the stable directions
\begin{align}
\ell_1&=e_1^*,\quad \ell_2=-e_1^*,
\\
\ell_\sigma&=\sigma_2e_2^*+\sigma_3e_3^*,
\quad \sigma=(\sigma_2,\sigma_3)\in\{\pm1\}^2,
\end{align}
we take the predecessor regions as follows:
\begin{align}
S_i&=\{s:\|s\|_\infty\le R,\ \ell_i(s)\ge0\}, \quad i=1,2,
\\
S_\sigma&=\{s:\|s\|_\infty\le R,\ \ell_\sigma(s)\ge1\}.
\label{eq:sm-CG3-predecessors}
\end{align}
The corresponding displacements are $\delta_1=\delta_2=0$ and $\delta_\sigma=1$, as in the main text. For each of these six stable directions, let $\zeta^{H,i,z}$ be the noiseless process on $\mathcal B_H(z)$ with its exterior frozen to ${\tt 1}$. Within the buffer, the process is initialized to ${\tt 0}$ on $\bigcup_{s\in S_i}Q_H(z+s)$ and ${\tt 1}$ elsewhere in the buffer. Translational invariance allows us to set $z=0$ throughout the proof.

Let $C_t=\{j:\chi_t(j)={\tt 1}\}$ be the defect set of the auxiliary line process defined in \eqnref{eq:maintext-3dCG-aux}. Its OR and reset updates inherit the rates $a$ and $b$ of $A^{(3)}$ and $B^{(3)}$, respectively. A clean site with $n$ defective neighbors becomes defective at rate $a\mathbf{1}_{\{n\ge1\}}\le an$, while each defect is removed at rate $b$. Let $\xi_t\subset\Z$ be the infected set of the one-dimensional nearest-neighbor contact process with infection rate $a$ per infected neighbor and recovery rate $b$. The rate inequalities give a monotone coupling under which $C_t\subseteq\xi_t$ whenever the two processes start from the same set.

Let $\lambda_c^{\rm 1D}$ denote the critical infection rate of the one-dimensional contact process with unit recovery rate and infection rate $\lambda$ per infected neighbor. If $a/b<\lambda_c^{\rm 1D}$, the contact process $\xi_t$ is strictly subcritical. Theorem~1 and Eq.~(1.7) of Ref.~\cite{Swart2016} then give $\EE_{\{0\}}|\xi_t|\le Ke^{-ct}$ for constants $K\ge1$ and $c>0$ depending only on $a,b$. Consequently, the single-defect kernel of the auxiliary line process satisfies
\begin{align}
g_t(j)&:=\PP_{\{0\}}(j\in C_t),\\
m(t)&:=\sum_{j\in\Z}g_t(j)=\EE_{\{0\}}|C_t|\le\EE_{\{0\}}|\xi_t|\le Ke^{-ct},
\label{eq:sm-chi-kernel-mass}
\end{align}
where the subscript $\{0\}$ indicates that the process starts from a single defect at the origin.

To show that the clean region advances, we first need to control the defect probability in the layer adjacent to the clean half-space. Every line in this exposed layer is initially fully defective, so we need an estimate for the auxiliary process $\chi$ started from $C_0=\Z$. The single-defect estimate above remains useful because a target site can be defective only if it is reached by the descendants of at least one initial defect.

To make this argument precise, first start with defects only in $C_N=\{-N,\ldots,N\}$. Since both updates of $\chi$ preserve unions, a union bound over these initial defects, together with translational invariance, gives:
\begin{equation}
\PP_{C_N}\bigl(\chi_t(0)={\tt 1}\bigr)\le\sum_{j=-N}^{N}\PP_{\{j\}}\bigl(\chi_t(0)={\tt 1}\bigr)=\sum_{j=-N}^{N}g_t(-j).
\end{equation}
Using the same update marks for every $N$, these processes increase to the process with every site initially defective as $N\to\infty$. Hence, we obtain:
\begin{equation}
\PP_{\Z}\bigl(\chi_t(0)={\tt 1}\bigr)\le\sum_{j\in\Z}g_t(-j)=m(t)\le Ke^{-ct}.
\label{eq:sm-chi-full-decay}
\end{equation}

Next, let us fix $\sigma\in\{\pm1\}^2$ and $h\in\Z$. Let $\hat\xi$ be the infinite-volume 3D CG process initialized as
\begin{equation}
\hat\xi_0(x)=\begin{cases}
{\tt 1} &\text{if }\ell_\sigma(x)<h,\\
{\tt 0} &\text{otherwise}.
\end{cases}
\end{equation}
Using the same initial configuration and the same rate-$a$ and rate-$b$ update marks, define $\widetilde\xi$ by
\begin{align}
\widetilde{\sf A}_x(\eta)&=\eta(x-e_1)\lor\eta(x)\lor\eta(x+e_1),\\
\widetilde{\sf B}_x(\eta)&=\eta(x+\sigma_2e_2)\lor\eta(x+\sigma_3e_3).
\end{align}
Monotonicity and the pointwise inequalities ${\sf A}^{(3)}_x\le\widetilde{\sf A}_x$ and ${\sf B}^{(3)}_x\le\widetilde{\sf B}_x$ give $\hat\xi_t\le\widetilde\xi_t$. The half-space $\{\ell_\sigma\ge h\}$ remains clean in both processes.

Label each layer by its depth $k=h-\ell_\sigma(x)$ and define $q_k(t)$ as follows:
\begin{equation}
q_k(t)=\PP\bigl(\widetilde\xi_t(x)={\tt 1}\bigr),\qquad \ell_\sigma(x)=h-k.
\end{equation}
Translational invariance parallel to the layers makes $q_k(t)$ independent of the choice of $x$ within a layer, and $q_k(t)=0$ for every $k\le0$.

Fix a line $\{x_0+je_1:j\in\Z\}$ at depth $k\ge1$. The $\widetilde{\sf A}$ updates on this line are the OR updates of $\chi$. For each $\widetilde{\sf B}$ update, it is useful to break it down into two steps: (i) a reset to ${\tt 0}$, followed by (ii) a defect insertion if either parent in layer $k-1$ is defective. Each insertion is counted as a new source, even if the site was already defective. Descendants of the initial defects and of each later insertion then evolve under the line OR and reset marks, with further insertions counted separately. Since both updates of $\chi$ preserve unions, the union of these contributions reproduces the line configuration exactly. For a fixed target and finite time, the backward graphical exploration encounters only finitely many possible sources almost surely.

The contribution from the initial defects is, therefore, governed by the autonomous process $\chi$ that started with every site being defective. By \eqnref{eq:sm-chi-full-decay}, its probability of occupying the target at time $t$ is at most $m(t)$. 
Consider a defect inserted at position $j$ at time $s$. Its descendants under $\chi$ are determined by the OR updates and resets during $(s,t]$, with later insertions counted separately. Conditional on the history through $s$, the Poisson clocks generating these updates evolve independently of that history. Thus, the probability that the inserted defect produces a descendant at position $0$ at time $t$ is $g_{t-s}(-j)$.

Consider adding a $\widetilde{\sf B}$ update at position $j$ and time $s$. The parent sites immediately before $s$ are unaffected by this added update. Since each parent site is defective with probability $q_{k-1}(s)$, a union bound gives a defect insertion probability of at most $2q_{k-1}(s)$. Conditional on such an insertion, the defect produces a descendant at the target site with probability $g_{t-s}(-j)$ under $\chi$. Its expected contribution is, therefore, at most $2q_{k-1}(s)g_{t-s}(-j)$.

The Campbell--Mecke formula (see, for example, Ref.~\cite{Last_Penrose}) averages these contributions over update times by integrating against $b\,ds$. Summing over all possible sites $j \in \Z$, using $\sum_{j\in\Z}g_{t-s}(-j)=m(t-s)$, and adding the contribution from initial defects gives
\begin{align}
q_k(t)
&\le m(t)+2b\int_0^t m(t-s)q_{k-1}(s)\,ds \nonumber\\
&\le Ke^{-ct}+2bK\int_0^t e^{-c(t-s)}q_{k-1}(s)\,ds.
\label{eq:sm-CG3-recursion}
\end{align}
Importantly, note that we did not use any independence between the parent sites or between different insertion events in the above calculation.

Starting from $q_0=0$, \eqnref{eq:sm-CG3-recursion} yields
\begin{equation}
q_k(t)\le K\exp\!\left(\theta k-\frac{ct}{2}\right),
\label{eq:sm-CG3-layer-bound}
\end{equation}
where $\theta=\log(1+4bK/c)$. To verify this bound, we proceed by induction over the layer depth. 
The bound holds for $k=0$ because $q_0=0$. Suppose it holds at depth $k-1$ for all times. Substituting $q_{k-1}(s)\le Ke^{\theta(k-1)-cs/2}$ into \eqnref{eq:sm-CG3-recursion} gives
\begin{align*}
q_k(t)
&\le Ke^{-ct}+2bK^2e^{\theta(k-1)}\int_0^t e^{-c(t-s)}e^{-cs/2}\,ds\\
&\le Ke^{-ct/2}\left[1+\frac{4bK}{c}e^{\theta(k-1)}\right]\\
&\le Ke^{-ct/2}e^{\theta(k-1)}\left(1+\frac{4bK}{c}\right)\\
&=K\exp\!\left(\theta k-\frac{ct}{2}\right).
\end{align*}
The first line uses the induction hypothesis, whose prefactor $K$ combines with the coefficient $2bK$ in the recursion. In the third line, we replace $1$ inside the bracket by its upper bound $e^{\theta(k-1)}$ and factor out this common term. Finally, we use $1+4bK/c=e^\theta$ by the definition of $\theta$, to obtain the last equality, completing the induction. Since $\hat\xi_t\le\widetilde\xi_t$, we also have $\PP(\hat\xi_t(x)={\tt 1})\le q_k(t)$ for every site at depth $k$.

\subsubsection{Block estimate}
Let us assume $a/b<\lambda_c^{\rm 1D}$, and let $K,c,\theta$ be the constants in \eqnref{eq:sm-CG3-layer-bound}. We choose the constant $\tau$ that appears in the definition of $T_H$ to be larger than $2(2\theta+1)/c$ and keep it fixed as $H\to\infty$. Let $\mathcal E_H$ be the event that a backward causal path from the core at time $T_H$ reaches the exterior of the buffer. For the fixed update rates $a,b$, with a total rate of $a+b$ at each site, Lemma~\ref{lem:speed} allows us to choose $R$ sufficiently large (but independent of $H$) after fixing $\tau$ so that
\begin{equation}
\PP(\mathcal E_H)\le C_{\rm ext}H^3e^{-c_{\rm ext}H}.
\label{eq:sm-CG3-exterior}
\end{equation}
The positive constants $C_{\rm ext}$ and $c_{\rm ext}$ may depend on $a,b,\tau,R$, but are independent of $H$.

For each advancing direction $\ell_\sigma$, choose
\begin{equation}
h=1+\max_{q\in Q_H(0)}\ell_\sigma(q),
\end{equation}
so that every core site is initially defective in the infinite-volume process $\hat\xi$. Writing $x=Hs+q$ with $q\in Q_H(0)$, we have $\ell_\sigma(x)\le h-1$ whenever $\ell_\sigma(s)\le0$. Hence, the clean half-space used to initialize $\hat\xi$ intersects the buffer only within the predecessor region:
\begin{equation}
\{\ell_\sigma\ge h\}\cap\mathcal B_H(0)\subseteq\bigcup_{s\in S_\sigma}Q_H(s).
\end{equation}
The finite-buffer initial condition used in the block process $\zeta^{\sigma}_t$ therefore satisfies $\zeta^{\sigma}_0\le\hat\xi_0$ inside the buffer. We now couple the two processes using the same Poisson clocks there. On $\mathcal E_H^c$, the core configuration at time $T_H$ depends only on the initial sites and updates inside the buffer, so monotonicity gives
\begin{equation}
\zeta^{\sigma}_{T_H}(x)\le\hat\xi_{T_H}(x),
\end{equation}
for every $x\in Q_H(0)$.

Next, note that the maximum and minimum value of $\ell_\sigma$ on $Q_H(0)$ differ by $2(H-1)$, implying that every core site lies at a depth of $1\le k\le2H-1$. Applying \eqnref{eq:sm-CG3-layer-bound} and taking a union bound over the $H^3$ core sites yields the following bound:
\begin{align}
r_{H,\sigma}
&\le KH^3\exp\!\left(2\theta H-\frac{c\tau H}{2}\right)+\PP(\mathcal E_H) \nonumber\\
&\le KH^3e^{-H}+\PP(\mathcal E_H),
\end{align}
where the last inequality follows from the choice of $\tau$.

For the two static directions $i=1,2$, the clean half-spaces $\{x_1\ge0\}$ and $\{x_1\le H-1\}$ contain the core and remain invariant under the noiseless dynamics for any $a,b>0$. Within the buffer, these half-spaces agree with the initially clean regions of the corresponding block tests. Thus, on $\mathcal E_H^c$, the core is clean at time $T_H$, giving $r_{H,i}\le\PP(\mathcal E_H)$. Combining these estimates with \eqnref{eq:sm-CG3-exterior}, we obtain
\begin{equation}
\max_i r_{H,i}\le CH^3e^{-\alpha H},
\end{equation}
where $C=K+C_{\rm ext}$ and $\alpha=\min\{1,c_{\rm ext}\}>0$ are independent of $H$.

Since $\tau$ and $R$ are fixed independently of $H$, all block-failure probabilities vanish as $H\to\infty$, while the buffer contains a fixed number of coarse cells. Together with the coarse-grained geometry established above, these estimates satisfy the requirements of Theorem~\ref{thm:maintext} and establish the stability of the $\bar{\tt 0}$ whenever $a/b<\lambda_c^{\rm 1D}$.

A bit flip combined with the coordinate exchange $e_1\leftrightarrow e_3$ exchanges the two update rules and their rates, mapping the process with rates $(a,b)$ to the process with rates $(b,a)$. Applying the preceding argument to the transformed process establishes stability of $\bar{\tt 1}$ in the original process whenever $b/a<\lambda_c^{\rm 1D}$. Consequently, both homogeneous configurations are stable under sufficiently weak noise whenever
\begin{equation}
(\lambda_c^{\rm 1D})^{-1}<\frac{a}{b}<\lambda_c^{\rm 1D}.
\end{equation}

\section{Non-examples}
\label{sec:app_D}
Having discussed several examples in which Theorem~\ref{thm:maintext} proves stability, we now discuss some important cases that are \emph{not} accounted for by the stability criterion. 

\begin{figure}[htbp]
    \centering
    \includegraphics[width=0.8\linewidth]{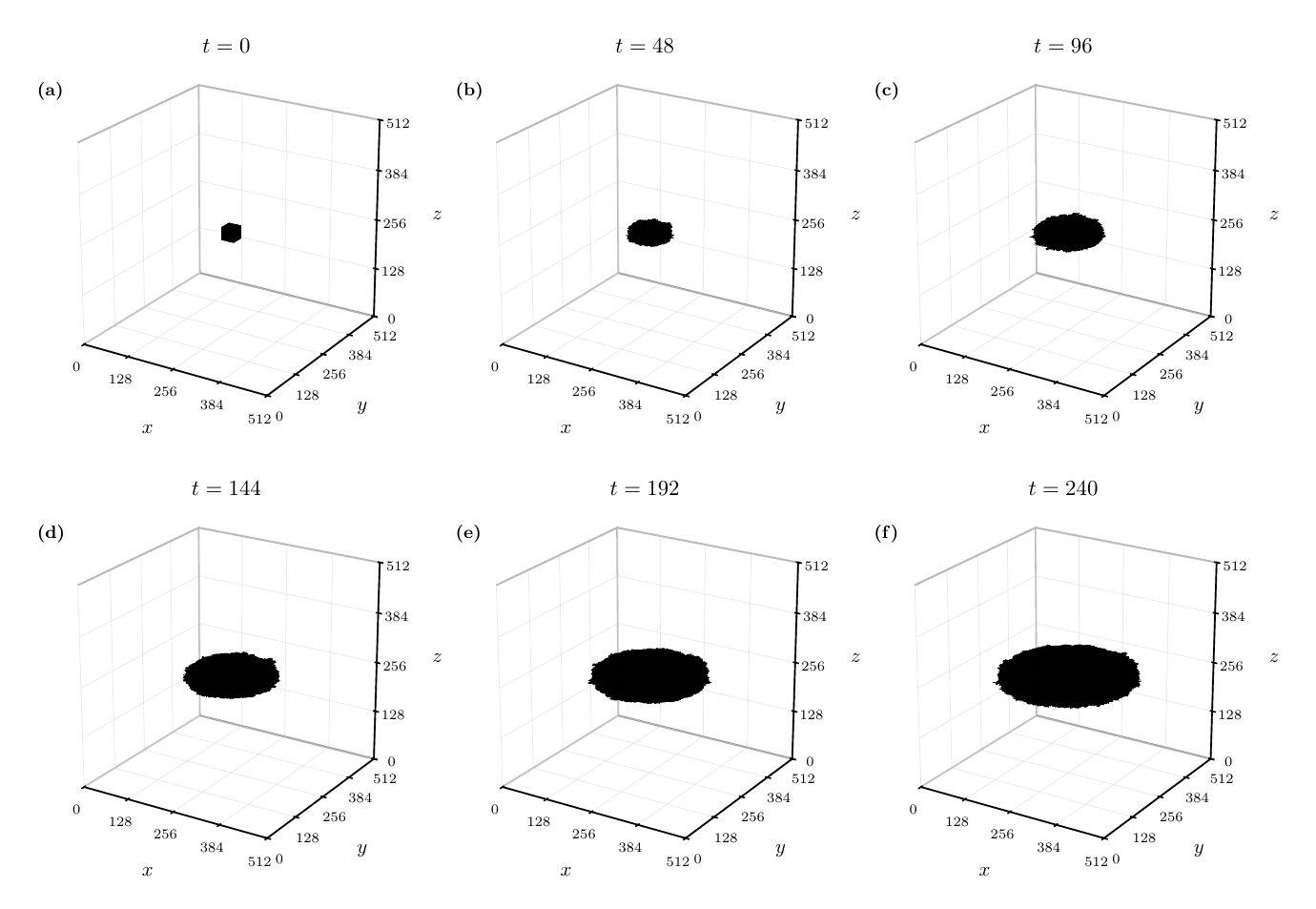}
    \caption{\textbf{Noiseless dynamics under the 3D {\sf R} process.} Snapshots of the asynchronous dynamics of a cubic system of size $ L = 512$ starting from a $32\times 32 \times 32$ cubic defect set. The instability is manifest from the expanding size of the defect set.}
    \label{fig:3D_R_instability}
\end{figure}

\subsection{Continuous-time {\sf T} process}
\label{sec:sm-T-proof}
In Ref.~\cite{LakeRo2025}, one of the squeezing codes introduced was the {\sf T} automaton, which was numerically shown to be robust under synchronous updates but not under asynchronous updates. Below, we show that the stability criterion does not apply to the {\sf T} automaton. We emphasize that the failure of Theorem~\ref{thm:maintext} to apply to a given continuous-time process does not prove the instability of the process, because Theorem~\ref{thm:maintext} only provides a sufficient (not necessary) criterion for stability. Therefore, the following discussion should be treated as a sanity check, as opposed to a proof of instability of the {\sf T} automaton. 

The {\sf T} automaton, like the other squeezing codes, involves an OR update and an AND update. The neighborhoods for both updates are identical and given by: $\cN_{\lor} = \cN_{\land} = \{0,e_1, -e_2\}. $

For the OR update to output {\tt 0}, every site in ${\cN}_{\lor}$ must be {\tt 0}, implying that there is only a single zero set for the OR update: $\cZ^{\rm min}_{\lor} = \{\{0, e_1, -e_2\}\}$. The AND update outputs {\tt 0} if at least one site in $\cN_{\land}$ is {\tt 0}, implying that the set of minimal zero sets is given by $\cZ^{\rm min}_{\land} = \{\{0\}, \{e_1\}, \{-e_2\}\}$. Let $\ell$ be a stable direction, i.e., $\cN^+_{a}(\ell)$ is a zero set for each update type $a$. The OR update places the following constraints on $\ell$: (i) $\ell(e_1) \geq 0$ and (ii) $\ell(e_2) \leq 0$, while the AND update does not place any constraints since $\{0\}$ is a zero-set. Therefore, the set of stable directions, $\cS_{\sf T}$, of the {\sf T} automaton is given by:
\begin{equation}
    \cS_{\sf T} = \{\ell\neq 0^{\ast}: \ell(e_1) \geq 0 \text{ and }\ell(e_2) \leq 0\}. 
\end{equation}
The stable directions correspond to rays pointing into the fourth quadrant (south-east). 

Let us assume that the stable directions $\ell_1, \ell_2, \cdots , \ell_m \in \cS_{\sf T}$, together with positive weights $\alpha_1, \alpha_2, \cdots, \alpha_m$, form a balanced family, i.e., $\sum_{i = 1}^{m} \alpha_i \ell_i = 0^{\ast}$. On applying both sides to the vector $e_2$, we obtain $\sum_{i} \alpha_i\ell_i(e_2) = 0.$ Since each $\ell_i$ is a stable direction, we know that $\ell_i(e_2) \leq 0$ for every $i\in\{1,2,\cdots, m\}$. As the weights $\{\alpha_i\}$ are all strictly positive, this implies that $\ell_i(e_2) = 0$ for every $i \in \{1,2,\cdots,m\}$. We now repeat the same steps with $e_2$ replaced by $e_1$. Since every $\ell_i$ in the balanced family satisfies $\ell_i(e_1) \geq 0$ and the weights are all positive, we conclude that $\ell_i(e_1) = 0$ for every $i \in \{1,2,\cdots, m\}$. The vectors $e_1$ and $e_2$ span $\RR^2$, and thus, if a linear form $\ell$ satisfies $\ell(e_1) = \ell(e_2) = 0$, then $\ell$ must be $0^{\ast}$. This means that each $\ell_i$ in the balanced family must be $0^{\ast}$, which is a contradiction since $0^{\ast}$ cannot be a part of $\cS_{\sf T}$. Therefore, no balanced family exists for the ${\sf T}$ automaton.
Hence, Theorem~\ref{thm:maintext} does not apply to the continuous-time {\sf T} process.

\subsection{Continuous-time {\sf R} process in 3D}
\label{sec:sm-3dR-proof}
We next consider the three-dimensional generalization of the {\sf R} process introduced in Ref.~\cite{LakeRo2025}, which was found numerically not to be bistable under asynchronous updates. Although it satisfies the geometric requirements of the stability criterion, the contact-process argument that produced the block estimates in two dimensions fails to work in three dimensions and reveals the cause of the lack of bistability. Our instability argument requires the condition $\lambda_c^{\rm 2D}<1/2$, where $\lambda_c^{\rm 2D}$ is the critical birth rate per neighbor of the two-dimensional contact process. This condition is supported by the numerical estimate $\lambda_c^{\rm 2D}\approx0.41$~\cite{2D_CP_Vojta}, which lies appreciably below $1/2$.

The 3D {\sf R} process on $\Z^3$ is characterized by OR updates along $e_1$ and $e_2$, and an AND update along $e_3$, each occurring at a rate of $1$. The neighborhood of each update is $\{0,\pm e_i\}$ for $i\in\{1,2,3\}$. We know that an OR update has a single zero set, equal to the full neighborhood of that update type. Repeating the same steps as in the discussion of the 2D {\sf R} process, we find a balanced family of stable directions
\begin{equation}
    \ell_{1,2}=\pm e_3^{\ast},
    \quad
    \alpha_1=\alpha_2=1.
\end{equation}
Both directions permit positive displacements in the proposed coarse rule, just as for the 2D {\sf R} process.
Thus, at the level of the stable-direction geometry, there is little difference between the two- and three-dimensional processes. 
The obstruction instead appears when we attempt to establish the required block estimates.

The erosion mechanism of the 2D {\sf R} process proceeds by successively clearing exposed rows of a defect droplet. This shrinks one dimension of the droplet until it disappears. As we saw above, this mechanism depends crucially on the dynamics of an exposed row: at equal OR and AND rates, the exposed-row process is stochastically dominated by a subcritical
one-dimensional contact process. In three dimensions, an exposed layer is two-dimensional, and the corresponding comparison actually goes in the opposite direction, as we show below.

Consider a finite cubical defect droplet in a background of ${\tt 0}$s, and let $\Pi$ be the horizontal plane containing its upper face (normal to $e_3$). The plane immediately above $\Pi$ remains clean under the noiseless dynamics. Consequently, every AND update on $\Pi$ resets its site to ${\tt 0}$, giving a defect-death rate of one. Meanwhile, the OR updates allow defects to spread along the two horizontal directions. Let $n_i\in{0,1,2}$ denote the number of defective neighbors along $\pm e_i$, for $i=1,2$. Since each OR update occurs at a unit rate and creates a defect whenever at least one neighbor along its corresponding axis is defective, the total birth rate at a vacant site on $\Pi$ is given by:
\begin{equation}
    b_{\sf R}(n_1,n_2) = \mathbf{1}_{\{n_1\geq1\}} + \mathbf{1}_{\{n_2\geq1\}}.
\end{equation}
For every $n_1,n_2\in\{0,1,2\}$, we have $b_{\sf R}(n_1,n_2) \geq \frac{1}{2}(n_1+n_2)$.
The right-hand side is precisely the birth rate of the standard contact process on $\Z^2$ with a death rate of $1$ and a per-neighbor birth rate of $\lambda=1/2$. 
Thus, under evolution with shared update marks, the exposed-plane ($\Pi$) dynamics of the 3D {\sf R} process stochastically dominate this two-dimensional contact process.

As noted above, the critical per-neighbor birth rate of the 2D contact process is numerically estimated to be $\lambda_c^{\rm 2D}\approx 0.41$~\cite{2D_CP_Vojta}. Under the numerically supported condition $\lambda_c^{\rm 2D}<1/2$, the comparison contact process is supercritical. Starting from any finite nonempty defect set, it therefore survives indefinitely and spreads to arbitrarily large distances with positive probability. By stochastic domination, the same holds for the exposed-plane dynamics of the noiseless 3D {\sf R} process.

This explains why the two-dimensional block argument does not extend to three dimensions. In 2D, exposure to a clean neighboring row produces a subcritical process whose defect density decays exponentially over time. In 3D, exposure to a clean neighboring plane instead produces a two-dimensional process capable of sustaining defects indefinitely. Thus, the exponential decay estimate underlying the 2D {\sf R} block construction no longer holds.

Figure~\ref{fig:3D_R_instability} illustrates the spreading of defects under noiseless asynchronous 3D {\sf R} dynamics. Starting from a finite cubical defect droplet, the defect set spreads over increasingly large spatial scales rather than being squeezed to extinction, in contrast to the 2D {\sf R} process.

We can, in fact, strengthen the above contact process comparison to establish weak-noise instability, provided the numerically supported inequality $\lambda^{\rm 2D}_c<1/2$ holds.
Let $\theta>0$ denote the survival probability of the standard 2D contact process with a death rate of $1$ at every site, and a birth rate of $1/2$ per neighbor, starting from a single defect. 
Let the one-sided noise rate be $\varepsilon>0$. As we have already established, on any fixed horizontal plane, the resulting dynamics stochastically dominate a 2D contact process with the same birth and death rates, along with spontaneous births at a rate of $\varepsilon$.

To establish the weak-noise instability of the 3D {\sf R} process under the aforementioned assumption about $\lambda^{\rm 2D}_c$, we consider all sites from which a defect could reach the origin (an arbitrary choice, owing to translational invariance) at time $t$ through the updates of the comparison process.

Because defect transmission is symmetric between neighboring sites, tracing the possible sources backward in time has the same probability law as a contact process started from a single defect. Thus, with a probability of at least $\theta$, the set of possible sources remains nonempty throughout $[0,t]$~\cite{Grimmett2018}. On this event, at every earlier time, there is at least one site where a noise-created defect would eventually reach the origin at time $t$. Since hostile-noise marks occur independently at a rate of $\varepsilon$ per site, the probability that at least one such defect is created is at least $1-e^{-\varepsilon t}$. By stochastic domination, we therefore obtain
\begin{equation}
\mathbb P_{\bar{\tt 0}}^{\varepsilon}
\bigl(\eta_t(0)={\tt 1}\bigr)
\geq \theta(1-e^{-\varepsilon t}).
\label{eq:3D-R-defect}
\end{equation}
Finally, taking the supremum over time and then taking the limit $\varepsilon\to0^+$ gives
\begin{equation}
\liminf_{\varepsilon\to0^+}
\sup_{t\geq0}
\mathbb P_{\bar{\tt 0}}^{\varepsilon}
\bigl(\eta_t(0)={\tt 1}\bigr)
\geq \theta>0,
\end{equation}
establishing weak-noise instability of the $\bar{\tt 0}$ configuration of the 3D {\sf R} process, provided that $\lambda^{\rm 2D}_c<1/2$.

\section{Numerical Details}
\label{sec:app_numerics}

\begin{figure}
    \centering
    \includegraphics[width=\linewidth]{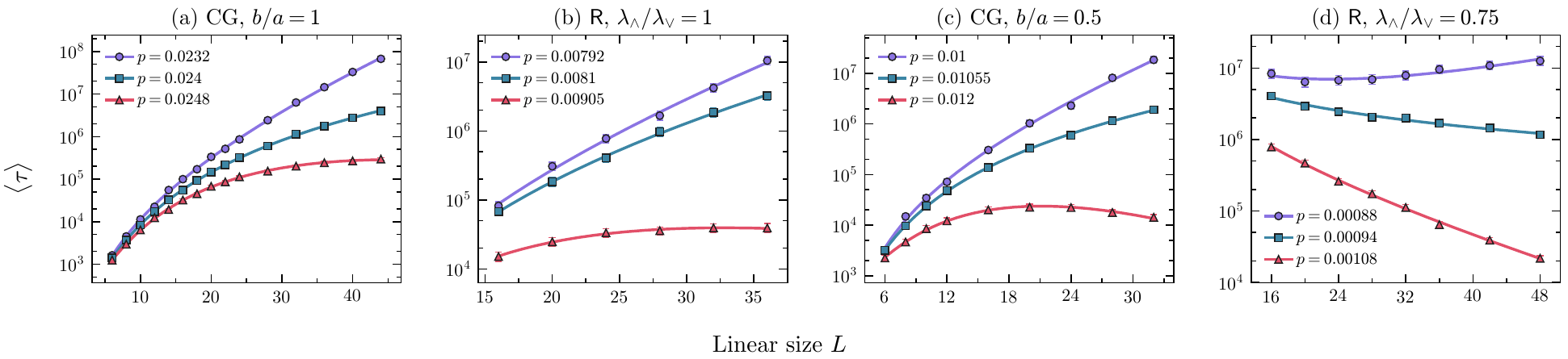}
    \caption{\textbf{Mean memory lifetime $\langle\tau\rangle$ versus linear system size $L$ for the asynchronous 2DCG and {\sf R} processes initialized in the $\bar{\tt 0}$ configuration under maximally adversarial noise.} Panels (a) and (c) show the memory lifetime fitting for the 2D CG process at $b/a=1$ and $0.5$, respectively. Panels (b) and (d) show the memory lifetime fitting for the 2D {\sf R} process at $\lambda_{\land}/\lambda_{\lor}=1$ and $0.75$, respectively. Each point shows the mean lifetime (in sweeps) from independent runs at the corresponding $L$ and $p$. Error bars show $95\%$ bootstrap confidence intervals for these means. To detect a logical flip, we average the magnetization over a sliding window of $w=L$ temporal blocks, each block $L/2$ sweeps long. The lifetime $\tau$ is the start time of the first such window with a nonnegative average. The solid curves come from weighted fits of $\log\langle\tau\rangle$ to $A(p)\log L+B(p)L+C(p)+D(p)/L$. Each noise probability $p$ has its own coefficients $A$, $B$, and $C$.}
    \label{fig:lifetime-scaling}
\end{figure}

Here, we detail the numerical calculations used to estimate the phase boundaries of the 2D CG and {\sf R} processes shown in Fig.~\ref{fig:phase-diagrams}. We simulated these asynchronous processes using event-driven Monte Carlo on periodic $L\times L$ square lattices. All runs were initialized in the $\bar{\tt 0}$ state under maximal bias. During the simulation, random biased noise resets a site to {\tt 1} with probability $p$, while the standard update rule is applied with probability $1-p$.

Following Ref.~\cite{Ray_2024}, we sampled the magnetization once per sweep, grouped these samples into blocks of $L/2$ sweeps, and computed a moving average over a window of $w=L$ blocks. The flip time is defined as the start of the earliest window that yields a nonnegative average magnetization. We estimated the mean lifetime $\tau(L,p)$ by averaging these flip times across independent runs, restricting our fits to parameter sets where every run successfully recorded a flip.

The lattice sizes and number of runs varied by process and rate ratio. For the CG process at unequal rates, we performed $192$ runs per $(L,p)$ pair for system sizes ranging from $6\le L\le32$. At equal rates, we simulated $15$ different sizes up to $L=44$, using $1000$ independent runs per $(L,p)$ (reduced to $250$ at the lowest values of $p$ simulated). For the {\sf R} process, we conducted $192$ independent runs per $(L,p)$ over sizes $16\le L\le36$. Near the phase boundary at $\lambda_{\land}/\lambda_{\lor}=0.75$, we extended the system size up to $L=48$, utilizing between $192$ and $384$ runs per $(L,p)$.

For each rate ratio, we modeled the size dependence of the mean memory lifetime using the ansatz (see Fig.~\ref{fig:lifetime-scaling}):
\begin{equation}
    \log \tau(L,p)=C(p)+A(p)\log L+B(p)L+\frac{D(p)}{L}.
\end{equation}
The parameters $A$, $B$, and $C$ were allowed to vary freely with $p$. For the 2D CG process, $D(p)$ was treated as either zero, constant, or linear, whereas it was strictly set to zero for the {\sf R} process. Because a positive $B(p)$ signifies an exponential growth of the lifetime with system size $L$, we estimated the critical noise threshold $p_c$ by linearly interpolating the zero-crossing of $B(p)$. Finally, the diamond marker at $p_c=0$ for the {\sf R} process in Fig.~\ref{fig:phase-diagrams} indicates an empirical cutoff at $\lambda_{\land}/\lambda_{\lor}=0.65$, where the mean lifetime was observed to decrease as $L$ increased from $52$ to $56$ in the low-noise regime.

\end{document}